%% file: LogRatioChannels.tex
\def\IsSubmission{} 

\ifdefined \IsSubmission \else
\def\IsDraft{} 
\fi

\ifdefined\IsLLNCS
\documentclass{llncs}
\usepackage{amsfonts,amsmath,amssymb,boxedminipage,color,url,nccmath}
\else
\documentclass[11pt]{article}
\usepackage{amsfonts,amsmath,amssymb,boxedminipage,color,url,fullpage,nccmath,amsthm}
\fi

\usepackage{color}

\definecolor{weborange}{rgb}{.8,.3,.3}
\definecolor{webblue}{rgb}{0,0,.8}
\definecolor{internallinkcolor}{rgb}{0,.5,0}
\definecolor{externallinkcolor}{rgb}{0,0,.5}

\usepackage[pagebackref,
hyperfootnotes=false,
colorlinks=true,
urlcolor=externallinkcolor,
linkcolor=internallinkcolor,
filecolor=externallinkcolor,
citecolor=internallinkcolor,
breaklinks=true,
pdfstartview=FitH,
pdfpagelayout=OneColumn]{hyperref}

\usepackage{enumerate,paralist}
\usepackage[numbers]{natbib} 
\usepackage[labelfont=bf]{caption}
\usepackage{aliascnt}
\usepackage{cleveref}
\usepackage{xspace}

\usepackage{relsize}
\usepackage{graphics}
\ifdefined\IsLLNCS
\fi

\input{macros}

\newcommand{\Jnote}[1]{\authnote{Jad}{#1}}
\newcommand{\Nnote}[1]{\authnote{Noam}{#1}}
\newcommand{\Rnote}[1]{\authnote{Ronen}{#1}}
\newcommand{\Inote}[1]{\authnote{Iftach}{#1}}

\newcommand{\DP}{\myOptName{DP}}
\newcommand{\WBSC}{\myOptName{WBSC}}
\newcommand{\SWBSC}{\myOptName{SWBSC}}
\newcommand{\OT}{{\myOptName{OT}}}
\newcommand{\CDP}{\myOptName{IND-DP}}
\newcommand{\Leak}{\myOptName{leakage}}
\newcommand{\CLeak}{\myOptName{comp-leakage}}
\newcommand{\NCLeak}{\myOptName{nu-comp-leakage}}
\newcommand{\NCDP}{\myOptName{SIM-DP}}
\newcommand{\CHAN}{\myOptName{CHN}}

\newcommand{\Agr}{\myOptName{agreement}}

\newcommand{\Fcalls}{\operatorname{final}}

\newcommand{\VA}{V^\Ac}
\newcommand{\VB}{V^\Bc}
\newcommand{\VP}{V^\p}

\newcommand{\UA}{U^\Ac}

\newcommand{\OA}{O^{\Ac}}
\newcommand{\OB}{O^{\Bc}}
\newcommand{\OP}{O^{\p}}

\newcommand{\tC}{\widetilde{C}}

\newcommand{\tAc}{\widetilde{\Ac}}
\newcommand{\tBc}{\widetilde{\Bc}}
\newcommand{\tPc}{\widetilde{\Pc}}

\newcommand{\tpi}{\widetilde{\pi}}
\newcommand{\tPi}{\tpi}

\newcommand{\tVA}{V^{\tAc}}
\newcommand{\tVB}{V^{\tBc}}
\newcommand{\tVP}{V^{\tPc}}

\newcommand{\tOA}{O^{\tAc}}
\newcommand{\tOB}{O^{\tBc}}
\newcommand{\tOP}{O^{\tPc}}

\newcommand{\hAc}{\widehat{\Ac}}
\newcommand{\hBc}{\widehat{\Bc}}
\newcommand{\hPc}{\widehat{\Pc}}
\newcommand{\hC}{\widehat{C}}

\newcommand{\hOA}{O^{\hAc}}
\newcommand{\hOB}{O^{\hBc}}

\newcommand{\hVA}{V^{\hAc}}
\newcommand{\hVB}{V^{\hBc}}
\newcommand{\hVP}{V^{\hPc}}

\newcommand{\ooB}{\obar^{\Bc}}
\newcommand{\ooA}{\obar^{\Ac}}
\newcommand{\oOA}{ \Obar^\Ac }
\newcommand{\oOB}{\Obar^\Bc }
\newcommand{\oOP}{\Obar^\Pc }

\newcommand{\oVA}{ \Vbar^\Ac }
\newcommand{\oVB}{\Vbar^\Bc }
\newcommand{\oVP}{\Vbar^\Pc }

\newcommand{\p}{\Pc}

\newcommand{\Obar}{\overline{O}}
\newcommand{\Vbar}{\overline{V}}

\newcommand{\obar}{\overline{o}}

\renewcommand{\pi}{\Pi}

\title{ Channels of Small Log-Ratio  Leakage  and \\
	 Characterization of Two-Party  Differentially Private  Computation
	\Draft{\\\small \sc Working Draft: Please Do Not Distribute}\thanks{An extended abstract of this work appeared  in TCC 2019 \cite{haitner2019channels}}
}

 \ifdefined\IsAnon
 	\author{}
	 \date{}
 	\LLNCS{ \institute{}}
 \else
 	\ifdefined\IsLLNCS
 \else
    \author{Iftach Haitner\thanks{School of Computer Science, Tel Aviv University. Emails: \texttt{iftachh@cs.tau.ac.il}. Research supported by ERC starting grant 638121. Member of the  Check Point Institute for Information Security}
    	\and Noam Mazor\thanks{School of Computer Science, Tel Aviv University. Emails: \texttt{noammaz@gmail.com}. Research supported by ERC starting grant 638121.}
    	\and Ronen Shaltiel\thanks{Department of Computer Science. University of Haifa, Email: \texttt{ronen@cs.haifa.ac.il}. Research supported by ISF grant 1628/17.}
    	\and Jad Silbak\footnotemark\thanks{School of Computer Science, Tel Aviv University. Emails: \texttt{jadsilbak@mail.tau.ac.il}. Research supported by ERC starting grant 638121 and by ISF grant 1628/17. Some of this research was done while Jad Silbak was a student at the University of Haifa.}
    }
 \fi
 \fi

\begin{document}

\maketitle

\input{Abstract}


	\thispagestyle{empty}
	\pagenumbering{gobble}
	\clearpage
	\pagenumbering{arabic}

\Tableofcontents

\input{Introduction}

\input{Technique}

\input{Preliminaries}

\input{OtAmp}

\input{OTAmpTree}
\input{LargeGapToOT}

\input{LowerBound}
\input{Comp}
\input{DPXOR}

\input{Comp-DPXOR}

\input{NonMonotone}

\section*{Conclusion and Open Problems}
A natural open problem is to characterize the (Boolean) AND differentially private functionality. That is, show a similar dichotomy that characterizes which accuracy and leakage require OT.

More generally, the task of understanding and characterizing other (non Boolean) differentially private  functionalities like hamming distance and inner product remains open.

\section*{Acknowledgement}

We are very grateful to Kobbi Nissim, Eran Omri and Ido Abulafya for helpful conversations and advice.
We thank the anonymous referees for detailed and very helpful comments.

\bibliographystyle{abbrvnat}
\bibliography{crypto}

\appendix
\input{Appendix}
\end{document}

%% file: macros.tex
\providecommand{\remove}[1]{}

\newcommand{\Draft}[1]{\ifdefined\IsDraft \texttt{ #1} \fi}
\newcommand{\LLNCS}[1]{\ifdefined\IsLLNCS #1 \fi}

\newcommand{\TLLNCS}[2]{\ifdefined\IsLLNCS#1\else #2 \fi}

\ifdefined\IsDraft
\newcommand{\authnote}[2]{{\bf [{\color{red} #1's Note:} {\color{blue} #2}]}}
\else
\newcommand{\authnote}[2]{}
\fi

\newcommand{\sdotfill}{\textcolor[rgb]{0.8,0.8,0.8}{\dotfill}} 

\newenvironment{protocol}{\begin{proto}}{\vspace{-\topsep}\sdotfill\end{proto}}
\newenvironment{algorithm}{\begin{algo}}{\vspace{-\topsep}\sdotfill\end{algo}}

\newcommand{\aka} {also known as,\ }
\newcommand{\resp}{resp.,\ }
\newcommand{\ie} {i.e.,\ }
\newcommand{\eg} {e.g.,\ }
\newcommand{\wrt} {with respect to\ }

\newcommand{\cf}{{cf.,\ }}

\newcommand{\set}[1]{\ens{#1}}

\newcommand{\floor}[1]{\left \lfloor#1 \right \rfloor}

\newcommand{\half}{\tfrac{1}{2}}

\newcommand{\eqdef}{:=}

\newcommand{\R}{{\mathbb R}}
\newcommand{\N}{{\mathbb{N}}}

\newcommand{\zo}{\set{0,1}}
\newcommand{\zn}{\zo^n}
\newcommand{\zs}{\zo^\ast}

\newcommand{\xor}{\oplus}

\newcommand{\eps}{\epsilon}

\newcommand{\from}{\leftarrow}
\newcommand{\la}{\gets}

\newcommand{\poly}{\operatorname{poly}}

\newcommand{\Exp}{\operatorname*{E}}

\newcommand{\negl}{\operatorname{neg}}

\newcommand{\Supp}{\operatorname{Supp}}

\newcommand{\MathAlg}[1]{\mathsf{#1}\xspace}

\newcommand{\class}[1]{\mathrm{#1}}

\newcommand{\IP}{\class{IP}}

\newcommand{\IA}{\class{IA}}

\renewcommand{\P}{\class{P}}

\renewcommand{\cref}{\Cref}

\TLLNCS{
	\newaliascnt{claiml}{theorem}
	\newtheorem{claiml}[claiml]{Claim}
	\aliascntresetthe{claiml}
	
	\renewenvironment{claim}{\begin{claiml}}{\end{claiml}}
	
	\newaliascnt{lemmal}{theorem}
	\newtheorem{lemmal}[lemmal]{Lemma}
	\aliascntresetthe{lemmal}
	\renewenvironment{lemma}{\begin{lemmal}}{\end{lemmal}}
	
	\newaliascnt{propositionl}{theorem}
	\newtheorem{propositionl}[propositionl]{Proposition}
	\aliascntresetthe{propositionl}
	\renewenvironment{proposition}{\begin{propositionl}}{\end{propositionl}}

	\newaliascnt{definitionl}{theorem}
	\newtheorem{definitionl}[definitionl]{Definition}
	\aliascntresetthe{definitionl}
	\renewenvironment{definition}{\begin{definitionl}}{\end{definitionl}}
	
	\newaliascnt{corollaryl}{theorem}
	\newtheorem{corollaryl}[definitionl]{Corollary}
	\aliascntresetthe{corollaryl}
	\renewenvironment{corollary}{\begin{corollaryl}}{\end{corollaryl}}

}
{
	\newtheorem{theorem}{Theorem}[section]

	\newaliascnt{lemma}{theorem}
	\newtheorem{lemma}[lemma]{Lemma}
	\aliascntresetthe{lemma}
	
	\newaliascnt{claim}{theorem}
	\newtheorem{claim}[claim]{Claim}
	\aliascntresetthe{claim}

	\newaliascnt{corollary}{theorem}
	\newtheorem{corollary}[corollary]{Corollary}
	\aliascntresetthe{corollary}
	
	\newaliascnt{proposition}{theorem}
	\newtheorem{proposition}[proposition]{Proposition}
	\aliascntresetthe{proposition}

	\newaliascnt{conjecture}{theorem}
	
	\aliascntresetthe{conjecture}

	\newaliascnt{definition}{theorem}
	\newtheorem{definition}[definition]{Definition}
	\aliascntresetthe{definition}

	\newaliascnt{remark}{theorem}
	
	\aliascntresetthe{remark}

	\newaliascnt{example}{theorem}
	
	\aliascntresetthe{example}
}

\crefname{lemma}{Lemma}{Lemmas}
\crefname{figure}{Figure}{Figures}
\crefname{claim}{Claim}{Claims}
\crefname{corollary}{Corollary}{Corollaries}
\crefname{proposition}{Proposition}{Propositions}
\crefname{conjecture}{Conjecture}{Conjectures}
\crefname{definition}{Definition}{Definitions}
\crefname{remark}{Remark}{Remarks}
\crefname{exmaple}{Example}{Examples}

\newaliascnt{construction}{theorem}

\aliascntresetthe{construction}
\crefname{construction}{Construction}{Constructions}

\newaliascnt{fact}{theorem}
\newtheorem{fact}[fact]{Fact}
\aliascntresetthe{fact}
\crefname{fact}{Fact}{Facts}

\newaliascnt{notation}{theorem}
\newtheorem{notation}[notation]{Notation}
\aliascntresetthe{notation}
\crefname{notation}{Notation}{Notation}

\crefname{equation}{Equation}{Equations}

\newaliascnt{proto}{theorem}

\newtheorem{proto}[proto]{Protocol}

\aliascntresetthe{proto}
\crefname{proto}{protocol}{protocols}

\newaliascnt{algo}{theorem}
\newtheorem{algo}[algo]{Algorithm}
\aliascntresetthe{algo}
\crefname{algo}{algorithm}{algorithms}

\newaliascnt{expr}{theorem}
\newtheorem{expr}[expr]{Experiment}
\aliascntresetthe{expr}
\crefname{experiment}{experiment}{experiments}

\def\FullBox{$\Box$}
\def\qed{\ifmmode\qquad\FullBox\else{\unskip\nobreak\hfil
\penalty50\hskip1em\null\nobreak\hfil\FullBox
\parfillskip=0pt\finalhyphendemerits=0\endgraf}\fi}

\def\qedsketch{\ifmmode\Box\else{\unskip\nobreak\hfil
\penalty50\hskip1em\null\nobreak\hfil$\Box$
\parfillskip=0pt\finalhyphendemerits=0\endgraf}\fi}

\newcommand{\Tau}{\mathrm{T}} 

\newcommand{\ex}[2]{\Exp_{#1}\left[#2\right]}

\newcommand{\pr}[1]{\Pr\left[#1\right]}
\newcommand{\ppr}[2]{\Pr_{#1}\left[#2\right]}

\newcommand{\ens}[1]{\left\{#1\right\}}
\newcommand{\size}[1]{\left|#1\right|}

\newcommand{\cindist}{\mathbin{\stackrel{\rm nuC}{\approx}}}
\newcommand{\ucindist}{\mathbin{\stackrel{\rm C}{\approx}}}

\newcommand{\sindist}{\mathbin{\stackrel{\rm S}{\approx}}}
\newcommand{\rindist}{\mathbin{\stackrel{\rm R}{\approx}}}

\newcommand{\tP}{\widetilde{P}}
\newcommand{\tf}{\tilde{f}}
\newcommand{\out}{\operatorname{out}}
\newcommand{\Out}{\operatorname{Out}}

\newcommand{\Uni}{{\mathord{\mathcal{U}}}}

\newcommand{\prob}[1]{\mathsf{\textsc{#1}}}

\newcommand{\SD}{\prob{SD}}

\newcommand{\ppt}{{\sc ppt}\xspace}
\newcommand{\pptm}{{\sc pptm}\xspace}

\newcommand{\cs}{{\cal{S}}}

\newcommand{\cA}{\mathcal{A}}

\newcommand{\cX}{\mathcal{X}}
\newcommand{\cY}{\mathcal{Y}}

\newcommand{\Ac}{\MathAlgX{A}}
\newcommand{\Bc}{\MathAlgX{B}}
\newcommand{\ABc}{\set{\Ac,\Bc}}

\newcommand{\Mc}{\MathAlgX{M}}
\newcommand{\Dc}{\MathAlgX{D}}

\newcommand{\Pc}{\MathAlgX{P}}

\newcommand{\tX}{\tilde{X}}

\newcommand{\Tableofcontents}{
	\ifdefined\IsLLNCS \else
	\thispagestyle{empty}
	\pagenumbering{gobble}
	\clearpage
	\ifdefined\IsSubmission \else
	\setcounter{tocdepth}{2}
	\tableofcontents
	\thispagestyle{empty}
	\clearpage
	\fi
	\pagenumbering{arabic}
	\fi
}

\newcommand{\hide}[1]{ }
\ifdefined\IsLLNCS
\else

\fi

\newcommand{\oo}{o}
\newcommand{\vv}{v}

\newcommand{\heps}{\widehat{\eps}}
\newcommand{\hdelta}{\widehat{\delta}}
\newcommand{\halph}{\widehat{\alpha}}

\newcommand{\pk}{\kappa}

\DeclareMathAlphabet \mathbfcal{OMS}{cmsy}{b}{n}

\newcommand{\C}{C}

\newcommand{\Ensuremath}[1]{\ensuremath{#1}\xspace}

\newcommand{\myOptName}[1]{\Ensuremath{\operatorname{#1}}}

\newcommand{\MathAlgX}[1]{\Ensuremath{\MathAlg{#1}}}

\newcommand{\nuppt}{\ensuremath{{\sf ppt}^{\sf NU}\xspace}}

\newcommand{\Ham}{\operatorname{Ham}}

\newcommand{\wGamma}{\widetilde{\Gamma}}

%% file: Abstract.tex
\begin{abstract}
	Consider a  \ppt two-party protocol  $\Pi = (\Ac,\Bc)$  in which the parties get no private inputs and obtain outputs $\OA,\OB \in \zo$, and let  $\VA$ and $\VB$ denote the parties' individual views. Protocol $\pi$ has \emph{$\alpha$-agreement} if $\Pr[\OA=\OB] = \half+\alpha$. The \emph{leakage} of $\pi$ is the amount of information a party obtains about the event $\set{\OA=\OB}$; that is, the \emph{leakage} $\eps$ is the maximum, over $\Pc \in \set{\Ac,\Bc}$, of the distance between $V^\Pc|_{\OA = \OB}$ and  $V^\Pc|_{\OA \ne \OB}$. Typically, this  distance is measured in \emph{statistical distance}, or, in the computational setting, in \emph{computational indistinguishability}.  For this choice, \citeauthor{wullschleger2009oblivious} [TCC '09] showed that if $\eps \ll \alpha$ then the protocol can be transformed into an OT protocol.
	
	We consider  measuring the protocol leakage  by the \textit{log-ratio distance} (which was popularized by its use in the differential privacy framework). The log-ratio distance between $X,Y$ over domain $\Omega$ is the minimal $\eps \ge 0$ for which, for every $v \in \Omega$, $\log \frac{\Pr[X=v]}{\Pr[Y=v]} \in [-\eps,\eps]$. In the computational setting, we use computational indistinguishability from having log-ratio distance~$\eps$. We show that a protocol with (noticeable) accuracy $\alpha \in \Omega(\eps^2)$ can be transformed into an OT protocol (note that this allows $\eps \gg \alpha$). We complete the picture, in this respect,  showing that a protocol with $\alpha \in o(\eps^2)$ does not necessarily imply OT. Our results hold for both the information theoretic and the computational settings, and can be viewed as a ``fine grained'' approach to ``weak OT amplification''.
	
	We then use the above result  to \emph{fully} characterize the complexity of differentially private two-party computation for the XOR function, answering the open question put  by  \citeauthor*{Goyal2016KMPS} [ICALP '16]  and \citeauthor*{HNOSS18}   [FOCS '18]. Specifically, we show that for any (noticeable) $\alpha\in \Omega(\eps^2)$,  a two-party protocol that computes the XOR function with  $\alpha$-accuracy and $\eps$-differential privacy  can be transformed into an OT protocol. This improves upon \citeauthor{Goyal2016KMPS} that only handle $\alpha\in\Omega(\eps)$, and upon \citeauthor{HNOSS18} who showed that such a protocol implies  (infinitely-often) key agreement  (and not OT). Our characterization is tight since OT does not follow from protocols in which $\alpha \in o( \eps^2)$, and extends to  functions (over many bits) that ``contain'' an ``embedded copy'' of the XOR function.
	
	\Inote{Made some polishing edits}
\end{abstract}

%% file: Introduction.tex
\section{Introduction}\label{sec:Intro}

Oblivious transfer (OT), introduced by \citet{Rabin81}, is one of the most fundamental primitives in cryptography and a complete primitive for secure multi-party computation \cite{Yao86,GoldreichMW87}. Oblivious transfer protocols are known to exist assuming (several types of) \textit{families of trapdoor permutations} \cite{EvenGL85,Haitner04},  \textit{learning with errors} \cite{PeikertVW2008}, \textit{decisional
	Diffie-Hellman} \cite{naor2001efficient, aiello2001priced},\textit{computational Diffie-Hellman} \cite{bellare1989non}  and \textit{quadratic residuosity} \cite{kalai2005smooth}. While in some of the constructions of OT in the literature, the construction  immediately yields a full-fledged OT, in others it only yields a ``weak'' form of OT, that is later ``amplified'' into a full-fledged one.

In this paper we introduce a new notion for a ``weak form of OT'', and show how to amplify this ``weak OT'' into full-fledged OT. This notion is more ``fine grained'' than some previously suggested notions, which  allows us to obtain OT in scenarios that could not be handled by previous works. Our approach is suitable for the  computational  and for the  information theoretic settings (\ie the dishonest parties are assumed to be computationally bounded or not).

\subsection{Our Results}
We start with presenting our results in the information  theoretic setting, and then move to the computation one.

\subsubsection{The Information Theoretic Setting}
The information theoretic analogue of a two-party protocol between parties $\Ac$ and $\Bc$, is a ``channel'': namely, a quadruple of random variables  $C=((\VA,\OA),(\VB,\OB))$, with the interpretation that when ``activating'' (or ``calling'')  the channel $C$, party $\P \in \ABc$ receives his ``output'' $\OP$ and his ``view'' $V^\Pc$. In other words, ``activating a channel'' is analogous to running a two-party protocol with fresh randomness. (We  assume that the view $V^\Pc$ contains the output $\OP$).

\paragraph{Log-ratio leakage (channels).}

We are  interested in the special case where the channel $C=((\VA,\OA),(\VB,\OB))$ has Boolean  outputs (\ie $\OA,\OB \in \zo$), and  assume for simplicity that the channel is \emph{balanced}, meaning that for both $\Pc \in \ABc$, $\OP$ is uniformly distributed. Such channels are  parameterized by their \emph{agreement} and \emph{leakage}:
\begin{itemize}
\item A channel  $C$ has \emph{$\alpha$-agreement} if $\Pr[\OA=\OB] = \half + \alpha$. (Without loss of generality,  $\alpha \ge 0$, as otherwise one of the parties can flip his output).

\item The \emph{leakage} of party $\Bc$ in $C$ is the distance between the distributions $V^\Ac|_{\OA = \OB}$ and  $V^\Ac|_{\OA \ne \OB}$. (Note that these two distributions are well defined if $\alpha \in [0,\half)$). The leakage of party $\Ac$ is defined in an analogous way, and the leakage of $C$ is the maximum of the two leakages.
\end{itemize}

This approach (with somewhat different notation) was taken by past work \cite{wullschleger2009oblivious,WullschlegerThesis}, using \textit{statistical distance} as the distance measure.

Loosely speaking, leakage measures how well can a party distinguish the case  $\set{\OA=\OB}$ from the case  $\set{\OA \ne \OB}$. As each party knows his output, this can be thought of as the ``amount of information'' on the input of one party that \emph{leaks} to the other party.\footnote{We remark that one should be careful with this intuition. Consider a ``binary symmetric channel": a channel in which $\VA=\OA$ and $\VB=\OB$ (\ie the parties receive no additional view except their outputs), $\OA$ is uniformly distributed, and $\OB=\OA \xor U_{p}$ (where $U_p$ is an independent biased coin which is one with probability $p$). The leakage of this channel is zero, for every choice of $p$, whereas each party can predict the output of the other party with probability $1-p$ by using his own output as a prediction.}

We will measure leakage using a \emph{different distance measure}, which we refer to as ``log-ratio distance''.

\begin{definition}[Log-Ratio distance]
Two numbers $p_0,p_1 \in [0,1]$ satisfy\\ $p_0 \rindist_{\eps,\delta} p_1$ if for both $b \in \zo$: $p_b \le e^{\eps} \cdot p_{1-b} + \delta$. Two distributions $D_0,D_1$ over the same domain $\Omega$, are {\sf $(\eps,\delta)$-log-ratio-close} (denoted $D_0 \rindist_{\eps,\delta} D_1$) if for every $A \subseteq \Omega$:
\[ \Pr[D_0 \in A] \rindist_{\eps,\delta} \Pr[D_1 \in A]. \]
\end{definition}

We use the notation $D_0 \sindist_{\delta} D_1$ to say that the \emph{statistical distance} between $D_0$ and $D_1$ is at most $\delta$.
Log-ratio distance is a generalization of statistical distance as $\sindist_{\delta}$ is the same as $\rindist_{0,\delta}$. This measure of distance was popularized by its use in the \emph{differential privacy} framework \cite{DworkMNS2006} (that we discuss in Section \ref{sec:intro:dp}).

Loosely speaking, log-ratio distance considers the ``log-ratio function'' $L_{D_0||D_1}(x) \eqdef \log \frac{\pr{D_0=x}}{\pr{D_1=x}}$, and the two distribution are $(\eps,\delta)$-log-ratio-close if this function is in the interval $[-\eps,\eps]$ with probability $1-\delta$. As such, it can be seen as a ``cousin'' of \textit{relative entropy}  (\aka \textit{Kullback--Leibler (KL) divergence}) that measures the expectation of the log-ratio function.

Note that for  $\eps \in [0,1]$, $D_0 \rindist_{\eps,0} D_1$ implies $D_0 \rindist_{0,2\eps} D_1$, but the converse is not true, and the  condition ($D_0 \rindist_{\eps,0} D_1$) gives tighter handle on the distance between independent samples of distributions (as we explain in detail in \cref{sec:technique:log ratio}).

We use the log-ratio  distance to measure leakage in channels. This leads to the following definition (in which we substitute ``log-ratio distance'' as a distance measure).

\begin{definition}[Log-ratio leakage, channels, informal]\label{def:RelLeakage:IT:Inf}
A channel $C=((\VA,\OA),(\VB,\OB))$ has {\sf  log-ratio leakage $(\eps,\delta)$, denoted $(\eps,\delta)$-leakage} if for both $\Pc\in \ABc$:
\[ V^\Pc|_{\OA=\OB} \rindist_{\eps,\delta} V^\Pc|_{\OA \ne \OB} .\]
\end{definition}

This definition is related (and inspired by) the \textit{differential privacy} framework \cite{DworkMNS2006}. In the terminology of differential privacy, this can be restated as follows: let $E$ be the indicator variable for the event $\set{\OA=\OB}$. For both $\Pc \in \ABc$, the ``mechanism'' $V^\Pc$ is $(\eps,\delta)$-differentially private with regards to the ``secret''/``database'' $E$.


\paragraph{Channels of small log-ratio leakage imply OT.}
 \citet{wullschleger2009oblivious} considered channels with small leakage (measured by statistical distance). Using our terminology, he showed for  $\alpha \in [0,\half)$ and $\eps \in [0,1]$ with $\eps$ ``sufficiently smaller than'' $\alpha^2$, a channel with $\alpha$-agreement  and $(0,\eps)$-leakage yields OT. This can be interpreted as saying that if the leakage $\eps$ is sufficiently \emph{smaller} than the agreement $\alpha$, then the channel yields OT. We prove the following ``fine grained'' amplification result, which is restated with precise notation in \Cref{thm:main:IT}.

\begin{theorem}[Channels of small log-ratio leakage imply OT, infromal]\label{thm:Amp:IT:Inf} There exists a constants $c_1>0$ such that the following holds for  every  $\eps,\delta,\alpha$ with  $ c_1 \cdot \eps^2 \le \alpha < 1/8 $ and $\delta \le  \eps^2$: a channel $C$ that has $\alpha$-agreement and $(\eps,\delta)$-\Leak yields OT (of statistical security).
\end{theorem}

For simplicity, let us focus on \cref{thm:Amp:IT:Inf} in the case that $\delta=0$. Two distributions that are $(\eps,0)$-log-ratio close, may have statistical distance $\eps$, and so, a channel with $(\eps,0)$-leakage, can only be assumed to have $(0,\eps)$-leakage (when measuring leakage in statistical distance). Nevertheless, in contrast to \cite{wullschleger2009oblivious}, \cref{thm:Amp:IT:Inf} allows the leakage parameter $\eps$ to be \emph{larger} than the agreement parameter $\alpha$.\footnote{\label{fn:example}To make this more concrete, consider the following channel $C=((\VA,\OA),(\VB,\OB))$: $\OA \from U_{1/2}$, $\OB \from \OA \xor U_{1/2-\alpha}$, $\VA \from \OB \xor U_{1/2-\eps}$, $\VB \from \OA \xor U_{1/2-\eps}$ (where $U_p$ denotes a biased coin which is one with probability $p$, and the three ``noise variables'' are independent). This channel is balanced, has $\alpha$-agreement, and $(O(\eps),0)$-leakage. However, if we were to measure leakage using statistical distance, then we would report that it has $(0,O(\eps))$-leakage. We are assuming that $\eps > \alpha$, and it will be critical that leakage is measured by log-ratio distance, as we do not know how to amplify leakage that is measured by statistical distance in this range.}

The above can be interpreted as saying that when the leakage is ``well behaved'' (that is the $\delta$ parameter in log-ratio distance is sufficiently small), OT can be obtained even from a channel whose leakage $\eps$ is \emph{much larger} than the agreement $\alpha$. This property will be the key for our applications in \cref{sec:intro:dp}.

\paragraph{Triviality of channels with large leakage.}

We now observe that the relationship between $\eps$ and $\alpha$ in \Cref{thm:Amp:IT:Inf} is best possible (up to constants). Namely,  a channel with agreement that is asymptotically smaller than the one allowed in \cref{thm:Amp:IT:Inf} does not necessarily yield OT.

\begin{theorem}[Triviality of channels with large leakage, informal]
\label{thm:triviality:IT:Inf}
There exists a constant $c_2>0$, such that the following holds for every $\eps>0$: there exists a two-party protocol (with no inputs) that when it ends, party $\Pc \in \ABc$  outputs $\OP$ and sees view $\VP$, and the induced channel $C=((\VA,\OA),(\VB,\OB))$ has $(c_2 \cdot \eps^2)$-agreement and $(\eps,0)$-leakage. \Inote{changed from $\alpha$-agreement for $\alpha >c_2 \cdot \eps^2$. Easier to read and good enough for the intro, TMO. Did the same for other triviality statements}
\end{theorem}

Together, the two theorems say that our characterization of ``weak-OT'' using agreement $\alpha$ and $(\eps,0)$-log-ratio leakage has a ``threshold behavior'' at $\alpha \approx \eps^2$: if $\alpha \ge c_1 \cdot \eps^2$ then the channel yields OT, and if $\alpha \le c_2 \cdot \eps^2$ then such a channel can be simulated by a two-party protocol with no inputs (and thus cannot yield OT with information theoretic security). The proof of  \cref{thm:triviality:IT:Inf} uses a variant of the well-known randomized response approach of \citet{W65a}.

\subsubsection{The Computational Setting}

We  consider  a no-input, Boolean output, two-party protocol $\pi=(\Ac,\Bc)$. Namely, both parties receive a security parameter $1^\pk$ as a common input, get no private input,  and both output one bit. We  denote the output of party $\Pc$ by $\OP_\pk$, and its view by $\VP_\pk$. In other words, an instantiation of $\pi(1^\pk)$ can be thought of as inducing a  channel $C_\pk=((\VA_\pk,\OA_\pk),(\VB_\pk,\OB_\pk))$. Similar to the information theoretic setting, protocol $\pi$ has $\alpha$-\Agr if for every $\pk \in \N$: $\pr{\OA_\pk = \OB_\pk} = 1/2 + \alpha(\pk)$.

\paragraph{Log-ratio leakage (protocols).}

We  extend the definition of log-ratio leakage to the computational setting (where adversaries are \ppt machines). We will use the simulation paradigm to extend the information theoretic definition to the computational setting.

\begin{definition}[Log-ratio leakage, protocols, informal]\label{def:RelLeakage:C:Inf}
A two-party no-input Boolean output protocol $\pi = (\Ac,\Bc)$ has {\sf  Comp-log-ratio leakage $(\eps,\delta)$, denoted $(\eps,\delta)$-\CLeak}, if  there exists an ``ideal channel'' ensemble  $\tC=\set{\tC_\pk=((\tVA_\pk,\tOA_\pk),(\tVB_\pk,\tOB_\pk))}_{\pk \in \N}$ such that the following holds:
\begin{itemize}
\item For every $\pk \in \N$: the channel  $\tC_\pk$ has $(\eps(\pk),\delta(\pk))$-leakage.	

\item For every   $\Pc \in \ABc$: the ensembles $\set{{\VP_\pk,\OA_\pk,\OB_\pk}}_{\pk \in \N}$ and $\set{\tVP_\pk,\tOA_\pk,\tOB_\pk}_{\pk \in \N}$ are \emph{computationally indistinguishable}.\footnote{In the technical section, we consider computational indistinguishability by both \emph{uniform} and \emph{nonuniform} \ppt machines. We ignore this issue in the introduction.}
\end{itemize}	
\end{definition}

\paragraph{Protocols of small log-ratio leakage imply OT.}
We  prove the following computational analogue of \cref{thm:Amp:IT:Inf} (the next Theorem is restated with precise notation in \Cref{thm:main:C}). 


\begin{theorem}[Amplification of protocols with small log-ratio leakage, informal]\label{thm:Amp:C:Inf} There exists a constant $c_1>0$ such that the following holds for every function $\eps,\delta,\alpha$ with $ c_1 \cdot \eps(\pk)^2 \le \alpha(\pk) < 1/8$, $\delta(\pk) \le  \eps(\pk)^2$ and $1/\alpha(\pk) \in \poly(\pk)$: a \ppt protocol that has  $\alpha$-agreement and $(\eps,\delta)$-\CLeak yields  OT (of computational security).
\end{theorem}

\paragraph{Triviality of protocols with large leakage.}
An immediate corollary of \cref{thm:triviality:IT:Inf}  is the relationship between $\eps$ and $\alpha$ in \Cref{thm:Amp:C:Inf} is best possible (up to constants). 

\begin{corollary}[Triviality of protocols with large leakage, informal]
	\label{cor:triviality:comp:Inf}
	There exists a constant $c_2>0$, such that the following holds for every function $\eps$ with $\eps(\pk) > 0$: there exists a \ppt protocol that has $( c_2 \cdot \eps^2)$-agreement and $(\eps,0)$-leakage.
\end{corollary}

\subsubsection{Application: Characterization of Two-Party Differentially Private Computation.}
\label{sec:intro:dp}

We use our results to characterize the complexity of differentially private two-party computation for the XOR function, answering the open question put  by  \cite{Goyal2016KMPS,HNOSS18}. The framework of differential privacy typically studies a ``one-party'' setup, where a ``curator'' wants to answer statistical queries on a database without compromising the privacy of individual users whose information is recorded as rows in the database \cite{DworkMNS2006}. In this paper, we are interested in \emph{two-party} differentially-private computation (defined in \cite{MironovPRV09}). This setting is closely related to the setting of secure function evaluation: the parties $\Ac$ and $\Bc$ have private inputs $x$ and $y$, and wish to compute some functionality $f(x,y)$ without compromising the privacy of their inputs. In secure function evaluation, this intuitively means that parties do not learn any information about the other party's input, that cannot be inferred from their own inputs and outputs. This guarantee is sometimes very weak: For example, for the XOR function $f(x,y)=x \xor y$, secure function evaluation completely reveals the inputs of the parties (as a party that knows $x$ and $f(x,y)$ can infer $y$). Differentially private two-party computation aims to give some nontrivial security even in such cases (at the cost of compromising the \emph{accuracy} of the outputs).

\begin{definition}[Differentially private computation \cite{MironovPRV09}]\label{def:DP:Inf}
	A \ppt two-party\\ protocol $\pi = (\Ac,\Bc)$ over input domain $\zn\times \zn$ is {$\eps$-DP}, if for every \ppt nonuniform machines $\Bc^\ast$ and $\Dc$, and every $x,x' \in \zn$ with $\Ham(x,x') =1$:   let $V^{\Bc^\ast}_\pk(x)$ be the view of $\Bc^\ast$ in a random execution of $(\Ac(x),\Bc^\ast)(1^\pk))$, then
	$$\pr{\Dc(V^{\Bc^\ast}_\pk(x)) = 1}  \leq   e^{\eps(\pk)} \cdot \pr{\Dc(V^{\Bc^\ast}_\pk(x')) = 1} + \negl(\pk),$$
	and the same hold for the secrecy of $\Bc$.
	
	Such a protocol is {\sf semi-honest $\eps$-DP}, if the above is only guaranteed for semi-honest adversaries (\ie for $\Bc^\ast = \Bc$).
\end{definition}

In this paper, we are interested in functionalities $f$, in which outputs are single bits (as in the case of the XOR function). In this special case, the accuracy of a protocol can be measured as follows:

\begin{definition}[accuracy]
	A \ppt two-party protocol $\pi = (\Ac,\Bc)$ over input domain $\zn\times \zn$ with outputs $\OA(x,y),\OB(x,y) \in \zo$ has \emph{perfect agreement} if for every $x,y \in \zn\times \zn$, and every $\pk \in \N$, in a random execution of the protocol $(\Ac(x),\Bc(y))(1^\pk)$, it holds that $\Pr[\OA(x,y)=\OB(x,y)]=1$.

The protocol implements a functionality $f$ over input domain $\zn\times \zn$ with \emph{$\alpha$-accuracy}, if for $\pk \in \N$, every $\Pc \in \ABc$, and every $x,y \in \zn\times \zn$, in a random execution of the protocol $(\Ac(x),\Bc(y))(1^\pk)$, it holds that $\Pr[\OP(x,y)=f^{\Pc}(x,y)]=\half + \alpha(\pk)$.
\end{definition}

A natural question is what assumptions are needed for two-party differentially private computation achieving a certain level of accuracy/privacy (for various functionalities). A sequence of works showed that for certain tasks, achieving high  accuracy requires one-way functions \cite{BeimelNO08, ChanSS12,MMPRTV11,GoyalMPS2013}; some cannot even be instantiated in the random-oracle model \cite{HaitnerOZ2016}; and some cannot be black-box reduced to key agreement  \cite{KhuranaMS2014}. See \cref{sec:intro:relatedWork} for more details on these results. In this work we fully answer the above question for the XOR function.

Consider the functionality $f_{\alpha}(x,y)$ which outputs $x \xor y \xor U_{1/2-
\alpha}$ (where $U_{1/2-\alpha}$ is an independent biased coin which is one with probability $1/2 - \alpha$). Assuming OT, there exists a two-party protocol that securely implement  $f_{\alpha}$, and this protocol is $\eps$-DP, for $\eps = \Theta(\alpha)$. This is the best possible differential privacy that can be achieved for accuracy $\alpha$. On the other extreme, an $\Theta(\eps^2)$-accurate, $\eps$-differential private,  protocol for computing  XOR  can  be constructed (with information theoretic security) using the so-called \textit{randomized response} approach of \citet{W65a}, as shown in \cite{GoyalMPS2013}. Thus, it is natural to ask whether OT \emph{follows} from $\alpha$-accurate, $\eps$-DP computation of XOR, for intermediate choices of $\eps^2 \ll \alpha \ll \eps$.  In this paper, we completely resolve this problem and prove that OT is implied for any intermediate $\eps^2 \ll \alpha \ll \eps$.

\paragraph{Differentially private XOR to OT, a tight characterization.}

\begin{theorem}\label{thm:DPXORInf}[Differentially private XOR to OT, informal] There exists a constant $c_1>0$ such that the following holds for every function $\eps,\alpha$ with $\alpha \ge c_1 \cdot \eps^2$ such that $1/\alpha \in \poly$: the existence of a perfect agreement, $\alpha$-accurate, semi-honest $\eps$-DP \ppt protocol  for computing  XOR implies OT (of computational security).
 \end{theorem}

The above  improves upon \citet{Goyal2016KMPS}, who gave a positive answer if the accuracy $\alpha$ is the best possible: if $\alpha \ge c \cdot \eps$ for a constant $c$. It also improves  (in the implication) upon  \citet{HNOSS18}, who showed that   $c\cdot \eps^2$-correct $\eps$-DP  XOR  implies (infinitely-often) key agreement. Finally, our result allows $\eps$ and $\alpha$ to be function of the security parameter (and furthermore, allow $\alpha$ and $\eps$ to be polynomially small in the security parameter) whereas previous reductions \cite{Goyal2016KMPS,HNOSS18} only hold for constant values of  $\eps$ and $\alpha$. Our characterization is tight as OT does not follow from protocols with  $\alpha \in o( \eps^2)$.

\begin{theorem}[Triviality of  differentially private XOR with large leakage. Folklore, see  \cite{GoyalMPS2013}]\label{prop:triviality:DP:Inf}
	There exists a constant $c_2>0$ such that for every functions $\eps$ there exists a  \ppt protocol  for computing  XOR with information-theoretic $\eps$-DP, perfect agreement and accuracy $c_2\cdot \eps^2$.~\footnote{The protocol is  the randomized response one, and the proof is very similar to that of  \cref{thm:triviality:IT:Inf} (see  \cref{sec:ChannelsAmp}).}
\end{theorem}

\paragraph{Perspective.} Most of the work in differentially private mechanisms/protocols is in the information theoretic setting (using the addition of random noise). There are, however, examples where using computational definitions of differential privacy together with cryptographic assumptions, yield significantly improved accuracy and privacy compared to those that can be achieved in the information theoretic setting (\eg  the inner product and the Hamming distance functionalities~\cite{MMPRTV11}, see more references in the related work section below).
Understanding the minimal assumptions required in this setting is a fundamental open problem. In this paper, we completely resolve this problem for the special case of the XOR function. We stress that the XOR function is the canonical example of a function $f(x,y)$ where the security guarantee given by secure function evaluation is very weak. More precisely, for $f(x,y)=x \oplus y$, the security guaranteed by secure function evaluation is meaningless, and the protocol in which both parties reveal their private inputs is considered secure. Differential privacy can be used to provide a meaningful definition of security in such cases, and we believe that the tools that we developed for the XOR function, can be useful to argue about the minimal assumptions required for other functionalities. As a first step, we provide a sufficient condition under which our approach applies to other functionalities $g:\zo^n \times \zo^n \to \zo$.

\paragraph{Extending the result to any function that is not monotone under relabeling.}
We can use our results on the XOR function to achieve OT from differentially private, and sufficiently accurate computation of a wide class of functions that are not ``monotone under relabeling''. A function $g:\zo^n \times \zo^n \to \zo$ is monotone under relabeling if there exist two bijective functions $\sigma_x,\sigma_y:[2^n]\to \zn$ such that for every $x \in \zn$ and $i\leq j \in [2^n]$:
	\begin{align*}
	g(x,\sigma_y(i))\leq g(x,\sigma_y(j)),
	\end{align*}
		and, for every $y \in \zn$ and $i\leq j \in [2^n]$:
	\begin{align*}
		g(\sigma_x(i),y)\leq g(\sigma_x(j),y).
	\end{align*}
We observe that every function $g$ that is not monotone under relabeling has an ``embedded XOR'', meaning that there exist $x_0,x_1,y_0,y_1 \in \zo^n$ such that for every $b,c \in \zo$, $g(x_b,y_c)=b \xor c$. This gives that a two-party protocol that computes $g$  can be used to give a two-party protocol that computes XOR (with some losses in privacy) and these yield OT by our earlier results.
Precise details are given in \cref{sec:DPXORtoOTl}.


\subsection{Related Work}\label{sec:intro:relatedWork}



\paragraph{Information-theoretic OT.}
Oblivious transfer protocols are also widely studies in their information theoretic forms \cite{prabhakaran2014assisted,CrepeauK88,crepeau1997efficient, nascimento2008oblivious, wolf2004zero}. In this form, and OT is simply a pair of jointly distributed random variable $(V_\Ac,V_\Bc)$ (a ``channel''). A pair of  unbounded parties $(\Ac,\Bc)$, having access to independent samples from this pair (from  each sample $(v_\Ac,v_\Bc)$, party $\Pc$ gets the value $v_\Pc$).  Interestingly, in the information theoretic form, we do have  a ``simple'' notion of weak OT, that is complete: such a pair can either be used to construct full-fledged (information theoretically secure) OT,   or is trivial---there exists a protocol that generates these views.  Unfortunately, these reductions are inherently inefficient: the parties wait till an event that might be of  arbitrary small  probability to occur, and thus, at least not in the most general form,   cannot be translated into  the computational setting.

\paragraph{Hardness amplification.}
Amplifying the security  of  weak primitives into ``fully secure'' ones is an important paradigm in cryptography as well as  other key fields in theoretical computer science. Most notable such works in cryptography  are amplification of one-way functions \cite{Yao82,GoldreichKL93,HaitnerHR11}, key-agreement protocols \cite{Holenstein06}, and interactive arguments \cite{HastadPPW10,Haitner13}. Among the above, amplification of key-agreement  protocols (KA) is the most similar   to the OT amplification we consider in this paper. In particular,  we do have  a ``simple'' (non distributional) notion  of  weak KA \cite{Holenstein06}. This is done by reduction to the information theoretic notion of key-agreement.  What enables this reduction to go through, is that unlike the case of the information theoretic  OT, the amplification of  information theoretic KA is efficient, since it only use the designated output of the (weak) KA (and not the parties' view).

\paragraph{Minimal assumptions for differentially  private symmetric computation.}
An accuracy parameter $\alpha$ is \emph{trivial}  \wrt  a given functionality $f$ and differential privacy parameter $\eps$, if a protocol computing $f$ with such accuracy and privacy exists information theoretically (\ie with no  computational assumptions). The accuracy  parameter  is called \emph{optimal},  if it matches the bound achieved in  the client-server model.
Gaps between the trivial and optimal accuracy parameters have been shown in the multiparty case for count queries~\cite{BeimelNO08, ChanSS12}  and in the two-party case for inner product and Hamming distance functionalities~\cite{MMPRTV11}. 
\cite{HaitnerOZ2016}  showed that the same holds also when a random oracle is available to the parties,   implying that non-trivial  protocols (achieving non-trivial  accuracy)  for computing  these functionalities    cannot be black-box reduced to one-way functions.  

\cite{GoyalMPS2013} initiated the study of Boolean  functions, showing  a gap  between the optimal  and trivial accuracy   for the XOR or the AND functionalities, and that non-trivial  protocols  imply  one-way functions. \cite{FairouzSV2014} showed that non-interactive randomised response is optimal among all the information theoretic protocols.  \cite{KhuranaMS2014} have shown that optimal protocols for computing the XOR or AND, cannot be black-box reduced to key agreement. 

\cite{Goyal2016KMPS} showed that an optimal protocol (with best possible parameters) computing the XOR can be viewed as a form of weak OT, which according to \citet{wullschleger2009oblivious} yields full fledged OT. Whereas for our choice of parameters the security guarantee is too weak, and it is essential that we correctly amplify the security.

Very recently, \cite{HNOSS18} showed that a non-trivial  protocol for computing XOR (\ie accuracy better than $\eps^2$) implies infinitely often key-agreement protocols. Their reduction, however, only holds for constant value of $\eps$, and is non black box.
Finally, \cite{beimel1999all,HarnikNRR2006} gave a criteria that proved the necessity of OT for computationally secure function evaluation, for a select class of functions. 

\subsection*{Paper Organization}
Due to space limitations, some of the technical details appear in the full version of this paper.
In \Cref{sec:technique} we give an overview of the main ideas used in the proof. In \Cref{sec:Preliminaries} we give some preliminaries and state some earlier work that we use. In \Cref{sec:ChannelsAmp} we give our amplification results, that convert protocols with small log-ratio leakage into OT. The proofs of our results on two-party differentially private computation of the XOR function, and on functions that are not monotone under relabeling omitted from this version.

%% file: Technique.tex
\section{Our Technique}\label{sec:technique}

In this section we give a high level overview of our main ideas and technique.

\subsection{Usefulness of Log-Ratio Distance}\label{sec:technique:log ratio}

Recall that the  \emph{leakage} we considered is measured using \emph{log-ratio distance}, and not \emph{statistical distance}. We  survey some advantages of log-ratio distance over statistical distance.

As is common in ``hardness amplification'', our construction will apply the original channel/protocol many times (using fresh randomness). Given a distribution $X$, let $X^{\ell}$ denote the distribution of $\ell$ independent samples from $X$. A natural question is how does the distance between $X^\ell$ and $Y^\ell$ relate to the distance between $X$ and $Y$. For concreteness, assume that $\SD(X,Y)=\eps$ (where $\SD$ denotes statistical distance) and that we are interested in taking $\ell={c}/{\eps^2}$ repetitions where $c>0$ is a very small constant. Consider the following two examples (in the following we use $U_p$ to denote a coin which is one with probability~$p$):
\begin{itemize}
\item $X_1=U_0$ and $Y_1=U_{\eps}$. In this case, $\SD(X_1^{\ell},Y_1^{\ell})=1-(1-\eps)^{\ell} \approx 1-e^{-c/\eps}$ which approaches one for small $\eps$.

\item $X_2=U_{1/2}$ and $Y_2=U_{1/2+\eps}$, in this case $\SD(X_2^{\ell},Y_2^{\ell})=\eta$, where $\eta \approx \sqrt{c}$ is a small constant that is independent of $\eps$, and can be made as small as we want by decreasing $c$.
\end{itemize}

There is a large gap in the behavior of the two examples. In the first, the distance is very close to one, while in the second it is very close to zero. This means that when we estimate $\SD(X^{\ell},Y^{\ell})$ in terms of $\SD(X,Y)$, we have to take a \emph{pessimistic} bound corresponding to the first example, which is far from the truth in case our distributions behave like in the second example.

Loosely speaking, log-ratio distance provides a ``fine grained'' view that distinguishes the above two cases. Note that $X_2 \rindist_{O(\eps),0} Y_2$, whereas there is no finite $c$ for which $X_1 \rindist_{c,0} Y_1$. For $X,Y$ such that  $X \rindist_{\eps,\delta} Y$ for $\delta=0$ (or more generally, for $\delta \ll\eps$)  we get the behavior of the second example under repetitions, yielding a better control on the resulting statistical distance. More precisely, it is not hard to show that if $X \rindist_{\eps}Y$ then for  $\ell=c/\eps^2$ it holds that  $X^{\ell} \sindist_{O(\sqrt{c \cdot ln(1/c)})} Y^{\ell}$.\footnote{
Let us explain the intuition behind the above  phenomenon. The maximum value of both $L_{X||Y}(s)=\log \frac{\Pr[X=s]}{\Pr[Y=s]}$ and $L_{Y||X}(s)=\log \frac{\Pr[Y=s]}{\Pr[X=s]}$, is at most $\eps$. The relative entropy (\aka KL divergence) $D(X||Y)$ measures the expectation of $L_{X||Y}(s)$ according to $s \from X$, and is therefore smaller than $\eps$. But in fact it is easy to show that both $D(X||Y)$ and $D(Y||X)$ are bounded by $\eps \cdot (e^{\eps}-1)$ which is approximately $\eps^2$ for small $\eps$.  It follows that  $D(X^\ell || Y^\ell) = \ell \cdot D(X||Y) \approx \ell \eps^2 = c$. In other words, the expectation of $L_{X^{\ell}||Y^{\ell}}= D(X^\ell || Y^\ell) = c$. The random variable $L_{X^{\ell}||Y^{\ell}}$ can be seen as the sum of $\ell$ independent copies of $L_{X||Y}$, and we know that each of these variables lies in the interval $[-\eps,\eps]$. By a standard Hoeffding bound it follows that the probability that $L_{X||Y}$ deviates from the expectation $c$, by say some quantity $\eta$ is at most $e^{-\Omega(\frac{\eta^2}{\ell \eps^2})}=e^{-\Omega(\eta^2/c)}$ and this means that we can choose $\eta$ to be roughly $\sqrt{c \cdot \ln(1/c)}$ and obtain that the probability of deviation is bounded by $\eta$. Overall, this gives that $X^{\ell} \rindist{\eta+c,\eta} Y^{\ell}$, meaning that except for an $\eta$ fraction of the space, the ratio is bounded by $\eta+c$, and therefore, the statistical distance is also bounded by $O(\eta+c)=O(\sqrt{c \cdot \ln(1/c)})$.} A more precise statement and proof are given in \cref{thm:rep}.\footnote{This phenomenon is the rationale behind the differential privacy boosting result of \cite{DworkR2016}, and can be derived from the proof in that paper.  In our setting, however, the proof is straightforward as outlined here, and shown in the proof of \cref{thm:rep}.}

%
%

\subsection{The Amplification Protocol}

In this section we give a high level overview of the proof of \cref{thm:Amp:IT:Inf}. The starting point is a channel $C=((\VA,\OA),(\VB,\OB))$ that has $\alpha$-\Agr, and $(\eps,\delta)$-\Leak. (A good example to keep in mind is the channel from \cref{fn:example}). For simplicity of exposition, let us assume that $\delta=0$ (the same proof will go through if $\delta$ is sufficiently small). Our goal is to obtain OT if $\alpha \ge c_1 \cdot \eps^2$ for some constant $c_1$, which we will choose to be sufficiently large.

\newcommand{\Wul}{{\mathsf{Wul}}}
\citet{wullschleger2009oblivious} showed that a balanced channel with
$\alpha'$-agreement, and $(0,\eps')$-leakage (that is $\eps'$ leakage in statistical distance) implies OT if $\eps' \le c_\Wul \cdot (\alpha')^2$ for some constant $c_\Wul>0$. Thus, we are looking for a protocol, that starts with a channel that has $(\eps,0)$-leakage and $\alpha$-agreement, where $\eps$ is \emph{larger} than $\alpha$, and produces a channel with $(0,\eps')$-leakage, and $\alpha'$-agreement where $\eps'$ is \emph{smaller} than $\alpha'$. We will use the following protocol  achieving $\alpha' \ge 1/5$ and an arbitrarily small constant $\eps'>0$.\footnote{Similar protocols were used in the context of key-agreement amplification \cite{bennett1995generalized,maurer1993secret}.}

\begin{protocol}[$\Delta_\ell^C = (\tAc,\tBc)$, amplification of log-ratio leakage]~\label{proto:technique:OT:IT}
	\item[Channel:] $C=((\VA,\OA),(\VB,\OB))$.
	\item[Prameter:]  Number of samples $\ell$.
	
	\item[Operation:] Do until the protocol produces output:
		\begin{enumerate}
		\item The parties activate the channel $C$ for  $\ell$ times. Let $\oOA$ and $\oOB$ be the ($\ell$-bit) outputs.

		\item $\tAc$ sends the (unordered) set  $\cs = \{\oOA, \oOA  \xor 1^\ell\}$ to $\tBc$.

		\item $\tBc$ informs $\tAc$ whether $\oOB \in \cs$.
		
		If positive,  party $\tAc$ outputs zero if $\oOA$ is the  (lex.)  smallest element in $\cs$, and one otherwise.  Party  $\tBc$ does the same \wrt $\oOB$.  \ \ (And the protocol halts.)

	\end{enumerate}
\end{protocol}

Let $\Delta = \Delta_\ell^C$ for  $\ell=1/4\alpha$. We first observe that $\Delta$ halts  in a given iteration  iff the event $E=\set{\oOA \oplus \oOB \in \set{0^{\ell},1^{\ell}}}$ occurs. Note that $\Pr[E] \ge 2^{-\ell}$, and thus the expected running time of $\Delta$   is $O(2^{\ell})=2^{O(1/\alpha)}$ (jumping ahead, the expected running time can be improved to $\poly(1/\alpha)$, see Section \ref{subsec:Efficient}).

We also observe that the outputs of the two parties agree, iff in the final (halting) iteration it holds that $\oOA=\oOB$. Thus, the agreement of $\Delta$  is given by:

\begin{align*}
\Pr[\oOA=\oOB|E] &= \frac{(\half+\alpha)^{\ell}}{(\half+\alpha)^{\ell} + (\half-\alpha)^{\ell}} = \left(1+\left(\frac{\half-\alpha}{\half+\alpha}\right)^{\ell} \right)^{-1} \\ &\approx  \frac{1}{1+e^{-4\alpha \ell}} \ge \frac{1}{1+e^{-1}} \ge  \half + \alpha',
\end{align*}
for $\alpha' \ge 1/5$.

\newcommand{\Final}{\operatorname{final}}
In order to understand the leakage of $\Delta$, we examine the views of the parties in  the \emph{final}  iteration of $\Delta$ (it is clear that the views of the previous iteration yields no information). Let us denote these part of a view $v$ by $\Final(v)$. We are interested in understanding the log-ratio distance between $\Final(\tVA|_{\tOA =\tOB})$ and $\Final(\tVA|_{\tOA \ne\tOB})$. Observe  that  $\Final(\tVA|_{\tOA =\tOB})$ is a (deterministic) function of  $\ell$ independent samples from $\VA|_{\OA=\OB}$ (\ie  the function that appends $\{\oOA, \oOA  \xor 1^\ell\}$ to the view), and  $\Final(\tVA|_{\tOA \ne\tOB})$ is the \emph{same}  deterministic function of $\ell$ independent samples from $\VA|_{\OA \ne \OB}$. Thus, by data processing,  it suffices to  bound  the  distance of  $\ell$ independent samples from $\VA|_{\OA=\OB}$ from   $\ell$ independent samples from $\VA|_{\OA\neq\OB}$. By assumption, $C$ has $(\eps,0)$-\Leak, which means that
\[ \VA|_{\OA=\OB} \rindist_{\eps,0} \VA|_{\OA \ne \OB}. \]
In the previous section we showed that by choosing a sufficiently small constant $c>0$ and taking $\ell=c/\eps^2$ repetitions of a pair of distributions with $(\eps,0)$-log ratio distance, we obtain two distributions with statistical distance that is an arbitrary small constant $\eps'>0$. Here we consider $\ell=1/(4\alpha)=1/(4 c_1\cdot  \eps^2)$ repetitions, and therefore
 \[ \Final(\tVA|_{\tOA =\tOB})\sindist_{\eps'} \Final(\tVA|_{\tOA \ne\tOB}). \]
By picking $c_1$ to be sufficiently large, we can obtain that the leakage in $\Delta$ is $\eps' \le c_\Wul \cdot (\alpha')^2$ as required.

\subsubsection{Efficient Amplification}\label{subsec:Efficient}
The (expected) running time of $\Delta_\ell$ is $2^{O(\ell)}$ that for the above  choice of $\ell=\Theta(1/\alpha)$ equals  $2^{O(1/\alpha)}$. To be useful in a setting when the running time is  limited, \eg in the computational setting, this dependency restricts us to ``large`` values of  $\alpha$.  Fortunately, \cref{proto:technique:OT:IT} can be modified so that its (expected)  running time is only polynomial in $1/\alpha$.

Intuitively, rather than making  $\ell$ invocations of $C$ at once, and hope that the tuple of invocations  happens to be \emph{useful}: $\oOA \oplus \oOB \in \set{0^{\ell},1^{\ell}}$, the efficient protocol   combines   smaller tuples of useful invocations, \ie $\oOA \oplus \oOB \in \set{0^{\ell'},1^{\ell'}}$, for some $\ell' < \ell$, into  a useful tuple of $\ell$ invocations. The advantage is that  failing to generate the smaller useful tuples, only  ``wastes'' $\ell'$ invocations of $C$. By recursively sampling the $\ell'$ tuples via the same approach, we get a protocol whose expected running time is $O(\ell^2)$ (rather than $2^{O(\ell)}$).

The actual protocol implements the above intuition in the following way: on parameter $d$, protocol $\Lambda_d$ mimics the interaction of the inefficient protocol $\Delta_{2^d}$ (\ie the inefficient protocols with sample parameter $2^d$). It does so  by using   $\Delta_2$ to combines the outputs of two of execution of $\Lambda_{d-1}$. Effectively,  this call to $\Delta_2$ combines the two $2^{d-1}$ useful tuples produced by $\Lambda_{d-1}$, into a single $2^{d}$ useful tuple.

Let $\Lambda_0^C =C$, and recursively define $\Lambda_d$, for $d>0$, as follows:
\begin{protocol}[$\Lambda_d^C = (\hAc,\hBc)$, efficient amplification of log-ratio leakage]~\label{proto:technique:OT:IT:Eff}
	\item[Channel:] $C$.
	\item[Prameter:] log  number of sample  $d$.
	
	\item[Operation:] The parties interact in $\Delta_2^{(\Lambda_{d-1}^C)}$.~\\
\end{protocol}
By induction, the expected running time of $\Lambda_{d}^C$ is $4^d$. A more careful analysis yields that the view of  $\Lambda_{d}^C$ can be \emph{simulated} by the view of $\Delta_{2^d}^C$. Indeed, there are exactly $2^d$ useful invocations of $C$ in an execution  of $\Lambda_{d}^C$: invocations whose value was not ignored by the parties, and their distribution is exactly the same as the $2^d$ useful invocations of $C$ in  $\Delta_{2^d}^C$.   Hence, using $\Lambda_d^C$ with $d= \log 1/4\alpha$, we get a protocol whose expected running time is polynomial in $1/\alpha$  and guarantees the same  level of agreement and security as of $\Delta_{1/4\alpha}$.

\subsection{The Computational Case}

So far, we considered information theoretic security. In order to prove \Cref{thm:Amp:C:Inf} (that considers security against \ppt adversaries) we note that
\cref{def:RelLeakage:C:Inf} (of computational leakage) is carefully set up to allow the argument of the previous section to be extended to the computational setting. Using the efficient protocol above, the reduction goes through as long as $\alpha$ is a noticeable function of  the security parameter.

\subsection{Two-Party Differentially Private XOR Implies OT}

In this section we explain the main ideas that are used in the proof of \Cref{thm:DPXORInf}. Our goal is to show that a perfect completeness, $\alpha$-accurate, semi-honest $\eps$-DP protocol for computing XOR, implies OT, if $\alpha \ge c \cdot \eps^2$ for a sufficiently large constant $c$. In order to prove this, we will show that such a protocol can be used to give a two-party protocol that has $\alpha$-agreement and (computational) $(\eps,0)$-leakage. Such a protocol yields OT by our earlier results.\footnote{We believe that our results extend to the case of $(\eps,\delta)$-differential privacy, as long as $\delta=o(\eps^2)$, and then we obtain $(\eps,\delta)$-leakage, which is sufficient to yield OT. Proving this requires a careful examination of some of the previous work (which was stated for $\delta=0$) and extending it to nonzero $\delta$, as well as a more careful analysis on our part. We will not do this in this paper.}

We remark that there are two natural definitions of ``computational differential privacy'' in the literature using either \emph{computational indistinguishability} or \emph{simulation} \cite{MironovPRV09}. \Cref{def:DP:Inf} is using indistinguishability, while for our purposes, it is more natural to work with simulation (as using simulation enables us to ```switch back and forth'' between the information theoretic setting and the computational setting). In general, these two definitions are not known to be equivalent. For functionalities like XOR, where the inputs of both parties are single bits, however, the two definitions are equivalent  by the work of \cite{MironovPRV09}. This means that when considering differential privacy of the XOR function, we can imagine that we are working in an information theoretic setting, in which there is a trusted party, that upon receiving the inputs $x,y$ of the parties, provides party $\Pc$, with its output $\OP$ and view $\VP$. We will use the following protocol to obtain a ``channel'' with $\alpha$-agreement and $(\eps,0)$-leakage.

\begin{protocol}[DP-XOR to channel]~\label{intro:pro:DPXOR:IT}
		\begin{enumerate}
		\item 	 $\Ac$ samples  $X \gets\zo$ and  $\Bc$ samples  $Y \gets\zo$.
		
		\item The parties apply the differentially private protocol for computing XOR, using inputs $X$ and $Y$ respectively, and receive outputs $\OA_{DP},\OB_{DP}$ respectively.
		
		\item $\Ac$ sends   $R\gets\zo$ to $\Bc$.
		
		\item $\Ac$ outputs $\OA=X \xor R$ and $\Bc$ outputs $\OB_{DP} \xor Y \xor R$.
	\end{enumerate}
\end{protocol}

The intuition behind this protocol is that if $\OB_{DP}=X \xor Y$, then $\OB=(X \xor Y) \xor Y \xor R = X \xor R = \OA$. This means that the channel induced by this protocol inherits $\alpha$-agreement from the $\alpha$-accuracy of the original protocol.
In  \cref{sec:ChannelsAmp} we show that this channel ``inherits'' log-ratio leakage of $(\eps,0)$ from the fact that the original protocol is $\eps$-DP.

%% file: Preliminaries.tex
\section{Preliminaries}\label{sec:Preliminaries}
\subsection{Notation}\label{sec:prelim:notation}
We use calligraphic letters to denote sets, uppercase for random variables and functions,  lowercase for values. 
For $a,b\in \R$,  let $a\pm b$ stand for the interval $[a-b,a+b]$. For $n\in \N$, let $[n] = \set{1,\ldots,n}$ and $(n) = \set{0,\ldots,n}$. The  Hamming distance between two strings $x,y\in\zn$, is defined by  $\Ham(x,y) = \sum_{i\in [n]} x_i \neq y_i$. Let $\poly$ denote the set of all polynomials, let \ppt stand for probabilistic  polynomial time and   \pptm denote a \ppt TM (Turing machine) and let \nuppt stands for a \emph{non-uniform} \pptm.
 A function $\nu \colon \N \to [0,1]$ is \textit{negligible}, denoted $\nu(n) = \negl(n)$, if $\nu(n)<1/p(n)$ for every $p\in\poly$ and large enough $n$.

\subsection{Distributions and Random Variables}
Given a distribution, or random variable,  $D$, we write $x\gets D$ to indicate that $x$ is selected according to $D$. Given a finite set $\cs$, let $s\gets \cs$ denote that $s$ is selected according to the uniform distribution over $\cs$. The support of $D$, denoted $\Supp(D)$, be defined as $\set{u\in\Uni: D(u)>0}$.  We will use the following distance measures.

\paragraph{Statistical distance.}

\begin{definition}[statistical distance]\label{def:SD}
The {\sf statistical distance} between two distributions $P,Q$ over the same domain $\Uni$, (denote by $\SD(P,Q)$) is defined to be:
\[ SD(P,Q)=max_{\cA \subseteq \Uni} |\Pr[P \in \cA]-\Pr[Q \in \cA]| .\]
We say that $P,Q$ are {\sf $\eps$-close} (denoted by $P \sindist_{\eps} Q$) if $\SD(P,Q) \le \eps$.
\end{definition}

\newcommand{\tY}{\tilde{Y}}
\newcommand{\cS}{{\cal{S}}}
We use the following fact, proof given in the appendix.
\def\PropCondSD
{

		Let $0 <\eps <\mu< 1$, and let $(X,Y)$, $(\tX, \tY)$ be two pairs of random variables  over the same domain $\cX \times \cY$, such that $\SD((X,Y), (\tX,\tY)) \leq \eps$. Let $E_0,E_1 \subseteq \cX \times \cY$ be two sets such that for every $b \in \zo$, $\pr{(X,Y)\in E_b} \geq \mu$. Then $\SD(\tX|_{\set{(\tX,\tY) \in E_0}}, \tX|_{\set{(\tX,\tY) \in E_1}})\leq \SD(X|_{\set{(X,Y) \in E_0}}, X|_{\set{(X,Y) \in E_1}}) +4\eps/\mu.$

}

\begin{proposition}\label{Prop:CondSD}
	\PropCondSD
\end{proposition}

\paragraph{Log-Ratio distance.}
We will also be interested in the following natural notion of ``log-ratio distance'' which was popularized by the literature on differential privacy.

\begin{definition}[Log-Ratio distance]\label{def:Log-Ratio}
Two numbers $p_0,p_1 \ge 0$ satisfy $p_0 \rindist_{\eps,\delta} p_1$ if for both $b \in \zo$: $p_b \le e^{\eps} \cdot p_{1-b} + \delta$. Two distributions $P,Q$ over the same domain $\Uni$, are {\sf $(\eps,\delta)$-log-ratio-close} (denoted $P \rindist_{\eps,\delta} Q$) if for every $\cA \subseteq \Uni$:
\[ \Pr[P \in \cA] \rindist_{\eps,\delta} \Pr[Q \in \cA]. \]
We let $\rindist_\eps $ stands for $\rindist_{\eps,0}$.
\end{definition}

It is immediate that $D_0 \sindist_{\delta} D_1$ iff $D_0 \rindist_{0,\delta} D_1$, and that $D_0 \rindist_{\eps,\delta} D_1$ implies $D_0 \sindist_{(e^{\eps}-1)+\delta} D_1$, and note that for $\eps \in [0,1]$, $e^{\eps}-1 =O(\eps)$. It is also immediate  that the  log-ratio distance respects data processing.

\begin{fact}\label{fact:LR:DP}
Assume 	$P \rindist_{\eps,\delta} Q$, then $f(P) \rindist_{\eps,\delta} f(Q)$ for any (possibly randomized) function $f$.
\end{fact}

%
%
%
%
%
%

\paragraph{Log-Ratio distance under independent repetitions.}

As demonstrated by the framework of differential privacy, working with this notion of ``relative distance'' is often a very convenient distance measure between distributions, as it behaves nicely when considering independent executions. Specifically, let $D^{\ell}$ denote $\ell$ independent copies from $D$, the following follows:

\begin{theorem}[Relative distance under independent repetitions]\label{thm:rep}~\\
If $D_0 \rindist_{\eps,\delta} D_1$ then for every $\ell \ge 1$, and every $\delta' \in (0,1)$
 \[D_0^{\ell} \rindist_{({\eta(\eps,\ell,\delta')},\ell \delta + \delta')} D_1^{\ell},\]
  where $\eta(\eps,\ell,\delta')=\ell \cdot \epsilon(e^\epsilon -1)+\epsilon \cdot \sqrt{2
\ell \cdot \ln(1/\delta')}$.
\end{theorem}

We remark that Theorem \ref{thm:rep} can also be derived by the (much more complex) result on ``boosting differential privacy'' \cite{dwork2010boosting}. However, it can be easily derived directly by a Hoeffding bound, as is done in the next two lemmata.

\begin{lemma}\label{lemma:avgLog}
	If $D_0 \rindist_{\eps} D_1$ then for each $b\in \zo$, $\ex{x \from D_b}{\log\frac{\ppr{D_b}{x}}{\ppr{D_{1-b}}{x}}}\leq \eps(e^\eps-1)$.
\end{lemma}
The term $\ex{x \from D_b}{\log\frac{\ppr{D_b}{x}}{\ppr{D_{1-b}}{x}}}$ is also known as the KL-divergence between $D_b$ and $D_{1-b}$, which is known to be non-negative for every two distribution $D_0, D_1$.
\begin{proof}
\begin{align*}
&\ex{x \from D_b}{\log\frac{\ppr{D_b}{x}}{\ppr{D_{1-b}}{x}}}  \leq \ex{x \from D_b}{\log\frac{\ppr{D_b}{x}}{\ppr{D_{1-b}}{x}}} + \ex{x \from D_{1-b}}{\log\frac{\ppr{D_{1-b}}{x}}{\ppr{D_{b}}{x}}}\\
&= \sum_{x \in \Uni}\ppr{D_b}{x} (\log\frac{\ppr{D_b}{x}}{\ppr{D_{1-b}}{x}} + \log\frac{\ppr{D_{1-b}}{x}}{\ppr{D_{b}}{x}}) + \sum_{x \in \Uni}(\ppr{D_{1-b}}{x}-\ppr{D_{b}}{x}) \log\frac{\ppr{D_{1-b}}{x}}{\ppr{D_{b}}{x}}\\
&=\sum_{x \in \Uni}(\ppr{D_{1-b}}{x}-\ppr{D_{b}}{x}) \log\frac{\ppr{D_{1-b}}{x}}{\ppr{D_{b}}{x}}\\
&\leq \eps \cdot \sum_{x \in \Uni}\size{\ppr{D_{1-b}}{x}-\ppr{D_{b}}{x}}\\
&\leq \eps \cdot \sum_{x \in \Uni}(e^\eps-1)\min(\ppr{D_{1-b}}{x},\ppr{D_{b}}{x})\\
&= \eps \cdot (e^\eps - 1) \sum_{x \in \Uni}\min(\ppr{D_{1-b}}{x},\ppr{D_{b}}{x}) \leq \eps \cdot (e^\eps - 1).\\
\end{align*}

Where the first inequality holds since KL-divergence is non-negative, and the second and third inequalities holds from the definition of Log-Ratio distance.
\end{proof}

\remove{
\begin{lemma}\label{lemma:avgLog}
	If for a fixed $b\in \zo$ and every $x\in \Uni$, $\ppr{D_{1-b}}{x} \leq e^\eps\cdot \ppr{D_b}{x}$ then $\ex{x \from D_b}{\log\frac{\ppr{D_b}{x}}{\ppr{D_{1-b}}{x}}}\leq 2\eps(e^\eps-1)$.
\end{lemma}	

The size $\ex{x \from D_b}{\log\frac{\ppr{D_b}{x}}{\ppr{D_{1-b}}{x}}}$ is also known as the KL-divergence between $D_b$ and $D_{1-b}$, which is know to be non-negative for every two distribution $D_0, D_1$.
\begin{proof}
	\begin{align*}
	&\ex{x \from D_b}{\log\frac{\ppr{D_b}{x}}{\ppr{D_{1-b}}{x}}}  \leq \ex{x \from D_b}{\log\frac{\ppr{D_b}{x}}{\ppr{D_{1-b}}{x}}} + \ex{x \from D_b}{\log\frac{\ppr{D_{1-b}}{x}}{\ppr{D_{b}}{x}}}\\
	&= \sum_{x \in \Uni}\ppr{D_b}{x} (\log\frac{\ppr{D_b}{x}}{\ppr{D_{1-b}}{x}} + \log\frac{\ppr{D_{1-b}}{x}}{\ppr{D_{b}}{x}}) + \sum_{x \in \Uni}(\ppr{D_{1-b}}{x}-\ppr{D_{b}}{x}) \log\frac{\ppr{D_{1-b}}{x}}{\ppr{D_{b}}{x}}\\
	&=\sum_{x \in \Uni}(\ppr{D_{1-b}}{x}-\ppr{D_{b}}{x}) \log\frac{\ppr{D_{1-b}}{x}}{\ppr{D_{b}}{x}}\\
	&\leq 2\eps \cdot \sum_{\ppr{D_{1-b}}{x}\geq \ppr{D_{b}}{x}}\ppr{D_{1-b}}{x}-\ppr{D_{b}}{x}\\
	&\leq 2\eps \cdot \sum_{x \in \Uni}(e^\eps-1)\ppr{D_{b}}{x}\\
	&= 2\eps \cdot (e^\eps - 1) \sum_{x \in \Uni}\ppr{D_{b}}{x} = 2\eps \cdot (e^\eps - 1).\\
	\end{align*}	

Where the first inequality holds since KL-divergence is non-negative, and the second and third inequalities holds from the definition of Log-Ratio distance.
\end{proof}
}

\begin{lemma}\label{lemma:closeToEpsZero}
	If $D_0 \rindist_{\eps,\delta} D_1$, then for each $b \in \zo$ there exist distributions $D'_b$ such that $D'_b \rindist_{\eps}D_{1-b}$, and,  $D'_b\sindist_{\delta} D_b$.
\end{lemma}
\remove{
\begin{proof}
For each $b \in \zo$, 	let $\cS_b\subseteq \Uni$ be the set of all $x$ such that $\ppr{D_{b}}{x} > e^\eps\cdot \ppr{D_{1-b}}{x}$.  Let $\Delta_b = \ppr{D_{b}}{\cS_b}-e^\eps\cdot \ppr{D_{1-b}}{\cS_b}$. We first show that $\Delta_b \leq \delta$. Indeed, by the log-ration distance between $D_0$ and $D_1$, we get that $\ppr{D_{b}}{\cS_b}\leq e^\eps\cdot \ppr{D_{1-b}}{\cS_b} +\delta$.

For each $b$, let us define $D'_b$ as following:
\begin{equation}
\ppr{D'_{b}}{x}=
\begin{cases}
e^\eps\cdot \ppr{D_{1-b}}{x}/(1-\Delta_b), &\text{for } x \in \cS_b\\
\ppr{D_{b}}{x}/(1-\Delta_b), &\text{for } x \notin \cS_b
\end{cases}
\end{equation}

To show that $D'_0 \rindist_{\eps+2\delta}D'_1$, recall that $\log (1/(1-\delta))\leq \log (1+2\delta) \leq 2\delta$. Therefore it is enough to show that for every $b \in 'zo$ and $x \in \Uni$, it holds that $\ppr{D'_{b}}{x} \leq e^\epsilon/(1-\delta)\cdot \ppr{D'_{1-b}}{x}$. Indeed, if $x \in \cS_b$, it holds that:
\begin{align}
\ppr{D'_{b}}{x} = e^\eps\cdot \ppr{D_{1-b}}{x}/(1-\Delta_b) = e^\eps\cdot \ppr{D'_{1-b}}{x}\cdot (1-\Delta_{1-b})/(1-\Delta_b)\leq e^\eps/(1-\delta)\cdot \ppr{D'_{1-b}}{x},
\end{align}
where the last equality holds since $\cS_0 \cap \cS_1 = \emptyset$.
Similarly, assuming that $x \in \cS_{1-b}$, we get that:
\begin{align}
\ppr{D'_{b}}{x} = \ppr{D_{b}}{x}/(1-\Delta_b) = e^{-\eps}\cdot \ppr{D'_{1-b}}{x}(1-\Delta_{1-b})/(1-\Delta_b)\leq 1/(1-\delta)\cdot \ppr{D'_{1-b}}{x}.
\end{align}
And lastly, if $x \notin \cS_b \cup \cS_{1-b}$:
\begin{align}
\ppr{D'_{b}}{x} = \ppr{D_{b}}{x}/(1-\Delta_b) \leq e^{\eps}\cdot \ppr{D_{1-b}}{x}(1-\Delta_{1-b})/(1-\Delta_b) \leq e^\eps/(1-\delta)\cdot \ppr{D'_{1-b}}{x}.
\end{align}

We are left to show that for each $b \in \zo$, $D'_b\sindist_{\delta} D_b$. Indeed,
\begin{align}
\SD(D_b, D'_b) &= \sum_{x \in \Uni, \ppr{D_{b}}{x} \geq \ppr{D'_{b}}{x}}\ppr{D_{b}}{x} - \ppr{D'_{b}}{x}\\\nonumber
&\leq \sum_{x \in \Uni, \ppr{D_{b}}{x} \geq \ppr{D'_{b}}{x}}\ppr{D_{b}}{x} - \ppr{D'_{b}}{x}\cdot(1-\Delta_b)\\\nonumber
&=\sum_{x \in S_b} \ppr{D_{b}}{x} - e^\eps \cdot \ppr{D_{1-b}}{x}\\\nonumber
&=\Delta_b \leq \delta.
\end{align}
\end{proof}

}
\begin{proof}
Fix $b \in \zo$, and	let $\cS^+\subseteq \Uni$ be the set of all $x$ such that $\ppr{D_{b}}{x} > e^\eps\cdot \ppr{D_{1-b}}{x}$, and $\cS^-$ be the set of all $x$ with $\ppr{D_{b}}{x} < e^{-\eps}\cdot \ppr{D_{1-b}}{x}$. First, notice that
\begin{align}\label{eq:sdMax}
\ppr{D_{b}}{\cS^+}-e^\eps\cdot \ppr{D_{1-b}}{\cS^+} \leq \delta, &&e^{-\eps}\cdot \ppr{D_{1-b}}{\cS^-}-\ppr{D_{b}}{\cS^-} \leq \delta
\end{align}
	 Indeed, by the log-ration distance between $D_0$ and $D_1$, we get that $\ppr{D_{b}}{\cS^+}\leq e^\eps\cdot \ppr{D_{1-b}}{\cS^+} +\delta$, and $\ppr{D_{1-b}}{\cS^-}\leq e^\eps\cdot \ppr{D_{b}}{\cS^-} +\delta$.
	
	 In the following we assume for simplicity that $\ppr{D_{b}}{\cS^+}-e^\eps\cdot \ppr{D_{1-b}}{\cS^+}  \geq e^{-\eps}\cdot \ppr{D_{1-b}}{\cS^-}-\ppr{D_{b}}{\cS^-}$, as the other case symmetrically follows.
	
	 In the following, it is shown how to modify $D_b$ to construct $D'_b$. This is done by reducing the probability of every $x \in \cS^+$, and increasing the probability of every $x \in \cS^-$, to keep it inside the range $[e^{-\eps}\ppr{D_{1-b}}{x}, e^{\eps}\ppr{D_{1-b}}{x}]$. To make sure that the resulting $D'_b$ is a probability distribution $(\sum_x \ppr{D'_{b}}{x} =1)$, the probability of other elements may have to be changed. For this purpose, consider the set $\cA = \set{x : \ppr{D_{1-b}}{x} > \ppr{D_{b}}{x}}$. Notice that $\cS^- \subseteq \cA$, and it holds that
\begin{align}
\ppr{D_{1-b}}{\cA} -\ppr{D_{b}}{\cA} = &\sum_{x \in \cA}\ppr{D_{1-b}}{x} - \ppr{D_{b}}{x}\\\nonumber
& = \sum_{x \notin \cA} \ppr{D_{b}}{x} - \ppr{D_{1-b}}{x} \geq \ppr{D_{b}}{\cS^+}-e^\eps\cdot \ppr{D_{1-b}}{\cS^+}
\end{align}
Where the second equality holds since $\sum_x \ppr{D_{b}}{x}=\sum_x \ppr{D_{1-b}}{x} =1$.

We get that
\begin{equation}
\ppr{D_{1-b}}{\cA} -\ppr{D_{b}}{\cA}\geq \ppr{D_{b}}{\cS^+}-e^\eps\cdot \ppr{D_{1-b}}{\cS^+}  \geq e^{-\eps}\cdot \ppr{D_{1-b}}{\cS^-}-\ppr{D_{b}}{\cS^-}
\end{equation}
Thus, we can define $D'_b$ as following:
\begin{itemize}
	\item For every $x \in \cS^+$, $\ppr{D'_b}{x} = e^\epsilon \cdot \ppr{D_{1-b}}{x}$.
	\item For every $x \notin \cA \cup \cS^+$, $\ppr{D'_b}{x} =  \ppr{D_{b}}{x}$.
	\item For every $x \in \cA$, $\max(\ppr{D_b}{x}, e^{-\epsilon} \cdot \ppr{D_{1-b}}{x}) \leq \ppr{D'_b}{x} \leq \ppr{D_{1-b}}{x}$, such that $\sum_{x \in \Uni}\ppr{D'_b}{x} = 1$.
\end{itemize}
It is clear from \cref{eq:sdMax} that $D'_b\sindist_{\delta} D_b$. Also, from definition it holds that for every $x$, $\ppr{D'_b}{x} \rindist_\eps \ppr{D_{1-b}}{x}$.
\end{proof}

\begin{theorem}[Hoeffding bound \cite{Hoeffding63}]\label{fact:Hoeffding}
	Let $A_1, ..., A_\ell$ be independent random variables s.t. $A_i \in [-c, c]$ and let $\widehat{A} = \Sigma_{i=1}^\ell A_i$.
	It holds that:
	\begin{align*}
	\pr{\widehat{A}-\ex{}{\widehat{A}}\geq  t}\leq e^{-t^2/2\ell c^2}.
	\end{align*}
\end{theorem}

\begin{proof}[Proof of \cref{thm:rep}]
	First, we show the proof for the case that $\delta= 0$. Later it is shown how to use \cref{lemma:closeToEpsZero} in order to reduce the general case to this one.
	
	For $\delta =0$, fix $b \in \zo$, and let $\cA \subseteq \Uni^\ell$ be some set. It suffices to show that $\pr{D^\ell_b \in \cA} \leq e^{\eta(\eps,\ell,\delta')}\cdot\pr{D^\ell_{1-b} \in \cA}+\delta'$.
	
	Consider the set $\cS:= \set{y  \mid \log\frac{D^\ell_b(y)}{D^\ell_{1-b}(y)} \geq \eta(\eps,\ell,\delta')}$. It holds that: \begin{align}
	\pr{D^\ell_b \in \cA} &=  \pr{D^\ell_b \in \cA \setminus \cS}+ \pr{D^\ell_b \in \cA \cap \cS}\\\nonumber
	& \leq e^{\eta(\eps,\ell,\delta')}\cdot\pr{D^\ell_{1-b} \in \cA \setminus \cS} +  \pr{D^\ell_b \in \cA \cap \cS}  \\\nonumber
	&\leq e^{\eta(\eps,\ell,\delta')} \cdot\pr{D^\ell_{1-b} \in \cA} +  \pr{D^\ell_b \in \cS}.
	\end{align}
	It therefore enough to show that $\pr{D^\ell_b \in \cS} \leq \delta'$.
	 For this goal, consider the random variable $\widehat{A} = \log \frac{\ppr{D^\ell_b}{X_1,...,X_\ell}}{\ppr{D^\ell_{1-b}}{X_1,...,X_\ell}}$, where $X_1,\dots, X_\ell$ are independent samples from $D_b$. Let $A_i := \log \frac{\ppr{D_b}{X_i}}{\ppr{D_{1-b}}{X_i}}$. Then it holds that $\widehat{A} = \sum_{i=1}^\ell A_i$, where for every $i$, $A_i \in [-\eps, \eps]$. By \cref{lemma:avgLog}, it holds that for every $i$, $E[A_i]\leq \eps \cdot (e^\eps - 1)$. Therefore, by the Hoeffding bound, 
	 \begin{align}
	 \pr{D^\ell_b \in \cS} \leq \pr{\widehat{A}\geq \ell \cdot \epsilon(e^\epsilon -1)+\epsilon \cdot \sqrt{2
	 		\ell \cdot \ln(1/\delta')}}\leq e^{-(\epsilon \cdot \sqrt{2
	 		\ell \cdot \ln(1/\delta')})^2/2\ell \eps^2}=\delta'.
	 \end{align}
	
	 In the general case, for $\delta > 0$, let $D'_b$ be the distribution promised in \cref{lemma:closeToEpsZero}. By applying the above on $D'_b,D_{1-b}$, we get that for every set $\cA \subseteq \Uni^\ell$, it holds that
	 \begin{align}
	 \ppr{D'^\ell_b}{\cA} \leq e^{\eta(\eps,\ell,\delta')}\ppr{D'^\ell_{1-b}}{\cA}+\delta'.
	 \end{align}
	 Using the triangle inequality for statistical distance, it follows that:
		 \begin{align}
	\ppr{D^\ell_b}{\cA} \leq \ppr{D'^\ell_b}{\cA} + \ell \delta \leq e^{\eta(\eps,\ell,\delta')}\ppr{D^\ell_{1-b}}{\cA}+\delta'+ \ell \delta.
	\end{align}
\end{proof}

\paragraph{Computational indistinguishability.}


\begin{definition}[Computational indistinguishability]
Two distribution ensembles $X=\set{X_\pk}_{\pk \in \N}$, $Y=\set{Y_\pk}_{\pk \in \N}$ are {\sf [\resp non-uniformly] computationally indistinguishable}, denoted $X \ucindist Y$ [\resp $X \cindist Y$] if for every  \ppt [\resp \nuppt]  $\Dc$:
\[ |\Pr[\Dc(1^\pk,X_\pk)=1]-\Pr[\Dc(1^\pk,Y_\pk)=1]| \le \negl(\pk) .\]


\end{definition}

\subsection{Protocols}

Let $\pi=(\Ac,\Bc)$ be a two-party protocol. Protocol  $\pi$ is \ppt if both $\Ac$ and $\Bc$ running time is polynomial in their input length. We denote by $(\Ac(x_\Ac),\Bc(x_\Bc))(z)$ a random execution of $\pi$ with private inputs $(x_\Ac,y_\Ac)$, and common input $z$. At the end of such an execution, party $\Pc \in \set{\Ac,\Bc}$ obtains his view $V^\Pc(x_\Ac,x_\Bc,z)$, which may also contain a ``designated output'' $O^\Pc(x_\Ac,x_\Bc,z)$ (if the protocol specifies such an output).  A protocol has Boolean output, if each party outputs a bit.

\subsection{Two-Output Functionalities and  Channels}
A two-output \emph{functionality} is just  a random function  that outputs  a tuple of two values  in a predefined domain.  In the following we  omit the   two-output term from the notation.

\paragraph{Channels.}
A \emph{channel} is simply a no-input  functionality  with designated output bits.  We naturally identify channels with the random variable characterizes their output.

\begin{definition}[Channels]\label{def:channel} A {\sf channel} is a no-input Boolean  functionality whose output pair is of the from  $((\VA,\OA),(\VB,\OB))$ and for both $\Pc\in \set{\Ac,\Bc}$,  $\OP$ is Boolean and determined by $\VP$.  A channel has  {\sf agreement $\alpha$}  if $\pr{\OA = \OB} = \half + \alpha$. A channel ensemble  $\set{C_\pk}_{\pk \in \N}$  has agreement $\alpha$ if $C_\pk$ has agreement $\alpha(\pk)$ for every $\pk$.
\end{definition}
It is  convenient  to view a  channel as the experiment in which there are two parties $\Ac$ and $\Bc$.  Party $\Ac$ receives ``output'' $\OA$ and ``view'' $\VA$, and party $\Bc$ receives ``output'' $\OB$ and ``view'' $\VB$.

We  identify a no-input Boolean  output protocol with the channel ``induced'' by its semi-honest execution.
\begin{definition}[The protocol's channel]\label{def:PrtoChannel}
For a no-input Boolean  output protocol $\pi$, we define the channel $\CHAN(\pi)$ by    $\CHAN(\pi) = ((\VA,\OA),(\VB,\OB))$, for $\VP$ and $\OP$ being the view and output of party $\Pc$ in a random execution of $\pi$. Similarly, for protocol $\pi$ whose only input is a security parameter, let $\CHAN(\pi)  = \set{\CHAN(\pi)_\pk =\CHAN({\pi(1^\pk)})}_{\pk \in \N}$.
\end{definition}

All  protocols we construct in this work are \emph{oblivious}, in the sense that  given oracle access to a channel, the parties  only make use of the channel output (though the channel's view becomes part of the party view).\footnote{This is in accordance with definition of channels in the literature  in which the view component of the channel is only accessible to the eavesdropper (and not to the honest parties using the channel).}




\subsection{Secure Computation}
We use the standard notion of securely computing a functionality, \cf  \cite{Goldreich04}.
\begin{definition}[Secure computation]\label{def:SFE}
	A two-party protocol {\sf securely computes a functionality $f$}, if it does so according to the real/ideal paradigm.   We add the term perfectly/statistically/computationally/non-uniform computationally, if the the simulator output is    perfect/statistical/computationally indistinguishable/  non-uniformly indistinguishable from  the real distribution.  The protocol have the above notions of security {\sf against semi-honest  adversaries}, if its security only  guaranteed to holds against an adversary that follows the prescribed protocol.   Finally, for the case of perfectly secure computation, we naturally apply the above notion also to the non-asymptotic case: the protocol with no security parameter perfectly  compute a functionality $f$.
	
	A two-party protocol {\sf securely computes a functionality ensemble $f$ in the $g$-hybrid model}, if it does so according to the above definition when the parties have access to a trusted party computing $g$. All the above adjectives naturally extend to this setting.

\end{definition}

\subsection{Oblivious Transfer}
The (one-out-of-two) oblivious transfer functionality is defined as follows.
\begin{definition}[oblivious transfer functionality $f_\OT$]\label{def:OTfunc}
	The oblivious transfer functionality over $\zo \times (\zs)^2$ is defined by  $f_\OT (i,(\sigma_0,\sigma_1)) = (\perp,\sigma_i)$.
\end{definition}
A protocol is $\ast$ secure OT,   for \\$\ast\in \set{\text{semi-honest statistically/computationally/computationally non-uniform}}$, if it  compute the $f_\OT$  functionality with $\ast$ security.

\subsection{Two-Party  Differential Privacy}\label{sec:DP}

We consider differential privacy in the 2-party setting.

\begin{definition}[Differentially private functionality]\label{def:DP:IT}
	A functionality $f$ over input domain $\zn\times \zn$ is $\eps$-\DP, if the following holds: let $(\VA_{x,y},\VB_{x,y}) = f(x,y)$, then  for every $x,x'$ with $\Ham(x,x') =1$, $y\in \zn$ and $v\in \Supp(\VB_{x,y})$:
	
	$$\pr{\VB_{x,y}= v}  \leq   e^ \eps \cdot \pr{\VB_{x',y} = v},$$
	and the for  every $y,y'$ with $\Ham(y,y') =1$, $x\in \zn$ and $v\in \Supp(\VA_{x,y})$:
	$$\pr{\VA_{x,y}=v}  \leq   e^ \eps \cdot \pr{\VA_{x,y'} = v}.$$
\end{definition}
Note that the above definition is equivalence to  asking that $\VB_{x,y} \rindist_\eps \VB_{x',y}$ for any $x,x'$ with $\Ham(x,x') =1$ and $y$, and analogously for the view of $\Ac$, for $\rindist_{\eps}$ being the log-ratio according to \cref{def:Log-Ratio}.

We also remark that a more general definition allows also an additive error $\delta$   in the above, making the functionality $(\eps,\delta)$-\DP. However, for the sake simplicity, we focus on the simpler notion of $\eps$-\DP stated above.

\begin{definition}[Differentially private computation]\label{def:DP:C}
	A \ppt two-output protocol $\pi = (\Ac,\Bc)$ over input domain $\zn\times \zn$ is {$\eps$-\CDP}  if the following holds for every  $\nuppt$  $\Bc^\ast$, $\Dc$ and  $x,x' \in \zn$ with $\Ham(x,x') =1$:   let ${\VB}^\ast_{x}$ be the view of $\Bc^\ast$ in a random execution of $ (\Ac(x),\Bc^\ast)(1^\pk)$, then
	$$\pr{\Dc({\VB}^\ast_{x}) = 1}  \leq   e^{\eps(\pk)} \cdot \pr{\Dc({\VB}^\ast_{x'}) = 1} + \negl(\pk),$$
	and the same hold for the secrecy of $\Bc$.
	
	Such a protocol is {\sf semi-honest $\eps$-\CDP}, if the above is only guaranteed to hold for semi-honest adversaries (\ie for $\Bc^\ast = \Bc$).
\end{definition}

\subsection{Passive Weak Binary Symmetric Channels}
We rely on the work of  \citet{wullschleger2009oblivious} that shows that certain channels imply oblivious transfer. The following notion, adjusted to our formulation,  of a ``Passive weak binary symmetric channel'' was studied in \cite{wullschleger2009oblivious}.

\begin{definition}[Passive weak binary symmetric channels, \WBSC, \cite{wullschleger2009oblivious}]\label{dfn:WSBC}
	An $(\mu,\epsilon_0,\epsilon_1,p,q)$-\WBSC is a channel
	$C = ((\VA,\OA),(\VB,\OB))$ such that the following holds:
	\begin{itemize}
		\item Correctness: $\pr{\OA=0} \in [\half -\mu/2, \half +\mu/2]$
		
		and for every $b_\Ac \in \zo$, $\pr{\OB\neq\OA \mid \OA=b_\Ac}\in [\epsilon_0,\epsilon_1]$.

		\item Receiver security: $(\VA,\OA)|_{\OB=\OA} \sindist_p (\VA,\OA)|_{\OB\neq\OA}$.\footnote{In the requirement above, one can replace $(\VA,\OA)$ with $\VA$ (as by our conventions the latter determines the former). We remark that \cite{wullschleger2009oblivious} does not use this convention, and this is why we explicitly include the random variable $\OA$.}
		
\remove{		
		\item Receiver security: $\VA|_{\OB=\OA} \sindist_p \VA|_{\OB\neq\OA}$.

\Rnote{Why are there two different statements of Receiver security?}

}		
		\item Sender security: for every $b_\Bc \in \set{0,1}$, $\VB|_{\OB=b_\Bc, \OA=0} \sindist_q \VB|_{\OB=b_\Bc,\OA=1}$.
	\end{itemize}
\end{definition}

The following was proven in \cite{wullschleger2009oblivious}.

\begin{theorem}[WBSC implies oblivious transfer]
	\label{thm:wullschleger}
	There exist a protocol $\Delta$ such that the following holds. Let $\eps,\eps_0 \in (0,1/2), p \in(0,1)$ be  such that $150(1-(1-p)^2)<(1-\frac{2\eps^2}{\eps^2+(1-\eps)^2})^2$, and $\eps_0 \leq \eps$. Let $C$ be a $(0,\epsilon_0,\epsilon_0,p,p)$-\WBSC. Then  $\Delta(1^\pk,\eps)$  is a semi-honest   statistically secure OT in the $C$-hybrid model, and its running time is polynomial in $\pk$, $1/\eps$  and  $1/(1-2\eps)$. Furthermore, the parties in $\Delta$ only makes use of the output bits of the channel.

\end{theorem}

\cref{thm:wullschleger} considers channels with $\mu=0$, and $\eps_0=\eps_1$. This is equivalent to saying that the channel is balanced (\ie each of the output bits is uniform) and has $\alpha$-agreement, for $\alpha=\half-\eps_0$.
When stated in this form, \cref{thm:wullschleger} says that such a channel implies OT if $p=O(\alpha^2)$, and in particular, it is required that $p<\alpha$.

\subsubsection{Specialized  Passive Weak Binary Symmetric Channels}

We will be interested in a specific choice of parameters for passive WBSC's, and for this choice, it will be more convenient to work with the following stronger notion of a channel (that is easier to state and argue about, as security is defined in the same terms for both parties).

\begin{definition}[Specialized passive weak binary symmetric channels] \label{def:SWBSC}
	An $(\epsilon_0,p)$-\SWBSC is a channel $C = ((\VA,\OA),(\VB,\OB))$ such that the following holds:
	\begin{itemize}
		\item Correctness: $\pr{\OA=0}=\half$, and for every $b_\Ac \in \set{0,1}$, \\$\pr{\OB\neq\OA \mid \OA=b_\Ac} = \eps_0$.
		
		\item Receiver security: $\VA|_{\OA=\OB} \sindist_p \VA|_{\OA\neq \OB}$.
		\item Sender security:  $\VB|_{\OB=\OA} \sindist_p \VB|_{\OB\neq \OA}$.
	\end{itemize}
\end{definition}

\def\PropWBSC
{	
	An $(\eps_0,p)$-\SWBSC is  a $(0,\eps_0,\eps_0,2p,2p)$-\WBSC.	
}

\begin{proposition}\label{clm:spcl->WBSC}~
	\PropWBSC
\end{proposition}
The proof for \cref{clm:spcl->WBSC} appears in \cref{sec:appendix}.

\subsection{Additional Inequalities}

The following fact is proven in \cref{sec:appendix}.
\def\PropSimple
{
	The following holds for  every  $b\in(0,1/2)$ and $\ell \in \N$ such that $b\ell < 1/4$. 	$$\frac{(1/2+b)^\ell}{(1/2+b)^\ell+(1/2-b)^\ell} \in [\half(1+b\ell),\half(1+3b\ell) ].$$
}

\begin{proposition}\label{Prop:simple}
	\PropSimple
\end{proposition}

%% file: OtAmp.tex
\newcommand{\amin}{\alpha_{\min}}
\newcommand{\amax}{\alpha_{\max}}
\newcommand{\UsedC}{\myOptName{Mng}}
\newcommand{\teps}{\widetilde{\eps}}
\newcommand{\tdelta}{\widetilde{\delta}}
\newcommand{\talph}{\widetilde{\alpha}}
\newcommand{\talpho}{\widetilde{\alpha}}

\section{Amplification of Channels with Small Log-Ratio  Leakage}\label{sec:ChannelsAmp}	

In this section we formally define log-ratio leakage and prove our amplification results. We start in \Cref{sec:ChannelsAmp:IT} with the information theoretic setting, in which we restate and prove \cref{thm:Amp:IT:Inf} and \Cref{thm:triviality:IT:Inf}. In the full version of this paper we extend our result to the computational setting, restating and proving \cref{thm:Amp:C:Inf}.

\subsection{The Information Theoretic Setting}\label{sec:ChannelsAmp:IT}

We start with a definition of log-ratio leakage (restating \cref{def:RelLeakage:IT:Inf} with more formal notation).

\begin{definition}[Log-ratio  leakage]\label{def:LRLeakage:IT}
A channel $((\OA,\VA),(\OB,\VB))$ has {\sf $(\eps,\delta)$-\Leak} if
\begin{itemize}
	\item Receiver security:  $\VA|_{\OA=\OB} \rindist_{\eps,\delta} \VA|_{\OA \ne \OB}$.
	\item Sender security:  $\VB|_{\OA=\OB} \rindist_{\eps,\delta} \VB|_{\OA \ne \OB}$.
\end{itemize}
\end{definition}

The following theorem is a formal restatement of \cref{thm:Amp:IT:Inf}

\begin{theorem}[Small log-ratio  leakage implies OT]\label{thm:main:IT}
	There exists an (oblivious)  \ppt protocol $\Delta$ and constant $c_1>0$ such that the following holds. Let   $\eps,\delta\in [0,1]$  be such that   $\delta \le\eps^2$,  and let $\alpha \le \amax <1/8$ be such that    $\alpha \ge \max\set{c_1 \cdot \eps^2,\amax/2}$. Then for any channel  $C$ with  $(\eps,\delta)$-\Leak and $\alpha$-agreement,  protocol  $\Delta^C(1^\pk,1^{\floor {1/\amax}})$ is a  semi-honest   statistically secure OT in the $C$-hybrid model. 
\end{theorem}
%
%

Before proving  \cref{thm:main:IT}, we first show that it is tight. The proof of the following theorem is given in the full paper.

\begin{theorem}[Triviality of channels with large leakage]\label{thm:triviality:IT}
There exists a constant $c_2>0$, such that for every $\eps>0$ there is a two-party protocol (with no inputs) where at the end of the protocol, every party $\Pc \in \ABc$ has output $\OP$ and view $\VP$. Moreover, the induced channel $C=((\VA,\OA),(\VB,\OB))$ has $\alpha$-agreement,and $(\eps,0)$-leakage, for $\alpha \ge c_2 \cdot \eps^2$.
\end{theorem}

Together, the two theorems show that if $\alpha \ge c_1 \cdot \eps^2$ then the channel yields OT, and if $\alpha \le c_2 \cdot \eps^2$ then such a channel can be simulated by a two-party protocol with no inputs (and thus cannot yield OT with information theoretic security).

The proof of \cref{thm:main:IT} is an immediate consequence of the following two  lemmata.

Recall (\cref{def:PrtoChannel}) that $\CHAN(\Pi)$ denotes the channel induced by a random execution of the no-input, Boolean output protocol  $\pi$.


\begin{lemma}[Gap amplification]\label{lemma:main:GA}
	There exists an (oblivious)  \ppt  protocol $\Delta$ and constant  $c_1>0$ such that the following holds. Let $\eps,\delta,\alpha,\amax$ be parameters satisfying requirements in  \cref{thm:main:IT} \wrt $c_1$.  Let  $C$ be a channel with  $(\eps,\delta)$-\Leak and $\alpha$-agreement, let $\ell = 2^{(\floor{\log 1/\amax}-2)}$ and let  $\tC=\CHAN(\Delta^C(1^\ell))$. Then
	\begin{itemize}
		\item $\tC$  has $\talpho\in [1/32,3/8]$-\Agr.
		\item For  any  $\delta'\in(0,1)$: $\tC$  has  $(\teps,\tdelta)$-\Leak  for $\teps=2\ell \epsilon^2+\epsilon\sqrt{2\ell\ln({1/\delta'})}$ and  $\tdelta= \delta'+\ell\delta$.
	\end{itemize}
\end{lemma}

\newcommand{\bound}{\MathAlgX{bound}}
\begin{definition}[Bounded execution]\label{def:boundedExecution}
	Given  Boolean output protocol $\pi$ and $n\in \N$, let $\bound_n(\pi)$ be the variant of $\pi$ that if the protocol does not halt after $n$ steps,  it halts and the parties  output uniform independent bits.
	
\end{definition}

\begin{lemma}[Large Gap to OT]\label{lemma:main:LaegeGapToOT}
	There exist  an (oblivious) \ppt protocol $\Delta$ and  constants $n,c > 0$ such that the following holds: let  $\pi$ be a protocol of expected running time at most $t$ that induces a channel $C$  with $\alpha\in [1/32,3/8]$-agreement,  and $(\eps,\delta)$-\Leak for $\eps,\delta \leq c$.
	
	Then $\Delta^{C'}(1^\pk)$ is a semi-honest statistically secure OT  in the $C' = \\ \CHAN(\bound_{n\cdot t}(\pi))$ hybrid model.  
\end{lemma}

We prove  the above two Lemmas in  the following subsections, but first we will prove \cref{thm:main:IT}.
\begin{proof}[Proof of  \cref{thm:main:IT}]
Let $\ell = 2^{(\floor{\log 1/\amax}-2)}$. By \cref{lemma:main:GA}, there exists an expected polynomially time  protocol $\Lambda$ such that $\Lambda^C(1^\ell)$   induces a channel $\tC$ of  $\talph\in [1/32,3/8]$-\Agr, and $(\teps,\tdelta)$-\Leak  for $\teps=2\ell \epsilon^2+\epsilon\sqrt{2\ell\ln({1/\delta'})}$ and  $\tdelta= \delta'+\ell\delta$,  for any  $\delta'\in(0,1)$.  

Let $t \in \poly$ be a polynomial that bounds the expected running time of $\Lambda$. By \cref{lemma:main:LaegeGapToOT}, there exist  universal  constants $n,c$ and \ppt protocol $\Delta$, such that if
\begin{align}\label{eq:main:IT}
&\teps=2\ell \epsilon^2+\epsilon\sqrt{2\ell\ln({1/\delta'})}\le c &\text{and}&    & \tdelta= \delta'+\ell\delta \le c
\end{align}
then the  protocol $\Gamma$, defined by $\Gamma^C(1^\pk, 1^{\floor{1/\amax}}) = \Delta^{C'}(1^\pk)$ for \\$C'= \CHAN(\bound_{n\cdot t(\ell)}(\Lambda^\C(1^\ell)))$,  is a semi-honest statistically secure OT. Hence, we conclude the proof noting that \cref{eq:main:IT} holds by setting  $\delta'=\ell\delta$  and choosing $c_1$ (the constant in \cref{thm:main:IT})  to be sufficiently large.



\end{proof}

\cref{lemma:main:LaegeGapToOT} is proved in \cref{sec:ChannelsAmp:LargetoOT} using the amplification result of
\cite{wullschleger2009oblivious}.  Toward  proving \cref{lemma:main:GA},  our main technical contribution, we start  in \cref{sec:ChannelsAmp:Inefficient} by presenting an inefficient protocol implementing the desired channel. In \cref{sec:ChannelsAmp:Efficient} we show how to bootstrap the  the above protocol into an efficient one.

\subsubsection{Inefficient Amplification}\label{sec:ChannelsAmp:Inefficient}	
The following protocol implements the channel stated in \cref{lemma:main:GA}, but its running time is \emph{exponential} in $1/\amax$.
\begin{protocol}\label{proto:OT:IT}[Protocol $\Delta^C= (\tAc,\tBc)$] 
    \item[Oracle:]  channel $C=((\VA,\OA),(\VB,\OB))$.
	
	\item[Input:] $1^\ell$.
	\item[Operation:] The parties repeat the following process until it produces outputs:
		\begin{enumerate}
		\item The parties (jointly) call the channel $C$ for $\ell$ times. Let $\ooA=(o^{\Ac}_1,\ldots,o^{\Ac}_\ell),\ooB= (o^{\Bc}_1,...,o^{\Bc}_\ell)$ be the outputs. 

		\item $\tAc$ computes and sends $\mathcal{S}=\set{\ooA, 1^\ell\xor\ooA}$ according to their lexical order to $\tBc$.

		\item $\tBc$ inform $\tAc$ whether $\obar^{\Bc}\in \mathcal{S}$.
		
		If positive,  both  parties output the index of their tuple in $\mathcal{S}$ (and the protocol ends).
	\end{enumerate}
\end{protocol}


We show that the channel induced by protocol $\Delta^C(1^\ell)$ satisfies all the requirement of \cref{lemma:main:GA} apart from its expected running time (which is  exponential in $\ell$).

Let $\tC= \CHAN(\Delta^C(\ell))=((\tVA,\tOA),((\tVB,\tOB))$. The following function outputs the calls to $C$ made in the final iteration in $\tC$.

\begin{definition}[Final calls]\label{not:ChannelsAmp:Inefficient}
	   For  $c\in \Supp(\tC)$ let $\Fcalls(c)$ denote the output of the $\ell$ calls to $C$ made in the {\sf final} iteration in $c$.
\end{definition}
We make the  following observation about the final calls.
\begin{claim}\label{clm:ChannelsAmp:Inefficient}
	 The following holds for  $((\cdot,\oOA),(\cdot,\oOB)) = \Fcalls(\tC= ((\cdot,\tOA),(\cdot,\tOB)))$.
	\begin{itemize}
		\item $\tOA=\tOB$ iff $\oOA=\oOB$.
		\item 	Let $C^\ell = ((\cdot,(\OA)^\ell), (\cdot,(\OB)^\ell))$ be the random variable induced by  taking $\ell$ copies of $C$ and let $E$ be the event that $(\OB)^\ell \in \set{( \OA)^\ell, (\OA)^\ell \xor 1^\ell  }$.  Then
		$\Fcalls(\tC) \equiv C^\ell|_E$.
	\end{itemize}

\end{claim}
\begin{proof}
	Immediate by construction.
\end{proof}

%

\paragraph{Agreement.}
\begin{claim}[Agreement]\label{claim:agreement:IT}
$ \pr{\tOA=\tOB} \in [17/32 , 7/8]$.
\end{claim}
\begin{proof}
 By \cref{clm:ChannelsAmp:Inefficient},
	\begin{align}\label{eq:exact-agree}
	\pr{\tOA=\tOB}&= \frac{\pr{ (\OA)^\ell =(\OB)^\ell \mid E}}{\pr{ (\OA)^\ell =(\OB)^\ell \mid E}+\pr{ (\OA)^\ell  \oplus(\OB)^\ell = 1^\ell \mid E}}\\
	&= \frac{\pr{ (\OA)^\ell =(\OB)^\ell}}{\pr{ (\OA)^\ell =(\OB)^\ell}+\pr{ (\OA)^\ell  \oplus(\OB)^\ell = 1^\ell}}\nonumber\\
	& =\frac{(1/2+\alpha)^\ell}{(1/2+\alpha)^\ell+(1/2-\alpha)^\ell}.\nonumber
	\end{align}
	Since, $\ell = 2^{(\floor{\log 1/\amax}-2)}$  and  $\amax/2 \le \alpha \le \amax$, we get that $1/4 \geq \ell\cdot \alpha \geq 1/16$. By \cref{Prop:simple},
	\begin{align}
	\frac{(1/2+\alpha)^\ell}{(1/2+\alpha)^\ell+(1/2-\alpha)^\ell} \in [\half(1+\alpha\ell),\half(1+3\alpha\ell) ]
	\end{align}
Thus, $\pr{\tOA=\tOB} \in [17/32 , 7/8]$, which concludes the proof.
\end{proof}

\newcommand{\tv}{\widetilde{v}}
\paragraph{Leakage.}
\begin{claim}[Leakage]\label{claim:Leak:IT}
$\tC$ has $(\teps,\tdelta)$-\Leak, where $\teps=2\ell\epsilon^2+\epsilon\sqrt{2\ell\ln({1/\delta'})}$ and  $\tdelta= \delta'+\ell\delta$ for every $\delta'\in(0,1)$.
\end{claim}
\begin{proof}
We need to prove that for both $\Pc\in \ABc$:
\begin{align}\label{eq:Leak:IT:1}
\tVP|_{\tOA =\tOB} \rindist_{(\teps,\tdelta)} \tVP|_{\tOA \neq \tOB}
\end{align}

By assumption $C$ has $(\eps,\delta)$-\Leak. Thus, by \cref{thm:rep},
\begin{align}\label{claim:Leak for cond}
(\VP)^\ell|_{( \OA)^\ell =  (\OB)^\ell}  \rindist_{(\teps,\tdelta)}   (\VP)^\ell|_{( \OA)^\ell =  (\OB)^\ell \xor 1^\ell}
\end{align}
Let $((\oVA,\oOA),(\oVB,\oOB) = \Fcalls(\tC)$. By the above and \cref{clm:ChannelsAmp:Inefficient},
\begin{align}
\oVP|_{\tOA =\tOB} \rindist_{(\teps,\tdelta)} \oVP|_{\tOA \neq \tOB}
\end{align}
 \cref{eq:Leak:IT:1} now follows by a data processing  argument: let $f$ be the randomized function that on input $v\in\Supp(\oVP)$ outputs a random sample from  $\tVP|_{\oVP = v}$.  It easy to verify that $f(\oVP|_{\tOA =\tOB})=\tVP|_{\tOA =\tOB}$ and $f(\oVP|_{\tOA \ne\tOB})\equiv\tVP|_{\tOA \ne\tOB}$. Thus \cref{eq:Leak:IT:1} follows by  \cref{fact:LR:DP}.
\end{proof}

\remove{
\begin{notation}\label{not:condi}
	Given a channel $C=((\VA,\OA)(\VB,\OB))$ and $\ell \in\N$, let $C^\ell=(C_1,C_2,\cdots C_\ell)$ denote  $\ell$ independent samples from $C$. For every $\p\in\set{\Ac,\Bc}$ let $C^{\p,\ell}=(C^\p_1,C^\p_2,\cdots C^\p_\ell)$ denote $C^\ell$ restricted to $\p$.
\end{notation}

\begin{claim}\label{claim: conditinal}
	Gi
\end{claim}
\begin{proof}
	Recall that $\tVP$, is composed of failed iterations and  a uniform bit $W$ (or $Z$ in the case of party $\tBc$ ). Note that since we condition on the event $\overline{C}^\p_F=u^\p$, the final iteration $F$ is thus decided and it holds that  $\oOA_F \oplus \oOB_F \in \set{0^\ell,1^\ell}$. All previous iterations are discarded, and are independent form the values of $\tOA$ and $\tOB$. Since $W_F$ was uniformly sampled after $\oOA_F$ and  $\oOB_F$ where fixed, it holds that $W_F$  is still an independent uniform bit after conditioning on the events $\set{\tOA=\tOB}$ and $\set{\tOA\ne\tOB}$. This concludes the proof.
\end{proof}


\begin{claim}\label{claim:used:IT}
	For every $u^\p\in\Supp(\overline{C}^\p_F)$:
$\tVP|_{\tOA=\OB,\overline{C}^\p_F=u^\p}\equiv\tVP |_{\tOA\ne \tOB,\overline{C}^\p_F=u^\p}$
\end{claim}
\begin{proof}
Recall that $\tVP$, is composed of failed iterations and  a uniform bit $W$ (or $Z$ in the case of party $\tBc$ ). Note that since we condition on the event $\overline{C}^\p_F=u^\p$, the final iteration $F$ is thus decided and it holds that  $\oOA_F \oplus \oOB_F \in \set{0^\ell,1^\ell}$. All previous iterations are discarded, and are independent form the values of $\tOA$ and $\tOB$. Since $W_F$ was uniformly sampled after $\oOA_F$ and  $\oOB_F$ where fixed, it holds that $W_F$  is still an independent uniform bit after conditioning on the events $\set{\tOA=\tOB}$ and $\set{\tOA\ne\tOB}$. This concludes the proof.
\end{proof}

It follows that we  need to consider the amount  of information $\overline{C}^\p$ gives about the protocol's outputs.

\begin{claim}\label{claim:security:IT}
	For both $\p\in\ABc$: 	$(\overline{C}^\Pc_F,\tOP) |_{\tOA=\tOB}\rindist_{(\teps,\tdelta)} (\overline{C}^\Pc_F,\tOP) |_{\tOA\neq \tOB}$.
\end{claim}
\begin{proof}
	First note that for both  $\p\in\ABc$,  $\tOP$ is independent from $\overline{C}^\Pc$
	thus it suffices to prove that,
	$\overline{C}^\Pc_F|_{\tOA=\tOB}\rindist_{(\teps,\tdelta)} \overline{C}^\Pc_F, |_{\tOA\neq \tOB}$.
	 By construction,  $\overline{C}^\Pc_F |_{\tOA=\tOB}$ consists of $\ell$ independent samples from the distribution $\VP|_{\OA=\OB}$ (for $\VP$ being the view given by the channel $C$).
	
	Similarly $\overline{C}^\Pc_F |_{\tOA\neq \tOB}$, consists of $\ell$ independent samples from the distribution $\VP|_{\OA\neq \OB}$.
	
	Since, by assumption,  $\VP|_{\OA=\OB} \rindist_{\eps,\delta} \VP|_{\OA \ne \OB}$,  \cref{thm:rep} yields that
	\[\tVP|_{\tOA=\tOB}  \rindist_{(\teps,\tdelta)} \tVP|_{\tOA\neq \tOB},\]
	for $\teps=\ell \cdot \epsilon(e^\epsilon-1)+\epsilon \cdot \sqrt{2
		\ell \cdot \ln(1/\delta')}$ and $\tdelta=\ell \delta + \delta'$ for any $\delta'\in(0,1)$.
\end{proof}

We now prove \cref{claim:Leak:IT} using the above two claims:

\begin{proof}[Proof of \cref{claim:Leak:IT}]
	We need to show that for every part $\tP\in\set{\tAc,\tBc}$:
	$\tVP|_{\tOA=\tOB}\rindist_{(\teps,\tdelta)}  \tVP |_{\tOA\ne \tOB}$ by \cref{fact:LR:Cond,claim:used:IT}, it suffices to show that $\overline{C}^\Pc_F |_{\tOA=\tOB}\rindist_{(\teps,\tdelta)} \overline{C}^\Pc_F|_{\tOA\neq \tOB}$.Thus by $\cref{claim:security:IT}$ this indeed holds.
\end{proof}

%
%
%
%

}

%% file: OTAmpTree.tex
\subsubsection{Efficient Amplification}\label{sec:ChannelsAmp:Efficient}	

We will show how to make \cref{proto:OT:IT} protocol more efficient in terms of $\alpha$. The resulting protocol will run in poly-time even if $\alpha$ is inverse polynomial.
The  efficient amplification protocol is defined as follows.
Let  $\Delta$ be the (inefficient)  protocol from   \cref{proto:OT:IT}. 
	
\begin{protocol}\label{proto:EffAmp}[Protocol $\Lambda^C= (\hAc,\hBc)$] 
	\item[Oracle:] 	Channel   $C$.
	\item[Prameter:] Recursion depth $d$.
	\item[Operation:] The parties interact in  $\Delta^{\Lambda^{C}(d-1)}(2)$, letting $\Lambda^{C}(0) = C$.~\\
\end{protocol}

 We show that the channel induced by protocol $\Lambda^C(d)$ satisfies all the requirement of \cref{lemma:main:GA}.  But we first show that the expected running time of $\Lambda^C(d)$ is $O(4^d)$, and therefore, the protocol that on input $1^\ell$ invoke $\Lambda^C(\log \ell)$, is \ppt, as stated in \cref{lemma:main:GA}.

\paragraph{Running time.}
\begin{claim}[Expected running time]\label{claim:RunEffAmp}
	Let $C$ be a channel, the for any $d\in \N$  the expected running time of $\Lambda^C(d)$ is at most $O(4^d)$.
\end{claim}
We will use the following claim:
\begin{claim}\label{claim:inefficient-time} 
	For any channel $C$, $\Delta^C(2)$ makes in expectation at most $4$ calls to $C$.
\end{claim}
\begin{proof}
	Let  $C$ with a channel with  agreement $\alpha\in[-1/2,1/2]$.  Let $\oOA=(\OA_1,\OA_2)$ and $\oOB=(\OB_1,\OB_2)$ denote the outputs of two invocations of $C$, respectively. By construction, $\Delta^C(2)$ concludes on the event  $E=\set{(\OB_1,\OB_2)\in\{\oOA,1^2\oplus\oOA\}}$. It is clear that  
	$\pr{E}=(\half+\alpha)^2+(\half-\alpha)^2=\half+\alpha^2\ge\half$.
	Thus, the  expected number of invocations preformed by  $\Delta^C(2)$  is bounded is $4$.  
\end{proof}	
We now prove $\cref{claim:RunEffAmp}$ using the above claim.
\begin{proof}[Proof of \cref{claim:RunEffAmp}]
	 For  $d\in\N$, let $T(d)$ denote the expected runtime of $\Lambda^C(d)$.  By   \cref{claim:inefficient-time},  
	\begin{align}
	T(d)=4\cdot{T(d-1)}+O(1),
	\end{align}
	letting  $T(0)=1$. Thus, $T(d)\in O(4^d)$.
\end{proof}

Let $\hC_d= \CHAN(\Lambda^C(d))= ((\hVA_d,\hOA_d),((\hVB_d,\hOB_d))$.  The following function outputs the   ``important' calls of $C$ made in $\hC_d$, the ones used to set the final outcome.

Let $\circ$ denote vectors concatenation.
\newcommand{\Icalls}{\operatorname{important}}
\begin{definition}[Important calls]\label{not:ChannelsAmp:Rfficient}
	 For  $d\in \N$ and  $c\in \Supp(\hC_d)$, let  $\Fcalls(c) = (c_0,c_1)$ be the two calls to $\Lambda^C(d-1)$ done in final execution of $\Delta^{\Lambda^C(d-1)}(2)$ in $c$.  Define 	 $\Icalls(c) = \Icalls(c_0) \circ \Icalls(c_1)$, letting $\Icalls(c) = c$ for $c\in  \Supp(\hC_0)$. 
	 
\end{definition}

Similarly to the analysis of inefficient protocol, the crux   is the  following observation about the important calls.

\begin{claim}\label{clm:ChannelsAmp:Efficient}
	Let $d\in\N$ and  set $\ell=2^d$.   The following holds for  $((\cdot,\oOA),(\cdot,\oOB)) = \Icalls(\hC_d= ((\cdot,\hOA),(\cdot,\hOB)))$.
	\begin{itemize}
		\item $\hOA=\hOB$ iff $\oOA=\oOB$.
		
		\item 	Let $C^\ell = ((\cdot,(\OA)^\ell), (\cdot,(\OB)^\ell))$ be the random variable induced by  taking $\ell$ copies of $C$ and let $E$ be the event that $(\OB)^\ell \in \set{( \OA)^\ell, (\OA)^\ell \xor 1^\ell  }$.  Then
		
		$\Icalls(\hC_d) \equiv C^\ell|_E$.
	\end{itemize}
\end{claim}

We prove  \cref{clm:ChannelsAmp:Efficient} below, but first use it for proving \cref{lemma:main:GA}.

\paragraph{Agreement.}
\begin{claim}[Agreement]\label{claim:agreement:IT:Eff}
	$ \pr{\hOA=\hOB} \in [17/32 , 7/8]$.
\end{claim}
\begin{proof}
	The proof follows by \cref{clm:ChannelsAmp:Efficient}, using the same lines  as the proof that  \cref{claim:agreement:IT} follows from \cref{clm:ChannelsAmp:Inefficient}.
\end{proof}

\paragraph{Leakage.}
\begin{claim}[Leakage]\label{claim:Leak:IT:Eff}
	$\hC$ has $(\teps,\tdelta)$-\Leak, where $\teps=2\ell\epsilon^2+\epsilon\sqrt{2\ell\ln({1/\delta'})}$ and  $\tdelta= \delta'+\ell\delta$ for every $\delta'\in(0,1)$.
\end{claim}
\begin{proof}
The proof follows by \cref{clm:ChannelsAmp:Efficient} and a data processing argument, using  similar lines  to the  proof that  \cref{claim:Leak:IT} follows from \cref{clm:ChannelsAmp:Inefficient}.
\end{proof}

\paragraph{Proving  \cref{lemma:main:GA}.}

\begin{proof}[Proof of \cref{lemma:main:GA}]
	Consider the protocol $\Tau^C(1^\ell) = \Lambda^C(\floor{\log \ell})$. The proof that $\Tau$ satisfies the requirements of \cref{lemma:main:GA} immediately follows  by Claims \ref{claim:RunEffAmp}, \ref{claim:agreement:IT:Eff} and \ref{claim:Leak:IT:Eff}. 
\end{proof}

\paragraph{Proving  \cref{clm:ChannelsAmp:Efficient}.}

\begin{proof}[Proof of \cref{clm:ChannelsAmp:Efficient}]
First note that the first item in the claim immediately follows	by construction. We now prove the second item.

Let  $d\in\N$ and  let  $\ell=2^d$. For $C^{\ell}= ((\cdot,(\OA)^\ell), (\cdot,(\OB)^\ell))$, let $D_\ell$ be  the distribution  of  $C^{\ell}|_{\set{(\OB)^\ell \in \set{( \OA)^\ell, (\OA)^\ell \xor 1^\ell }}}$.  We need to prove that
\begin{align*}
\Icalls(\hC_d) \equiv D_\ell
\end{align*}

We prove the claim by induction on $d$. The base case $d=1$ follows by \cref{clm:ChannelsAmp:Inefficient}.

 Fix $d>1$,  for $j\in\zo$,  let $\hC_{d-1,j}$ be an invocations of the channel on input $d-1$  and let $((\cdot,\oOA_j),(\cdot,\oOB_j)) = \Icalls(\hC_{d-1,j})$. By the induction hypothesis,
 \begin{align}
 \Icalls(\hC_{d-1,j})\equiv D_{\ell/2}
 \end{align}

 The key observation is that by construction, the event $\Fcalls(\hC_d)=  \hC_{d-1,0}\circ\hC_{d-1,1}$ occurs if and only if, 
\begin{align}
\oOB_0\circ \oOB_1\in\set{\oOA_0 \circ\oOA_1,1^\ell\oplus\oOA_0 \circ\oOA_1}
\end{align}

Recall this means that, 
$$\Icalls(\hC_d)=\big(\Icalls(\hC_{d-1,0}) \circ \Icalls(\hC_{d-1,1})\big)\mid_{\overline{E}}$$

where $\overline{E}={\set{{\oOB_0\circ \oOB_1\in\set{\oOA_0 \circ\oOA_1,1^\ell\oplus\oOA_0 \circ\oOA_1}}}}$. The above observations  yields that\\ $\Icalls(\hC_d) \equiv D_\ell$.
	\end{proof}

\remove{
We analyze the correctness and security of $\hC_d$ by relating it to the channel $\tC_{2^d} = ((\tVA,\tOA),(\tVB,\tOB)$ induced by a random execution of $\Delta^{C}(2^d)$.  

\paragraph{Agreement for efficient protocol.}
We next bound the  agreement of $\hC_d$ by showing that it exactly equals the agreement of $\tC_{2^d}$ (seen at $\cref{claim:agreement:IT}$).

\begin{claim}[Agreement]\label{claim:agreement}
For every $d\in\N$,	$\hC_d=((\hVA_d,\hOA_d),(\hVB_d,\hOB_d))$ has the same agreement as $\tC_{2^d}=((\tVA_{2^d},\tOA_{2^d}),(\tVB_{2^d},\tOB_{2^d}))$.
\end{claim}
\begin{proof}
	The proof follows by induction, we will show that for every $d\in\N$, $$\half+\halph_d=\pr{\hOA_d=\hOB_d}=\pr{\tOA_{2^d}=\tOB_{2^d}}=\frac{(1/2+\alpha)^{2^d}}{(1/2+\alpha)^{2^d}+(1/2-\alpha)^{2^d}}$$ (where $\alpha$ is the agreement of the channel $C$) see \cref{claim:agreement:IT} ,\cref{eq:exact-agree}.
	 For the bases case where $d=1$,  by construction $\hC_1\equiv\tC_2$ (it's the same protocol).
	 Assume that the hypothesis holds for some $d$, we will show that it also holds for $d+1$. Consider the channel $\hC_{d+1}$, by definition $\hC_{d+1}$ is the output we get by randomly executing $\Delta^{\hC_d}(2)$, thus it follows that,
	 $$\half+\halph_{d+1}=\frac{(1/2+\halph_d)^{2}}{(1/2+\halph_d)^{2}+(1/2-\halph_d)^{2}} $$
	 Since by assumption $\half+\halph_d=\frac{(1/2+\alpha)^{2^d}}{(1/2+\alpha)^{2^d}+(1/2-\alpha)^{2^d}}$, it follows that,
	 \begin{align}
	 \half+\halph_{d+1}=\frac{\Big(\frac{(1/2+\alpha)^{2^d}}{(1/2+\alpha)^{2^d}+(1/2-\alpha)^{2^d}}\Big)^{2}}
	 {\Big(\frac{(1/2+\alpha)^{2^d}}{(1/2+\alpha)^{2^d}+(1/2-\alpha)^{2^d}}\Big)^{2}
	 	+\Big(\frac{(1/2-\alpha)^{2^d}}{(1/2+\alpha)^{2^d}+(1/2-\alpha)^{2^d}}\Big)^{2}}
 	=\frac{(1/2+\alpha)^{2^{d+1}}}{(1/2+\alpha)^{2^{d+1}}+(1/2-\alpha)^{2^{d+1}}}
	 \end{align}
	 Thus the claim holds.
\end{proof}

 \paragraph{Leakage for efficient protocol.}
For  $d\in \N$   let $\hC_d = ((\hVA,\hOA),(\hVB,\hOB)$  be the channel induced by  a random execution of $\Lambda^C(d)$. We omit the subscript  $d$ from the notation unless when addressing the channel induced by  execution of $\Lambda^C$ with different values of $d$. 

We prove the following claim

\begin{claim}\label{claim:ChannelsAmp:EfficientLekage}
	$\hC$ has $(\heps,\hdelta)$-\Leak, where $\heps=2\cdot2^d\cdot\epsilon^2+\epsilon\sqrt{2\cdot2^d\ln({1/\delta'})}$ and  $\hdelta= \delta'+2\cdot2^d\delta$ for every $\delta'\in(0,1)$.
\end{claim}

The ``Meaningful'' samples of $C$ in  $\hC$  are those samples of $C$   whose value effected the  output of $\hC$.
\begin{definition}[Meaningful samples]\label{Def:mean}
We define $\UsedC(\tC_{2^d})$ as the samples of $C$ done  in the last iteration in $\tC_{2^d})$ (refereed as  $\overline{C}_F$ in \cref{not:OT:IT}). We define   $\UsedC(\hC_d) = \UsedC(V_0), \UsedC(V_1)$, where    $V_0$ and $V_1$ he joint views of the two invocations of $\Lambda^C(d)$ done in the final execution of $\Delta_2^{\Lambda^C(d-1)}$ in $\hC_d$.  We let $\UsedC(\hC_0) = \hC_0$.

For   $\p\in\set{\Ac,\Bc}$ and $u \in \Supp(\UsedC^\Pc(\cdot ))$, let $u^\Pc$ be as the part seen by party $\Pc$ in  $u$.

\end{definition}

 \newcommand{\tUA}{\widetilde{\UA}}

\begin{claim}\label{claim:used with out}
For every $d\in \N$:   	$(\hOA_d,\hOB_d,\UsedC(\hC_d))\equiv (\tOA_{2^d},\tOB_{2^d},\UsedC(\tC_{2^d}))$.

\end{claim}
\begin{proof}
%
	Fix $d\in\N$ and let $\ell=2^d$. For the channel $C=((\VA,\OA),(\VB,\OB))$, let 
	$\Vbar=((\VA_1,\VB_1)\cdots(\VA_\ell,\VB_\ell))$ and $\oOP=(\OP_1,\cdots\OP_\ell)$
	denote the views and outputs of $\ell$ independent samples of $C$. By the definition of $\Delta^C$, $\Vbar\in \Supp(\UsedC(\tC_\ell))$  iff $\oOA\oplus \oOB\in\set{1^\ell,0^\ell}$. Similarly, by the definition of $\Lambda^C$,  $\Vbar\in \Supp(\UsedC(\hC_d))$  iff $\oOA\oplus \oOB\in\set{1^\ell,0^\ell}$. For the event $E_b=\set{\oOA\oplus \oOB= b^\ell}$,
$$\UsedC(\tC_\ell)\mid_{E^b} \equiv  \Vbar^\ast\mid_{E^b} \equiv \UsedC(\hC_d)\mid_{E^b}$$
Since we do a rejective\Jnote{?} sampling that depends only on the relationship between $\OA$ and $\OB$. It holds that after conditioning on $E^b$  every index $i\in\ell$, is independent.

Thus it remains to show that for every $b\in\zo$, 
\begin{align}\label{eq:used=used}
\ppr{\UsedC(\tC_\ell)}{E^b}=\ppr{\UsedC(\hC_\ell)}{E^b}
\end{align}
But since $\ppr{\UsedC(\tC_\ell)}{E^0}=\pr{\tOA=\tOB}$ and  $\ppr{\UsedC(\hC_\ell)}{E^0}=\pr{\hOA=\hOB}$ by \cref{claim:agreement} , \cref{eq:used=used} holds, this concludes the proof.
%
%
%
%
\end{proof}

\begin{claim}\label{claim:leakage_with_used}
	For every    $u_d\in \Supp(\UsedC(\hC_d = (\hVA_d,\hOA_d,\hVB_d,\hOB_d)))$ and $\Pc \in \ABc$:
	$$\hVP_d|_{\hOA_d = \hOB_d,\UsedC(\hC_d)^\Pc=u_d^\Pc}\equiv\hVP_d |_{\hOA_d \ne \hOB_d,\UsedC(\hC_d)^\Pc=u_d^\Pc}.$$
\end{claim}
\begin{proof}

	The proof will follow by induction. We will prove the claim by induction over $d\in\N$.  
	For the bases case where $d=1$,  by construction $\hC_1\equiv\tC_2$ (it's the same protocol), and by \cref{claim:used:IT} our assumption holds.
	
	Assume that the hypothesis holds for $d$, 
	we will show that it also holds for $d+1$. Consider the channel $\hC_{d+1}$, by definition $\hC_{d+1}$ is the output we get by randomly executing $\Delta^{\hC_d}(2)$,
	for $i\in\zo$ let  $(\hVP_{d,1},\hVP_{d,2})$
denote the  view  of $\p$ in the two executions of $\hC_d$ done by $\Delta^{\hC_d}(2)$. 
From  \cref{claim:used:IT} we can deduce that for every $(v_1,v_2)\in\Supp(\hVP_{d,1},\hVP_{d,2})$, 
	$$\hVP_{d+1}\mid_{\hOA_{d+1} = \hOB_{d+1},(\hVP_{d,1},\hVP_{d,2})=(v_1,v_2)}\equiv\hVP_{d+1} \mid_{\hOA_{d+1} \ne \hOB_{d+1},(\hVP_{d,1},\hVP_{d,2})=(v_1,v_2)}$$
	By the hypothesis assumption,
	$$\hVP_{d+1}|_{\hOA_{d+1} = \hOB_{d+1},\UsedC(\hC_{d+1})^\Pc=u_{d+1}^\Pc}\equiv\hVP_{d+1} |_{\hOA_{d+1} \ne \hOB_d,\UsedC(\hC_d)^\Pc=u_d^\Pc}$$
	
		By \cref{Def:mean},
	$$\UsedC(\hC_{d+1}) = \UsedC(\hC_{d}),\UsedC(\hC_{d})$$ 
	
	the claim holds.
	\Jnote{As of now, the proof for this is ugly. I am considering taking Iftachs approach.}
\end{proof}

\

}

%% file: LargeGapToOT.tex
\subsubsection{From Channels with Large Gap to OT}\label{sec:ChannelsAmp:LargetoOT}	

\begin{definition}
A channel $C=((\VA,\OA),(\VB,\OB))$ is  {\sf balanced} if  $\pr{\OA = 1} = \pr{\OB = 1} = \half$.
\end{definition}

We use the following claim.
\begin{claim}\label{claim:ChannelsAmp:LG}	
	Let $C=((\VA,\OA),(\VB,\OB))$ be a balanced channel that has  $\alpha \in [\amin, \amax]$-agreement and $(\eps,\delta)$-\Leak. Then $C$ is a $(\eps_0,p)$-\SWBSC  for some $\eps_0 \in [ \half-\amax,\half-\amin]$, and $p=2\eps +\delta$. 
\end{claim}
\begin{proof}
	For every $\p\in\set{\Ac,\Bc}$ we have that, 
	$\tVP|_{\tOA=\tOB}  \rindist_{(\eps,\delta)} \tVP|_{\tOA\neq \tOB},$
	thus by definition it follows that,	$\tVP |_{\tOA=\tOB}\sindist_{(2\eps+\delta)} \tVP |_{\tOA\neq \tOB}$, and the claim holds. 
\end{proof}
The following claim, states that a given a channel with bounded leakage and agreement we can construct a new protocol using the olds one, that has the same leakage and agreement, while having the additional property of being balanced.

\begin{claim}\label{claim:balanced:LG}
	There exists a constant-time single oracle call protocol  $\Delta$ such that for every channel  $C$, the channel $\tC$ induced by $\Delta^C$ is balanced and has the same agreement and leakage as of $C$.  
\end{claim}
\begin{protocol}\label{proto:balanced:LG}[Protocol $\Delta= (\tAc,\tBc)$] 
	\item[Oracle:] 	Channel   $C$.
	\item[Operation:] ~
	\begin{enumerate} 
		\item The parties (jointly) call the channel $C$. Let $o^{\Ac}$ and $o^{\Bc}$  denote their  output respectively.

		\item $\tAc$ sends  $r \la \zo $ to $\tBc$ .
		\item $\tAc$ outputs $o^{\Ac} \xor r$  and $\tBc$ outputs $ o^{\Bc} \xor r$.
	\end{enumerate}
\end{protocol}
\begin{proof}[Proof of \cref{claim:balanced:LG}]
	Let $\tC=\CHAN(\Delta^C)$. By construction $\tC$  is balanced and has $\alpha$-agreement.   Finally, by a data processing argument,  $\tC$ has the same leakage as $C$.
\end{proof}

\paragraph{Proving   \cref{lemma:main:LaegeGapToOT}.}

\begin{proof}[Proof of    \cref{lemma:main:LaegeGapToOT}]
	Set $n = 10^8$, let $C = \CHAN(\pi)$, let $\Pi' = \bound_{n\cdot t}(\Pi)$  and let $C'= \CHAN(\Pi')$. By Markov inequality, 
 
	\begin{align}\label{eq:main:LaegeGapToOT:SD}
		C' \sindist_{1/n}  C
	\end{align}

	By \cref{claim:balanced:LG}, there exist a protocol $\Delta$ such that $\Delta^C$ is balanced and has the same leakage and agreement as $C$. Moreover, since $\Delta$ only uses one call to the channel $C$, by data processing argument,
		\begin{align}
	\CHAN(\Delta^{\C'}) \sindist_{1/n} \CHAN(\Delta^C)
	\end{align}

 By \cref{claim:balanced:LG}, $\Delta^{\C'}$ is also balanced. \cref{claim:ChannelsAmp:LG} yields that $\Delta^C$ is a $(15/32,p)$-\WBSC  for $p=2\eps+\delta$. Hence, using \cref{Prop:CondSD}, we get that $\Delta^{\C'}$ is $(\eps_0,\overline{p})$-\WBSC, for $\eps_0 = \epsilon+1/10^8$ and $\overline{p}=p+4/10^7$. 

	In the following we use \cref{thm:wullschleger} to show that $\Delta^{\C'}$ can be used to construct semi-honest statistically secure OT.  To do this, we need to prove that
	\begin{align}\label{eq:main:LaegeGapToOT}
	150(1-(1-2\overline{p})^2)<(1-\frac{2\eps_0^2}{\eps_0^2+(1-\eps_0)^2})^2
	\end{align}
	Indeed, since $(1-2\frac{\eps_0^2}{\eps_0^2+(1-\eps_0)^2})^2\geq 1/100$, for $\delta'=1/10^{7}$ it holds that, for small enough $c$,

	\begin{align}
		(1-(1-2\overline{p})^2) \leq 4\overline{p} & \leq 4p+2/10^6 \leq  2\eps + \delta +2/10^6\\
		&\leq 3c+2/10^6\nonumber\\
		&<1/(150\cdot 100).\nonumber
\end{align}

	\newcommand{\tDelta}{\widetilde{\Delta}}
	And therefore $\Delta^{\C'}$ satisfies the requirement of \cref{thm:wullschleger}. Let $\Gamma$ be the protocol guaranteed  in \cref{thm:wullschleger}, and let $\wGamma^{C'}(1^\pk) = \Gamma^{\Delta^{C'}}(1^\pk, 49/100)$. By  \cref{eq:main:LaegeGapToOT,thm:wullschleger}, $\wGamma^{C'}(1^\pk)$ is statistically secure semi-honest OT.  Since $\epsilon_0$ is a bounded from $0$ and $1/2$ by constants, $\wGamma$ running time  in polynomial  in $\pk$. 
\end{proof}

%% file: Comp.tex
\newcommand{\SPCL}{\myOptName{SPCL}}
\newcommand{\PSV}{\myOptName{passive}}
\newcommand{\COMP}{\myOptName{Comp}}

\subsection{The Computational Setting}\label{sec:ProtAmp}	

In this section we extend \cref{thm:main:IT} to the computational setting.  We start by defining the computational analogue of log-ratio leakage. We give two such definition, for the uniform and non-uniform settings. As in similar computational analogue  of information measures \cite{Holenstein06,HaitnerReVa13}, for the uniform version we need to give the uniform distinguisher the ability  to sample from the distributions in consideration,

\begin{definition}[Computational log-ratio leakage]\label{def:compLeakNU}~\\
A channel ensemble   $C= \set{C_\pk= ((\VA_\pk,\OA_\pk),(\VB_\pk,\OB_\pk))}_{\pk \in \N}$ has {\sf $(\eps,\delta)$-\CLeak} {\sf [\resp $(\eps,\delta)$-\NCLeak]} if there exists a channel ensemble   
$\tC= \set{\tC_\pk= ((\tVA_\pk,\tOA_\pk),(\tVB_\pk,\tOB_\pk))}_{\pk \in \N}$  such that the following holds: 
\begin{itemize}

		\item For every $\pk \in \N$: the channel  $\tC_\pk$ has $(\eps(\pk),\delta(\pk))$-\Leak 
		(according to \cref{def:LRLeakage:IT}).
		
		
		\item For every $\Pc \in \ABc$ and  \ppt $\Dc$:
		$$\size{\pr{\Dc^{C_\pk, \tC_\pk}(1^\pk,\VP_\pk,\OA_\pk,\OB_\pk) = 1} - \pr{\Dc^{C_\pk, \tC_\pk}(1^\pk,\tVP_\pk,\tOA_\pk,\tOB_\pk) = 1}} \le \negl (\pk).$$

		[ \resp  for every   $\Pc \in \ABc$:  $\{{\VP_\pk,\OA_\pk,\OB_\pk}\}_{\pk\in\N} \cindist \{\tVP_\pk,\tOA_\pk,\tOB_\pk\}_{\pk\in\N}$]
	
\end{itemize}

\end{definition}
That is, the distinguisher $\Dc$ aiming to tell $\Pc$'s view in $C$ from its view in $\tC$ is equipped the ability to oracle access to $C$ and  $\tC$. This ability is crucial when arguing about the leakage of many  samples of  such channels. We note that typically,  the channel $C$ in consideration is a one induced by an efficient protocol, and thus the oracle access to $C$ given  to $\Dc$ can be simulated efficiently.

\begin{theorem}[Small computational log-ratio leakage implies OT]\label{thm:main:C}
There exists  constant $c_1>0$ such that the following holds. Let   $\eps,\delta, \alpha$ be functions such that for every $\pk \in \N$: $\eps(\pk),\delta(\pk) \in [0,1]$,  $1/8 >\alpha(\pk)  \ge c_1 \cdot \eps(\pk)^2$ and $\delta(\pk) \le  \eps(\pk)^2$ and $\alpha(\pk) > 1/p(\pk)$ for some $p\in \poly$. Let $C$ be a  channel  ensemble that has $(\eps,\delta)$-\CLeak  [\resp $(\eps,\delta)$-\NCLeak]  and $\alpha $-\Agr. Then in the $C$-hybrid model there exists a semi-honest  [\resp non-uniform]  computational OT.  
\end{theorem}

\cref{thm:main:C}  yields the following result. 
\begin{corollary}[Protocols with small  log-ration leakage implies OT]\label{cor:main:C}
	Let $\eps,\delta,\alpha$ be as in \cref{thm:main:C}.  Assume there exists  a \ppt  protocol that induces a   channel ensemble that has $\alpha$-\Agr and $(\eps,\delta)$-\CLeak [\resp $(\eps,\delta)$-\NCLeak], then there exists a  [\resp non-uniform] computational OT. 
\end{corollary}

\begin{proof}
	We only prove the uniform security case,  the non-uniform case follow analogously. By \cref{thm:main:C}, the existence of the guaranteed protocol yields a semi-honest  computational OT protocol $\pi$. By \cite{ImpagliazzoLu89}, the existence of $\pi$ implies the existence of one-way functions. Finally, by \cite{GoldreichMW87}, using one-way functions we can compile $\pi$ into an OT secure against arbitrary adversaries. 
\end{proof}

\newcommand{\tDelta}{\widetilde{\Delta}}
\newcommand{\talpha}{\widetilde{\alpha}}
\begin{proof}[Proof of \cref{thm:main:C}]

Let $\Delta = (\Ac,\Bc)$ be the protocol guaranteed by \cref{thm:main:IT}. Consider the following protocol.
\begin{protocol}\label{proto:OT:C}[Protocol $\tDelta= (\Ac,\Bc)$] 
	
	\item[Oracle:] channel $C$.
	
	\item[Parameter:] security parameter $1^\pk$.
	\item[Operation:] ~
	
	\begin{enumerate}
		\item  $\Ac$ samples $t(\pk)$ independent instances from $C_\pk$, and sends the average agreement  $\talpha$ to   $\tBc$.

		If $\talpha< 1/p(\pk)$, the two parties abort.

		\item The parties interact in $\Delta^{C_\pk}(1^\pk,1^\ell)$, for $\ell=\max(1,{2^{\floor{\log (2/(3\cdot\amax))}-2}})$ (and output the  same values  as the  parties in this interaction do).   
\end{enumerate}
\end{protocol}
It is clear that $\tDelta^C$ runs in polynomial time.   Let $\widetilde{A}_\pk$ be the value of $\talpha$ in a random execution of $\tDelta^C(1^\pk)$. By Hoffeding bound,
 $$\pr{\widetilde{A}_\pk \notin [\alpha-1/3\alpha,\alpha+1/3\alpha]} \le \pr{\widetilde{A}_\pk \notin [\alpha-1/3\cdot p(\pk),\alpha+1/3\cdot p(\pk)]}\le \negl(\pk),$$ 
 which implies that $\alpha \in [3/4\cdot \widetilde{A}_\pk,3/2\cdot \widetilde{A}_\pk]$. The correctness of $\tDelta$ thus follows by \cref{thm:main:IT}.

We prove security only for the uniform security case,  the non-uniform case follow analogously.  Let $\tC$ be the channel ensemble that realizes  the  $(\eps,\delta)$-\CLeak of $C$. First note that the correctness of $\tC$ is the same as $C$ up to some negligible additive value, as otherwise it is easy  to distinguish between $C$ and $\tC$. By the above observation about $\widetilde{A}_k$ and \cref{thm:main:IT}, it follows that $\tDelta^{\tC}$ is a semi-honest secure OT in the $\tC$-hybrid model. Assume there exists a distinguisher  that violates the security of one of the parties in   $\tDelta^C$, a simple hybrid argument yields that a distinguisher with the ability to sample from $C$ and $\tC$ can exploit the above security breach  to violates  the assumed indistinguishability of $C$ and $\tC$.
\end{proof}

\remove{
Recall that a protocol $(\Ac,\Bc)(1^\pk)$ where $\Ac,\Bc$ have no private inputs, and both output bits, induces a channel ensemble $\set{\OA_\pk,\OB_\pk,{\VA_\pk},{\VB_\pk}}_{\pk \in \N}$ where $\OA_\pk,\OB_\pk$ and ${\VA_\pk},{\VB_\pk}$ are the outputs and views of $\Ac,\Bc$, respectively, on the input $1^\pk$.

We identify the protocol with the channel ensemble, and will say that the protocol $(\Ac,\Bc)$ has \emph{agreement} or \emph{leakage}.

\Rnote{Once we're done, we'll state a corollary in terms of \ppt protocols.}

As in the information theoretic settings, we prove \cref{thm:main:C} by converting the channel ensemble into a \COMP-\SPCL-\PSV-\WBSC. It will be helpful to consider the following ``simulation variant '' of \COMP-\SPCL-\PSV-\WBSC.

\begin{definition}[Sim-\SPCL-\WBSC]
A channel ensemble $\set{\OA_\pk,\OB_\pk,{\VA_\pk},{\VB_\pk}}_{\pk \in \N}$ is \emph{$(\eps_0,p)$-Sim-\COMP-\SPCL-\PSV-\WBSC} if for every $\pk \in \N$, there exist random variables $\widetilde{\VA_\pk},\widetilde{\VB_\pk}$ that are jointly distributed with $\OA_\pk,\OB_\pk,{\VA_\pk},{\VB_\pk}$ such that:
\begin{itemize}
\item For every $\pk \in \N$, the channel  $(\OA_\pk,\OB_\pk,\widetilde{\VA_\pk},\widetilde{\VB_\pk})$ is an $(\eps_0,p)$-spcl-\PSV-\WBSC.
\item Let $E_{\pk}$ denote the indicator variable for the event $\set{\OA_\pk=\OB_\pk}$, the following holds:
    \begin{itemize}
    \item For every infinite sequence of bits $\set{b_\pk}_{\pk \in \N}$, $\set{{\VA_\pk}|_{E_\pk=b_\pk}} \cindist \set{\widetilde{\VA_\pk}|_{E_\pk=b_\pk}}$.
    \item For every infinite sequence of bits $\set{b_\pk}_{\pk \in \N}$,  $\set{{\VB_\pk}|_{E_\pk=b_\pk}} \cindist \set{\widetilde{\VB_\pk}|_{E_\pk=b_\pk}}$.
    \end{itemize}
\end{itemize}
\end{definition}


\begin{lemma}
	\label{lem:main prot works comp}
	If $\set{\OA_\pk,\OB_\pk,{\VA_\pk},{\VB_\pk}}_{\pk \in \N}$ is a channel ensemble that satisfies the requirements in Theorem \ref{thm:main:C}, then:
\begin{itemize}
 \item The probability that an invocation of \cref{proto:ot:IT} is successful is at least $(1/2+\alpha)^{\ell}=\Omega(2^{-(1/4\alpha)})$.
 \item The channel induced by a successful invocation of protocol \ref{proto:ot} is an
$(\eps_0,p)$-Sim-\COMP-\SPCL-\PSV-\WBSC for $3/8 \geq \eps_0 \geq 1/8$, and $p=2\ell\epsilon^2+2\epsilon\sqrt{2\ell\ln({1/\delta'})}+\delta'+\ell\delta$, for any $0<\delta'<1$.
\end{itemize}
\end{lemma}

\begin{proof}
The operation of \cref{proto:ot} does not depend on the ``view variables'' of the channel, and only depends on the outputs.
Therefore, the run of such a protocol is identical on the channels $\set{\OA_\pk,\OB_\pk,{\VA_\pk},{\VB_\pk}}$ and $\set{\OA_\pk,\OB_\pk,\widetilde{\VA_\pk},\widetilde{\VB_\pk}}$, and in particular the outputs are the same. Thus, the first item of \cref{lem:main prot works comp} follows directly from the first item of \cref{lem:main prot works}.
By the same argument, it follows that applying \cref{proto:ot} on channel $\set{\OA_\pk,\OB_\pk,{\VA_\pk},{\VB_\pk}}$ satisfy the correctness property of a spcl-\PSV-\WBSC, with the same $\eps_0,p$ as in \cref{lem:main prot works}.

Thus, the expected number of invocations needed to get a successful invocation (and the parties receive an output) is $O(2^{1/4\alpha(\pk)})$ which is polynomial as $\alpha(\pk) =\Omega(\frac{1}{\log \pk})$.\Nnote{tree}

	For security, note that the views of the parties in a successful execution of this protocol is $\ell(\pk) \in poly(\pk)$ independent samples from $\VA_\pk|_{E_\pk=b_\pk}$ and $\VB_\pk|_{E_\pk=b_\pk}$ for some $b_\pk \in \zo$, and an additional uniform independent bit ($B^\Ac$ and $B^\Bc$).
Therefore, using a hybrid argument, we get that the views of the parties in this protocol are computationally indistinguishable to the views of the parties in a execution of the protocol with the channel $(\OA_\pk,\OB_\pk,\widetilde{\VA_\pk},\widetilde{\VB_\pk})$. That is,
	
	\begin{equation*}
	\set{(\OA_\pk,{\VA_\pk}|_{E_\pk=b_\pk})^\ell,B^\Ac}_{\pk \in \N}  \cindist \set{(\OA_\pk,\widetilde{\VA_\pk}|_{E_\pk=b_\pk})^\ell,B^\Ac}_{\pk \in \N}
	\end{equation*}
	\begin{equation*}
	\set{(\OB_\pk,{\VB_\pk}|_{E_\pk=b_\pk})^\ell,B^\Bc}_{\pk \in \N}  \cindist \set{(\OB_\pk,\widetilde{\VB_\pk}|_{E_\pk=b_\pk})^\ell,B^\Bc}_{\pk \in \N}
	\end{equation*}
	
	 Since, from \cref{lem:main prot works}, an execution of \cref{proto:ot} using this channel is  $(\eps_0,p)$-spcl-\PSV-\WBSC, we get that execution of the protocol using $(\OA_\pk,\OB_\pk,{\VA_\pk},{\VB_\pk})$ is $(\eps_0,p)$-Sim-\COMP-\SPCL-\PSV-\WBSC.
\end{proof}

To complete the proof of \cref{thm:main:C}, we show that Sim-\COMP-\SPCL-\PSV-\WBSC indeed implies OT, the next lemma state that it implies \COMP-\SPCL-\PSV-\WBSC. Since by \cref{clm:spcl->WBSC} and \cref{thm:wullschleger} \COMP-\SPCL-\PSV-\WBSC implies OT for the same choice of parameters as in the proof of \cref{thm:mainIT}.

\begin{claim}
	An $(\epsilon_0, p)$-Sim-\COMP-\SPCL-\PSV-\WBSC is also a $(\epsilon_0, p+1/\pk)$-\COMP-\SPCL-\PSV-\WBSC.
\end{claim}

We remark that the choice of the function $1/\pk$ above is quite arbitrary and any function that is noticeable will do.

\begin{proof}
    Let $\set{\OA_\pk,\OB_\pk,{\VA_\pk},{\VB_\pk}}_{\pk \in \N}$ be an $(\epsilon_0, p)$-Sim-\COMP-\SPCL-\PSV-\WBSC. First, notice that the correctness of \COMP-\SPCL-\PSV-\WBSC holds from definition.

	For receiver security, assume for contradiction that there is distinguisher $\Dc$ that can break the receiver security of the protocol using advises $\set{z_\pk}_{\pk \in \N}$. Then, there is a polynomial $f(\pk)$, s.t. for infinite many $\pk \in \N$  it holds that
	
	$\size{\pr{\Dc(1^\pk,({\VA_\pk}, \OA_\pk)|_{E_\pk=0},z_\pk)=1}-\pr{\Dc(1^\pk,({\VA_\pk}, \OA_\pk)|_{E_\pk=1},z_\pk)=1}}\geq p +f(\pk)$.
	
	From the definition of Sim-\COMP-\SPCL-\PSV-\WBSC, there exists random variable $\widetilde{\VA_\pk}$ s.t. for any large enough $\pk$ and a bit $b_\pk$:
	
	$\size{\pr{\Dc(1^\pk,({\VA_\pk}, \OA_\pk)|_{E_\pk=b_\pk},z_\pk)=1}-\pr{\Dc(1^\pk,(\widetilde{\VA_\pk}, \OA_\pk)|_{E_\pk=b_\pk},z_\pk)=1}}\leq neg(\pk)$, and,
	
	$\size{\pr{\Dc(1^\pk,(\widetilde{\VA_\pk}, \OA_\pk)|_{E_\pk=0},z_\pk)=1}-\pr{\Dc(1^\pk,(\widetilde{\VA_\pk}, \OA_\pk)|_{E_\pk=1},z_\pk)=1}}\leq p$

	Now, from the triangle inequality it holds that:
	
	\begin{align*}
	p+f(x) &\leq \size{\pr{\Dc(1^\pk,({\VA_\pk}, \OA_\pk)|_{E_\pk=0},z_\pk)=1}-\pr{\Dc(1^\pk,({\VA_\pk}, \OA_\pk)|_{E_\pk=1},z_\pk)=1}}\\
	&\begin{aligned}
	=&|\pr{\Dc(1^\pk,({\VA_\pk}, \OA_\pk)|_{E_\pk=0},z_\pk)=1}-\pr{\Dc(1^\pk,(\widetilde{\VA_\pk}, \OA_\pk)|_{E_\pk=0},z_\pk)=1}\\
	&+\pr{\Dc(1^\pk,(\widetilde{\VA_\pk}, \OA_\pk)|_{E_\pk=0},z_\pk)=1}-\pr{\Dc(1^\pk,(\widetilde{\VA_\pk}, \OA_\pk)|_{E_\pk=1},z_\pk)=1}\\
	&+\pr{\Dc(1^\pk,(\widetilde{\VA_\pk}, \OA_\pk)|_{E_\pk=1},z_\pk)=1}-\pr{\Dc(1^\pk,({\VA_\pk}, \OA_\pk)|_{E_\pk=1},z_\pk)=1}|
	\end{aligned}\\
	&\begin{aligned}
	\leq&\size{\pr{\Dc(1^\pk,({\VA_\pk}, \OA_\pk)|_{E_\pk=0},z_\pk)=1}-\pr{\Dc(1^\pk,(\widetilde{\VA_\pk}, \OA_\pk)|_{E_\pk=0},z_\pk)=1}}\\
	&+\size{\pr{\Dc(1^\pk,(\widetilde{\VA_\pk}, \OA_\pk)|_{E_\pk=0},z_\pk)=1}-\pr{\Dc(1^\pk,(\widetilde{\VA_\pk}, \OA_\pk)|_{E_\pk=1},z_\pk)=1}}\\
	&+\size{\pr{\Dc(1^\pk,(\widetilde{\VA_\pk}, \OA_\pk)|_{E_\pk=1},z_\pk)=1}-\pr{\Dc(1^\pk,({\VA_\pk}, \OA_\pk)|_{E_\pk=1},z_\pk)=1}}
	\end{aligned}\\
	&\begin{aligned}
	\leq&\size{\pr{\Dc(1^\pk,({\VA_\pk}, \OA_\pk)|_{E_\pk=0},z_\pk)=1}-\pr{\Dc(1^\pk,(\widetilde{\VA_\pk}, \OA_\pk)|_{E_\pk=0},z_\pk)=1}}\\
	&+\size{\pr{\Dc(1^\pk,(\widetilde{\VA_\pk}, \OA_\pk)|_{E_\pk=1},z_\pk)=1}-\pr{\Dc(1^\pk,({\VA_\pk}, \OA_\pk)|_{E_\pk=1},z_\pk)=1}}\\
	&+p
	\end{aligned}\\
	\end{align*}
	
	We get that for infinite many $\pk \in \N$,
		\begin{align*}
	f(x) \leq &\size{\pr{\Dc(1^\pk,({\VA_\pk}, \OA_\pk)|_{E_\pk=0},z_\pk)=1}-\pr{\Dc(1^\pk,(\widetilde{\VA_\pk}, \OA_\pk)|_{E_\pk=0},z_\pk)=1}}\\
	&+\size{\pr{\Dc(1^\pk,(\widetilde{\VA_\pk}, \OA_\pk)|_{E_\pk=1},z_\pk)=1}-\pr{\Dc(1^\pk,({\VA_\pk}, \OA_\pk)|_{E_\pk=1},z_\pk)=1}}
	\end{align*}
	
	Which imply that at least one of the following holds for infinite many $\pk$:
	\begin{itemize}
		\item $f(x)/2 \leq \size{\pr{\Dc(1^\pk,({\VA_\pk}, \OA_\pk)|_{E_\pk=0},z_\pk)=1}-\pr{\Dc(1^\pk,(\widetilde{\VA_\pk}, \OA_\pk)|_{E_\pk=0},z_\pk)=1}}$
		\item $f(x)/2 \leq \size{\pr{\Dc(1^\pk,(\widetilde{\VA_\pk}, \OA_\pk)|_{E_\pk=1},z_\pk)=1}-\pr{\Dc(1^\pk,({\VA_\pk}, \OA_\pk)|_{E_\pk=1},z_\pk)=1}}$
	\end{itemize}
	in contradiction to the definition of $\widetilde{\VA_\pk}$.
	
	The proof for the security of the sender follows similar lines.
\end{proof}

}

%% file: DPXOR.tex
\section{Characterization of  Channel for Distributed Differentially Private Computation}\label{sec:DPXORtoOTl}

In this section we prove our results on 2-party differentially private computation. Our goal is to show that a sufficiently accurate 2-party differentially private computation of the XOR function implies OT.
In \cref{sec:DPXORtoOT:IT} we consider differential privacy in an information theoretic setting. In \cref{sec:DPXORtoOT:C} we consider the computational setting, giving formal definitions with which we restate and prove \cref{thm:DPXORInf}. Finally, in \cref{sec:non-monotone} we extend our result to functions over many bits that are not ``monotone under relabeling''.

Throughout, we use the following notions of  agreement and  accuracy for functionalities. Since we care about lower bounds, we only consider (a weaker) average-case variant of these notions.

\begin{definition}[Accuracy and agreement, functionalities]\label{def:AccAgr:funcionality}
Let $f:\cX\times \cY \to \zn$ be a  Boolean output functionality  and let  $(\OA_{x,y}, \OB_{x,y}) = f(x,y)$. We say that $f$ has {\sf average agreement $\alpha$} if  $\ppr{x,y \gets \cX,\cY}{\OA_{x,y} = \OB_{x,y}}  =  \half + \alpha$. We say that $f$  {\sf computes a Boolean function $g$ with average correctness   $\beta$}, if $\pr{\OA_{x,y} = \OB_{x,y}=g(x,y)} = \half + \beta$.

A non-Boolean  output functionality $f$ has agreement $\alpha$  if the Boolean functionality $f'$, defined by $f'(x,y) = (o^\Ac_1,o^\Bc_1)$ for $(o^\Ac,o^\Bc) \gets f(x,y)$, has agreement $\alpha$. Similarly,  $f$  computes  $g$ with correctness   $\beta$, if  the functionality $f'$ does.
\end{definition}
Namely,  a non-Boolean functionality $f$ has certain agreement and correctness (\wrt Boolean function $g$) if  this holds \wrt the first bits it outputs (\ie its ``designated output bits'').


\subsection{The XOR Functionality}
\subsubsection{The Information Theoretic Case}\label{sec:DPXORtoOT:IT}

We prove the  following  characterization of differential private functionalities  for computing XOR.

\begin{theorem}\label{thm:DPXORtoOT:IT}
There exists a \ppt protocol $\Delta$ and a constant $c_1>0$ such that the following holds. Let $\eps,\beta \in [0,1]$ be such that $\beta \geq c_1 \cdot \eps^2$. Let $f = (f_\Ac,f_\Bc)$ be a functionality  that is $\eps$-DP, has perfect agreement and computes the XOR function with average correctness $\beta$. Then $\Delta^f(1^\pk,1^{\floor{1/\beta}})$  is a  semi-honest   statistically secure OT in the $f$-hybrid model. Furthermore, the parties in $\Delta$ only make use of the first bit of the outputs of  $f$.
\end{theorem}

We prove \cref{thm:DPXORtoOT:IT}  by constructing  in the $f$-hybrid model  a balanced protocol that induced balanced channel with $\beta$-agreement and  that has $(2\eps,0)$-\Leak.

\begin{protocol}[$\pi^f= (\Ac,\Bc)$]\label{pro:DPXOR:IT}
	\item Oracle:  $f$.
	
	\item Operation:
	
	\begin{enumerate}
		\item 	 $\Ac$ samples  $i_\Ac\gets\zo$ and  $\Bc$ samples  $i_\Bc\gets\zo$.
		
		\item The parties make  a joint call to  $f(i_\Ac,i_\Bc)$. Let $\out_\Bc$ be the first bit of the output given to $\Bc$.

		\item $\Ac$ sends   $r\gets\zo$ to $\Bc$. \label{pro:item: B}
		
		\item The parties  output $i_\Ac\xor r$ and $\out_\Bc \xor i_\Bc \xor r$, respectively.
	\end{enumerate}
\end{protocol}

The proof of \cref{thm:DPXORtoOT:IT}  immediately follows be the next lemma and the tools we devolved in the previous section.

\begin{lemma}\label{lem:DPXORtoOT:IT}
Let $\beta,\eps$ and $f$ be as in \cref{thm:DPXORtoOT:IT}, then in the $f$-hybrid model protocol $\pi^f$ induces  a channel of $(2\eps,0)$-\Leak and $\beta$-\Agr.
\end{lemma}
We prove  \cref{lem:DPXORtoOT:IT} below, but first use it for proving  \cref{thm:DPXORtoOT:IT}.

\begin{proof}[Proof of \cref{thm:DPXORtoOT:IT}]
	The proof directly follows from  \cref{thm:main:IT} and \cref{lem:DPXORtoOT:IT}. Note that, by differential privacy properties, $\beta$ is bounded. Specifically, for sufficiently large $c_1$, $\beta \leq 1/8$. 
\end{proof}


Let $C= ((\OA,\VA),(\OB,\VB))$ denote the channel induces by a random execution of $\pi^f$. \cref{lem:DPXORtoOT:IT} is an immediate consequence of the following three claims.

\begin{claim}\label{clm:DPXORtoOT:IT:Agr}
$\pr{\OA=\OB}  = 1/2+\beta$.
\end{claim}
\begin{proof}
	Follows by construction and the assumed accuracy of $f$.
\end{proof}

 \begin{claim}\label{clm:DPXORtoOT:IT:Leak}
  For both $\p\in\set{\Ac,\Bc}$:
 	$(\VP,\OP)\mid_ {\OA= \OB}\rindist_{2\eps} (\VP,\OP) \mid_{ \OA\neq \OB}$.
 \end{claim}

\renewcommand{\IP}{I^\Pc}
\renewcommand{\IA}{I^\Ac}
\newcommand{\IB}{I^\Bc}
\newcommand{\ob}{\overline{b}}
\newcommand{\oa}{\overline{a}}

\remove{\begin{proof}
We assume for ease of notation that the value of the bit $r$ chosen  by $\Ac$ in $\pi$ is always set to $0$ (rather than uniformly at random), it is clear that this change has no effect on log-ration distance we measure. We also only prove for $\p= \Ac$, where the case $\p= \Bc$ follows analogously. Fix $a\in \zo$. Let $\IP$ be the value of $i_\Pc$ as appears in $\VP$. The differential privacy of $f$ yields that:
\begin{align}\label{eq:DPXORtoOT:IT:Leak:1}
\VA|_{\IA = a,\IB=0}  \rindist_\eps\VA|_{\IA = a,\IB=1}
\end{align}
Let $\Out$ be the (always common) output of $f$ in $\VA$. Fix an  any event $E$  and let $E_o$, for $o\in \zo$, be in event in $\set{E \land \Out = o}$. Since $\OA = \IA$ and since $\OB= \IB \xor \Out$ (recall we assume that the coin $r$ is $0$), it follows that
\begin{align}\label{eq:DPXORtoOT:IT:Leak:2}
\begin{aligned}
\ppr{\VA \mid \OA= a,\OB=b}{E_o}  &=  \frac{\ppr{\VA,\VB \mid \OA= a}{E_o,\OB=b}}{\ppr{\VA,\VB \mid \OA= a}{\OB=b}}\\
&=  \frac{\ppr{\VA,\VB \mid \IA= a}{E_o,\IB=b\xor o}}{\ppr{\VA,\VB \mid \OA= a}{\OB=b}}\\
&=  \ppr{\VA,\VB \mid \IA= a,\IB=b\xor o}{E_o}\frac{\ppr{\VA,\VB \mid \IA=a}{\IB= b\xor o}}{\ppr{\VA,\VB \mid \OA= a}{\OB=b}}\\
&=  \ppr{\VA,\VB \mid \IA= a,\IB=b\xor o}{E_o}\frac{1}{2\ppr{\VA,\VB \mid \IA= a}{\OB=b}}\\
&\leq  \ppr{\VA,\VB \mid \IA= a,\IB=b\xor o}{E_o}\frac{1+e^\eps}{2}\\
&\leq  e^\eps\ppr{\VA,\VB \mid \IA= a,\IB=b\xor o}{E_o}\\
\end{aligned}
\end{align}
Where the last inequality implied by the facts that  $\ppr{\VA,\VB}{\OB=b}=1/2$, and,

 $\ppr{\VA,\VB \mid \IA= a}{\OB=b} \leq e^\eps \ppr{\VA,\VB \mid \IA= \bar{a}}{\OB=b}$.

 It follows that for   $b\in \zo$:

 \begin{align*}
 \ppr{\VA \mid \OA= a,\OB=b}{E} & =  \ppr{\VA \mid \OA= a,\OB=b}{E_0} +  \ppr{\VA \mid \OA= a,\OB=b}{E_1}\\
 &\le e^\eps \cdot  \left( \ppr{\VA \mid \IA= a,\IB=b}{E_0} +  \ppr{\VA \mid \IA= a,\IB=\ob}{E_1}\right)\\
 &\le e^{2\eps} \cdot  \left(\ppr{\VA \mid \IA= a,\IB=\ob}{E_0} +  \ppr{\VA \mid \IA= a,\IB=b}{E_1} \right)\\
 &= e^{2\eps} \cdot  \left(\ppr{\VA \mid \IA= a,\OB=\ob}{E_0} +  \ppr{\VA \mid \IA= a,\OB=\ob}{E_1} \right)\\
 &= e^{2\eps} \cdot   \ppr{\VA \mid \OA= a,\OB=\ob}{E} .
 \end{align*}
 The second and forth qualities are by \cref{eq:DPXORtoOT:IT:Leak:2}, and the inequality is  by \cref{eq:DPXORtoOT:IT:Leak:1}. We conclude that for every  $a,r\in \zo$:
\begin{align}\label{eq:DPXORtoOT:IT:Leak:3}
\VA|_{\OA = a,\OB=0,R=r}  \rindist_{2\eps}\VA|_{\OA = a,\OB=1,R=r}
\end{align}
To see that $\VA|_{\OA = a,\OB=0}  \rindist_{3\eps}\VA|_{\OA = a,\OB=1}$, note that, for $E^r=\set{E \land R = r}$
 \begin{align*}
\ppr{\VA \mid \OA= a,\OB=b}{E^r} & =  \ppr{\VA \mid \OA= a,\OB=b, R=r}{E^r}\ppr{\VA \mid \OA= a,\OB=b}{R=r}\\
\end{align*}
and it hold that:
 \begin{align*}
\ppr{\VA \mid \OA= a,\OB=b}{R=r}&=\ppr{\VA \mid \OA=a}{\OB= b,R=r}/\ppr{\VA \mid \OA=a}{\OB= b}\\
&=\ppr{\VA \mid \OA=a,R=r}{\IB\xor \Out= b \xor r}\ppr{\VA \mid \OA=a}{R=r}/\ppr{\VA \mid \OA=a}{\OB= b}\\
&=1/2\cdot\ppr{\VA \mid \OA=a,R=r}{\IB\xor \Out= b \xor r}/\ppr{\VA \mid \OA=a}{\OB= b}.
\end{align*}

Finally, the convexity of the log-ration distance, \cref{fact:LR:Convex}, yields that
\Nnote{It only holds when $r$ is not fixed, right? \Inote{didnt get your comment}}
\begin{align*}
\VA|_{\OA =\OB}  \rindist_{2\eps}\VA|_{\OA \neq \OB}.
\end{align*}
\end{proof}
}

  We  use the following claim that states that we have bounded leakage with respect to the outputs of protocol $\tPi$.
\begin{claim}\label{claim:dpnew}
	For every $a,b \in \zo$ it holds that,
	\begin{itemize}
		\item $(\VA,\OA)\mid_ {\OB= b}\rindist_\eps (\VA,\OA) \mid_{ \OB=\ob}$
		\item $(\VB,\OB)\mid_ {\OA= a}\rindist_\eps (\VB,\OB) \mid_{ \OA=\oa}$
	\end{itemize}
\end{claim}
We now prove \cref{clm:DPXORtoOT:IT:Leak} using the above claim.  We prove for $\p= \Ac$, where the case $\p= \Bc$ follows analogously.
\begin{proof}[Proof of \cref{clm:DPXORtoOT:IT:Leak}]

	 For  $v\in\Supp(\VA)$ and $a\in\zo$, let $H^{\vv,a}=\set{(\VA,\OA)=(v,a)}$. We need to  show that for every $v,a$:
	 \begin{align}\label{eq:DPIT:toProve}
	  \ppr{\VA\mid {\OA=\OB}}{H^{\vv,a}}\rindist_{2\eps}\ppr{\VA\mid {\OA\neq\OB}}{H^{\vv,a}}
	 \end{align}
	
	  Compute,
	\begin{align}\label{eq:DPIT:eq}
	\begin{aligned}
		\ppr{\VA\mid {\OA=\OB}}{H^{\vv,a}} &= \ppr{\VA\mid {\OA=\OB}, \OA=a}{H^{\vv,a}}\cdot\pr{ \OA=a \mid \OA= \OB}\\
	&= \ppr{\VA\mid \OA=a, \OB=a}{H^{\vv,a}}\cdot\pr{ \OA=a \mid \OA= \OB}\\
	&= \ppr{\VA\mid \OB=a}{H^{\vv,a}}\cdot\frac{\pr{ \OA=a \mid \OA= \OB}}{\pr{ \OA=a \mid \OB=a}}.
	\end{aligned}
	\end{align}
	
	In the same way,
	\begin{align}\label{eq:DPIT:neq}
	\ppr{\VA\mid {\OA\neq\OB}}{H^{\vv,a}} = \ppr{\VA\mid \OB=\oa}{H^{\vv,a}}\cdot\frac{\pr{ \OA=a \mid \OA\neq \OB}}{\pr{ \OA=a \mid \OB=\oa}}
\end{align}
	By construction,
	 \begin{align}\label{eq:DPIT:output}
	\pr{ \OA=a \mid \OA= \OB} = \pr{ \OA=a \mid \OA\neq \OB} = 1/2
	\end{align}
	We conclude that
		\begin{align*}
		\ppr{\VA\mid {\OA=\OB}}{H^{\vv,a}} &=  \ppr{\VA\mid \OB=a}{H^{\vv,a}}\cdot\frac{\pr{ \OA=a \mid \OA= \OB}}{\pr{ \OA=a \mid \OB=a}}\tag{by \cref{eq:DPIT:eq}}\\
	&\leq e^\eps \cdot \ppr{\VA\mid \OB=\oa}{H^{\vv,a}}\cdot\frac{\pr{ \OA=a \mid \OA= \OB}}{\pr{ \OA=a \mid \OB=a}}\tag{by  \cref{claim:dpnew}}\\
	&= e^\eps \cdot\ppr{\VA\mid \OB=\oa}{H^{\vv,a}}\cdot\frac{\pr{ \OA=a \mid \OA\neq \OB}}{\pr{ \OA=a \mid \OB=a}}\tag{by  \cref{eq:DPIT:output}}\\
	&\leq e^{2\eps} \cdot \ppr{\VA\mid \OB=\oa}{H^{\vv,a}}\cdot\frac{\pr{ \OA=a \mid \OA\neq \OB}}{\pr{ \OA=a \mid \OB=\oa}}\\
	&=e^{2\eps}\cdot \ppr{\VA\mid {\OA\neq\OB}}{H^{\vv,a}}\tag{by  \cref{eq:DPIT:neq}}.
\end{align*}
The last inequality holds since  by \cref{claim:dpnew}, $\frac{\pr{ \OA=a \mid \OB=\oa}}{\pr{ \OA=a \mid \OB=a}}\leq e^\eps$.

The proof that $\ppr{\VA\mid {\OA=\OB}}{H^{\vv,a}}  \ge e^{2\eps}\cdot \ppr{\VA\mid {\OA\neq\OB}}{H^{\vv,a}}$ is identical, thus the claim holds.
\end{proof}
\paragraph{Proving \cref{claim:dpnew}.}
	
\begin{proof}[Proof of \cref{claim:dpnew}]
We write $f= (f^\Ac,f^\Bc)$. Let $\IA$ and $\IB$ be the values of the inputs of the parties, and let $\Out$ be the (common) values of the output $f^\Pc(\IA,\IB)_1$ and the random bit $r$ in $\VP$ respectively.
	Fix  $\oo \in\zo$ and $v\in \sup(\VB)$, and let $b,r$ be the values of $\IB$ and $R$ according to $v$. Since $R$ and $\IA$ are uniform bits, the value of $R$ is independent from $\IA$, and independent from $\OA=\IA \xor R$ (separately). Thus,
			\begin{align}\label{eq:DPIT:outDP}
	&\pr{\VB=v \mid \OA= \oo } =\pr{\VB=v, \IB=b, R=r  \mid \OA= \oo}\\
		&=\pr{\VB=v, \IB=b  \mid \OA= \oo, R=r }\cdot\pr{R=r \mid \OA= \oo } \nonumber\\
		&=\pr{\VB=v, \IB=b \mid \IA=\oo\xor r, R=r}\cdot\pr{R=r \mid \OA= \oo }\nonumber\\
			&=\pr{\VB=v, \IB=b, R=r \mid \IA=\oo\xor r}\cdot\frac{\pr{R=r \mid \OA= \oo }}{\pr{R=r \mid \IA=\oo\xor r}}\nonumber\\
		&=\pr{\VB=v, \IB=b, R=r \mid \IA=\oo\xor r}\nonumber\\
		&=1/2\cdot \pr{\VB=v \mid \IA=\oo\xor r, \IB=b}.\nonumber
	\end{align}
	Since \cref{eq:DPIT:outDP} holds for every $\oo \in \zo$, and since
	\begin{align*}
1/2\cdot \pr{\VB=v \mid \IA= r, \IB=b}\rindist_\eps 1/2\cdot \pr{\VB=v \mid \IA=\overline{r}, \IB=b},
	\end{align*}
	we conclude that  $\pr{\VB=v \mid \OA= 0 }\rindist_\eps \pr{\VB=v \mid \OA= 1 }$.

The proof of the second item follows in by a similar argument.
For every  $\oo \in\zo$ and  $v\in \sup(\VA)$, let $a,r'$ be the values of $\IA$ and $R\xor\Out$ according to $v$ respectively. Since the value of $R\xor\Out$ is independent from $\IB$, and independent from $\OB$ (separately), we conclude that
	\begin{align*}
	&\pr{\VA=v  \mid \OB= \oo }
	=\pr{\VA=v, \IA=a, R\xor\Out=r'  \mid \OB= \oo }\\
	&=\pr{\VA=v, \IA=a \mid \OB= \oo, R\xor\Out=r' }\cdot\pr{R\xor\Out=r' \mid \OB=\oo}\\
	&=\pr{\VA=v, \IA=a \mid \IB= \oo\xor r', R\xor\Out=r' }\cdot\pr{R\xor\Out=r' \mid \OB=\oo}\\
	&=\pr{\VA=v, \IA=a, R\xor\Out=r' \mid \IB= \oo\xor r' }\cdot\frac{\pr{R\xor\Out=r' \mid \OB=\oo}}{\pr{R\xor\Out=r' \mid \IB= \oo\xor r'}}\\
	&=\pr{\VA=v, \IA=a, R\xor\Out=r' \mid \IB= \oo\xor r' }\\
	&=1/2\cdot\pr{\VA=v \mid \IB= \oo\xor r', \IA=a }.
	\end{align*}
\end{proof}

\paragraph{Proving \cref{lem:DPXORtoOT:IT}.}
\begin{proof}[Proof of \cref{lem:DPXORtoOT:IT}]
	Let $C=(\OA,\OB,\VA,\VB)$ denotes the channel induces by $\pi^f$. \cref{clm:DPXORtoOT:IT:Agr} yields that  $C$ has  $\beta$-\Agr, and  \cref{clm:DPXORtoOT:IT:Leak} yields that  $C$ has $(2\eps,0)$-\Leak.
\end{proof}

%% file: Comp-DPXOR.tex
\newcommand{\Ox}{O_{x,y,\pk}}
\newcommand{\OAx}{\OA_{x,y,\pk}}
\newcommand{\OBx}{\OB_{x,y,\pk}}
\newcommand{\VAx}{\VA_{x,y,\pk}}
\newcommand{\VBx}{\VB_{x,y,\pk}}
\newcommand{\VPx}{\VP_{x,y,\pk}}
\newcommand{\tOAx}{\tOA_{x,y,\pk}}
\newcommand{\tOBx}{\tOB_{x,y,\pk}}
\newcommand{\tVAx}{\tVA_{x,y,\pk}}
\newcommand{\tVBx}{\tVB_{x,y,\pk}}
\newcommand{\tVPx}{\tVP_{x,y,\pk}}
\newcommand{\Rx}{R_{x,y,\pk}}
\newcommand{\hf}{\widehat{f}}

\subsubsection{The  Computational Case}\label{sec:DPXORtoOT:C}

In this section we restate and prove \cref{thm:DPXORInf}. We will use the following definition.

\begin{definition}[Accuracy and agreement, protocols]\label{def:AccAgr:protocol}
	Let $\Pi$ be a  Boolean output protocol  and let  $(\OA_{x,y}, \OB_{x,y}) = \Pi(x,y)$. We say that $\Pi$ has {\sf average agreement $\alpha$} if,  $\ppr{x,y \gets \cX,\cY}{\OA_{x,y} = \OB_{x,y}}  = \half + \alpha$. We say that $\pi$  {\sf computes a Boolean function $g:\cX\times \cY \to \zn$ with average correctness   $\beta$}, if   \\$\ppr{x,y \gets \cX,\cY}{\OA_{x,y} = \OB_{x,y}=g(x,y)} = \half + \beta$. Similarly, we say that $\pi$  {\sf computes  $g$ with worst-case correctness   $\beta$}, if for every inputs $x,y \in \cX,\cY$, \\ $\pr{\OA_{x,y} = \OB_{x,y}=g(x,y)} \geq \half + \beta$.
\end{definition}

The following is a restatement of \cref{thm:DPXORInf}.

\begin{theorem}\label{thm:DPXORtoOT:C}
	There exists a constant $c>0$ such that the following holds. Let $\eps, \beta$ be functions such that for every $\pk\in \N$: $\eps(\pk), \beta(\pk) \in [0,1]$, $\beta(\pk)\geq c\cdot\eps(\pk)^2$ and $\beta(\pk) \geq 1/p(\pk)$ for some $p \in \poly$. Assume there exist  a \ppt Boolean output protocol  that is semi-honest $\eps$-\CDP  and  computes the XOR functionality with perfect agreement and average correctness  at least $\beta(\pk)$. Then there exists a computationally non-uniform secure  OT.
\end{theorem}

We make use of the following  notion of simulation based computational differential privacy, in the spirit of \cite{MironovPRV09}.
\begin{definition}[Simulation based computational differential privacy]\label{def:SImDP}
	
	A two-output functionality ensemble $\set{f_\pk = (f^\Ac_\pk,f^\Bc_\pk)}_{\pk\in\N}$ over input domain $\zn\times \zn$ is {\sf $\eps$-\NCDP}
	if the there exists a functionality ensemble \\$\set{\tf_\pk = (\tf_\pk^\Ac,\tf_\pk^\Bc)}_{\pk\in\N}$  such that the following holds:
	\begin{itemize}
		
		\item For every $\pk \in \N$: the functionality  $\tf_\pk$ is $\eps(\pk)$-\DP
		(according to \cref{def:DP:IT}).

		\item For both  $\Pc \in \ABc$ and every $x,y\in \zn$:
		
		$$\set{f_\pk^\Pc(x,y)}_{\pk\in\N} \cindist \set{\tf_\pk^\Pc(x,y)}_{\pk\in\N}$$
	\end{itemize}	
\end{definition}
As \cref{lem:CDPtoSimDP} (give below) shows, for Boolean inputs,  the above functionality (\cref{def:SImDP}) is closely related to the more standard  $\eps$-\CDP (\cref{def:DP:C}).

The proof of \cref{thm:DPXORtoOT:C} immediately follows by the next two lemmata.

\begin{lemma}\label{lem:CDPtoSimDP}
For any   $\eps$-\CDP  protocol $\pi = (\Ac,\Bc)$,   the functionality ensemble
$\set{f_\pk =  (f^\Ac_\pk(x,y),f_\pk^\Bc(x,y)}_{\pk\in\N}$ defined by $f_\pk(x,y)$  outputting  the parties' views in a random execution  of $(\Ac(x),\Bc(y))(1^\pk)$, is $\eps$-\NCDP.
\end{lemma}

\begin{lemma}\label{lem:DPXORfunc:C}
	 Let $\eps, \beta$ be functions satisfying the requirements of  \cref{thm:DPXORtoOT:C}. Let $f$ be a functionality ensemble  that is $\eps$-\NCDP, has perfect agreement and computes the XOR function with average correctness at least $\beta(\pk)$. Then in the $f$-hybrid model there exists a semi-honest  secure OT.
\end{lemma}

\paragraph{Proving \cref{thm:DPXORtoOT:C}.}
\begin{proof}[Proof of \cref{thm:DPXORtoOT:C}]
	Let $\pi$ be a protocol  satisfying the requirements in \cref{thm:DPXORtoOT:C}, and let $f$ be the functionality ensemble guaranteed  by \cref{lem:CDPtoSimDP}  for  $\Pi$. By construction, $f$ has perfect agreement and computes the XOR function with correctness  $\beta$. Thus, the theorem proof follows by \cref{lem:DPXORfunc:C}.	
\end{proof}

\paragraph{Proving \cref{lem:CDPtoSimDP}.}
\begin{proof}[Proof of \cref{lem:CDPtoSimDP}]
		Let $\Mc^\Ac_x (1^\pk,y)= f_\pk^\Ac(x,y)$ and $\Mc^\Bc_y (1^\pk,x)= f_\pk^\Bc(x,y)$. Fix $\p \in \ABc$. Since $\pi$ is $\eps$-\CDP,  it is clear that $\Mc^\p_b$ is $\eps$-\CDP mechanism for every $b \in \zo$. From \cite{MironovPRV09}, for every $ b\in \zo$ there exists distributions ensembles $\set{D_\pk^{\Pc,b,0}}_{\pk\in\N},\set{D_\pk^{\Pc,b,1}}_{\pk\in\N}$, such that
	\begin{enumerate}
		\item for every $\pk \in \N$: $D_\pk^{\Pc,b,0}\rindist_{\eps} D_\pk^{\Pc,b,1}$, and
		
		\item for every  $c \in \zo$: 	 $\set{D_\pk^{\Pc,b,c}}_{\pk\in\N} \cindist \set{\Mc^\Pc_b(1^\pk,c)}_{\pk\in\N}$\label{item:compDP:sim}
	\end{enumerate}
	
	Consider the functionality   ensemble $\set{\tf_\pk}_{\pk\in\N}$ defined by \\$\tf_\pk(x,y) = (\tf^\Ac_\pk(x,y),\tf^\Bc_\pk(x,y))$  outputting   a random sample from  $(D_\pk^{\Ac,x,y},D_\pk^{\Bc,y,x})$. By definition, for every $\p\in \ABc$ it holds that
	$$\set{\tf^\Pc_\pk(x,y)}_{\pk\in\N} \cindist \set{f^\Pc_\pk(x,y)}_{\pk\in\N}.$$
	
	Thus, $\tf$ realizes the  $\eps$-\CDP functionality of $f$.
\end{proof}

\paragraph{Proving \cref{lem:DPXORfunc:C}.}
The proof  immediately follows by the next claim.

\begin{claim}\label{claim:NCDP:agreement}
	Let $f= \set{f_\pk = (f^\Ac_\pk,f^\Bc_\pk)}_{\pk\in\N}$ be a functionality ensemble that is $\eps$-\NCDP,  and has perfect agreement. Then $f$ is $\eps$-\NCDP with respect to a $\eps$-\DP functionality ensemble $\set{\tf_\pk = (\tf^\Ac_\pk,\tf^\Bc_\pk)}_{\pk\in\N}$  that satisfies  that for every  $x,y\in \zn$:
	
	\begin{align}\label{eq:NCDP:agreement}
	\set{(v^\Ac,v^\Bc_1)_{ (v^\Ac,v^\Bc) \gets f_\pk(x,y)}}_{\pk\in\N} \cindist \set{(v^\Ac,v^\Bc_1)_{ (v^\Ac,v^\Bc) \gets \tf_\pk(x,y)}}_{\pk\in\N},
	\end{align}
	and the same holds for or the view of $\Bc$.
	
\end{claim}
That is, the above claim states that if a functionality is $\eps$-\NCDP and has perfect agreement, then the view of each party is indistinguishable from the view in an $\eps$-\DP functionality, even when adding the output of the other party.

\begin{proof}
	The straightforward proof replaces  an arbitrary  functionality realizing the  $\eps$-\NCDP of $f$ with one that has (almost) perfect agreement.
	
	  By \cref{lem:CDPtoSimDP}, there exists functionality ensemble $\hf = \set{\hf_\pk = (\hf^\Ac_\pk,\hf^\Bc_\pk)}_{\pk\in\N}$ that realizes the $\eps$-\CDP of $f$. We show  there exists a functionality ensemble $\set{\tf_\pk = (\tf^\Ac_\pk,\tf^\Bc_\pk)}_{\pk\in\N}$ such that
	
	\begin{enumerate}
		\item $\tf^\Pc_\pk(x,y)$ and $\widehat{f}^\Pc_\pk(x,y)$ are the same  for every $x,y\in \zo$, $\pk\in \N$ and $\Pc \in \ABc$, and
		
		\item $\pr{\tf^\Ac_\pk(x,y)_1=\tf^\Bc_\pk(x,y)_1}\geq 1-\negl(\pk)$.
	\end{enumerate}
Namely, $\tf$ also realizes the 	$\eps$-\CDP of $f$ and has an almost perfect agreement. Since $f$ has perfect agreement, $\tf$ satisfies   \cref{eq:NCDP:agreement}.
		
In the rest of the proof we construct the desired $\tf$.  Since $f$ has  perfect agreement, and since $\hf^\Pc$ is computationally close to $f^\Pc$,  for every $x,y\in \zo$ it holds that
		\begin{align}
		\size{\pr{\hf^\Ac_\pk(x,y)_1=1}-\pr{\hf^\Bc_\pk(x,y)_1=1}} \leq \negl(\pk)
		\end{align}
		
		Therefore, for every $x,y\zo$ there  exists ensembles of Boolean random variables pairs  $\set{(\Rx^\Ac,\Rx^\Bc)}_{\pk\in\N}$ such that for any $\pk$:
		 \begin{align}
		 \Rx^\Pc \equiv \hf^\Pc_\pk(x,y)_1
		 \end{align}
		  and
		 \begin{align}
		 \pr{\Rx^\Ac = \Rx^\Bc} \geq 1-\negl(\pk)
		 \end{align}

	 For  $r\in \zo$, define  $\tf^\Pc_\pk(x,y,r)\eqdef\hf^\Pc_\pk(x,y)|_{\hf^\Pc_\pk(x,y)_=r}$, and let  $\tf^\p_\pk(x,y) =  \tf^\Pc_\pk(x,y,\Rx^\p)$.  By construction, the distributions $\tf^\p_\pk(x,y)$ and $\hf^\p_\pk(x,y)$ are the same and $\tf$ has almost perfect agreement.
\end{proof}
\begin{proof}[Proof of \cref{lem:DPXORfunc:C}.]
	The proof follows \cref{thm:DPXORtoOT:IT,claim:NCDP:agreement}, using a similar hybrid argument as in the proof of \cref{thm:main:C}.
\end{proof}

%% file: NonMonotone.tex
\newcommand{\sx}{\sigma_x}
\newcommand{\sy}{\sigma_y}

\subsection{Extension to Functions that are not Monotone under Relabeling}\label{sec:non-monotone}

We now extend our results to a large class of functions: functions that are not ``monotone under relabeling''.

\begin{definition}[Monotone under relabeling]\label{def:monotoneFunction}~
	A function $g\colon  \zn\times \zn \to \zo$ is {\sf monotone under relabeling} if there exists bijective functions $\sx,\sy:[2^n]\to \zn$ such that for every $x \in \zn$ and $i\leq j \in [2^n]$: 
	\begin{align*}
	g(x,\sy(i))\leq g(x,\sy(j)),
	\end{align*} 
		and, for every $y \in \zn$ and $i\leq j \in [2^n]$:
	\begin{align*}
		g(\sx(i),y)\leq g(\sx(j),y).
	\end{align*}  
	
\end{definition}

\begin{theorem}\label{thm:mainMonotone}
	There exists a constant $c>0$ such that the following holds for every  $n \in \N$.  Let $\eps, \beta$ be functions such that for every $\pk\in \N\colon\eps(\pk), \beta(\pk) \in [0,1]$, $1/2\geq \beta(\pk)\geq c\cdot n^2\cdot\eps(\pk)^2$ and $\beta(\pk) \geq 1/p(\pk)$ for some $p \in \poly$. Let $\pi$ be a \ppt two-party protocol  that is $\eps$-\CDP, and computes a function $g$ over $\zn\times\zn$ that is not monotone under relabeling, with worst-case correctness at least $\beta(\pk)$ and perfect agreement,  then there exists a non-uniform computationally secure OT.
\end{theorem}

We will show that every function that is not monotone under relabeling, has a copy of the XOR function that is ``embedded'' in it.

\begin{definition}[Embedded XOR]
A function $g \colon \zn\times \zn \to \zo$ has  {\sf embedded XOR} if there exists $x_0,x_1 \in \zn$ and $y_0,y_1 \in \zn$ such that for every $b,c \in \zo$, $g(x_b,y_c) = b \xor c$. 
\end{definition}

For example the Hamming distance function $\Ham(x,y)$ over $\zo^n \times \zo^n$ has an embedded XOR, by using  the inputs $x_b=b \circ 0^{n-1}$ and $y_c=c \circ 0^{n-1}$.

It is clear that a function that is monotone under relabeling does not have an embedded XOR. In the following we show the opposite direction:  every function $g$ that is not monotone under relabeling  has an embedded XOR. Moreover, we show that if $\pi$ is a $\eps$-\CDP protocol that computes  function  $g$ with worst-case correctness $\beta$, then there exists a $n\cdot\eps$-\CDP protocol $\tpi$ that compute XOR with the same correctness. \cref{thm:mainMonotone}  then follows by \cref{thm:DPXORtoOT:C}.

\begin{lemma}\label{lem:monToXor:corr}
	A  function that is not monotone under relabeling, has an embedded  XOR.
\end{lemma}

\begin{proof}
	Let $g\colon\zn\times \zn \to \zo$ be a  function that has no embedded XOR. We show that $g$ is monotone under relabeling.
	
	For  input $x \in \zn$, let $Z_x=\set{y \mid g(x,y)=0}$. We claim that for any $x_0$ and $x_1$ in $\zo$, it must hold that either $Z_{x_0} \subseteq Z_{x_1}$, or, $Z_{x_1} \subseteq Z_{x_0}$. Indeed, otherwise there is $y_0,y_1$ such that $y_0 \in Z_{x_0}\setminus Z_{x_1}$ and $y_1 \in Z_{x_1}\setminus Z_{x_0}$, and therefore, for $b,c \in \zo$, $g(x_b,y_c) = b \xor c$.
	
	Let $\sx:[2^n] \to \zn$ be a bijective function such that for every $i \leq j$, $\size{Z_{\sx(i)}} \geq \size{Z_{\sx(j)}}$. Then it must hold that $Z_{\sx(j)}\subseteq Z_{\sx(i)}$, and therefore for every $y \in \zn$, $g(\sx(i),y) \leq g(\sx(j),y)$.
	
	Repeating this argument to construct $\sy$ ends the proof. 
\end{proof}

\begin{lemma}\label{lem:monToXor:dp}
	Let $\eps$ be a function with $\eps(\pk)\in [0,1]$ and let $\pi = (\Ac,\Bc)$ be a  $\eps$-\CDP protocol. Then for every $x_0,x_1,y_0,y_1\in \zn$, the protocol $\tpi= (\tAc,\tBc)$ defined by $(\tAc(b),\tBc(c))(1^\pk) = (\Ac(x_b),\Bc(y_c))(1^\pk)$ is $(n\eps)$-\CDP.
\end{lemma}
\begin{proof}
	For $x,y \in \zn$ and $\pk \in \N$, let $\VBx$ be the view of $\Bc$ in a random execution of $(\Ac(x),\Bc(y))(1^\pk)$. Let $\Dc$ be a \nuppt. Since $\Pi$ is $\eps$-\CDP, for every   $x,x'  \in \zn$ with $\Ham(x,x') =1$ it  hold that

	\begin{align}
	\pr{\Dc(\VBx, 1^\pk) = 1}  \leq   e^{\eps(\pk)} \cdot \pr{\Dc(\VB_{x',y,\pk}, 1^\pk) = 1} + \negl(\pk)
	\end{align}
	
A simple  calculation (known as ``singleton privacy implies group privacy'') shows that for every   $x,x' 
\in  \zn$ with $\Ham(x,x') =d$:
	$$\pr{\Dc(\VBx, 1^\pk) = 1}  \leq   e^{d\cdot\eps(\pk)} \cdot \pr{\Dc(\VB_{x',y,\pk}, 1^\pk) = 1} + \negl(\pk).$$	
	
	The proof for $\Ac$'s privacy thus followed by the fact that for any $x_0,x_1 \in \zn$, the Hamming distance $\Ham(x,x')$ is at most $n$. The proof for the privacy of $\Bc$ follows  similar lines.
\end{proof}

We remark that the loss incurred in \cref{lem:monToXor:dp} is sometimes unnecessary. For example, in the XOR-embedding of the Hamming distance function that we considered above, the distance between $x_0$ and $x_1$ (and also between $y_0$ and $y_1$) is only one, and therefore, no losses in privacy are incurred in this case, and \cref{thm:mainMonotone} holds for $g(x,y)=\Ham(x,y)$ without the loss of $n^2$ factor, in the privacy.

\paragraph{Proving  \cref{thm:mainMonotone}.}
We now ready to prove \cref{thm:mainMonotone}.
\begin{proof}[Proof of \cref{thm:mainMonotone}.]
Let $\pi$ be a protocol that satisfies the requirements of \cref{thm:mainMonotone} \wrt a  function $g$ that is not monotone under relabeling. By \cref{lem:monToXor:corr}, there exist $x_0,x_1,y_0,y_1 \in \zn$ such that for every $b,c \in \zo$, $g(x_b,y_c) = b \xor c$. Therefore, the protocol defined by $(\tAc(b),\tBc(c))(1^\pk) \eqdef (\Ac(x_b),\Bc(y_c))(1^\pk)$ computes the XOR functionality with average correctness at least $\beta(\pk)$, and by \cref{lem:monToXor:dp} this protocol is $(n\eps)$-\CDP. Thus,  the theorem follows by \cref{thm:DPXORtoOT:C}. 
\end{proof}

%% file: Appendix.tex
\section{Missing Proofs}\label{sec:appendix}

\paragraph{Proving \cref{Prop:CondSD}}

\begin{proposition}[\cref{Prop:CondSD}, recited]\label{Prop:CondSD:app}
	\PropCondSD
\end{proposition}
\begin{proof}
	In the following we show that for every $b\in \zo$,\\ $\SD(\tX|_{\set{(\tX,\tY) \in E_b}},X|_{\set{(X,Y) \in E_b}})\leq 2\eps/\mu$. The proof then follows using the triangle inequality.
	
	Note that, by data processing,\\  $\SD(\tX|_{\set{(\tX,\tY) \in E_b}},X|_{\set{(X,Y) \in E_b}})\leq \SD((\tX,\tY)|_{\set{(\tX,\tY) \in E_b}},(X,Y)|_{\set{(X,Y) \in E_b}})$
	
	For every set $\cA \subseteq \cX\times \cY$, and $b \in \zo$, we want to bound\\ $\pr{(X,Y) \in \cA \mid (X,Y) \in E_b}-\pr{(\tX,\tY) \in \cA \mid (\tX,\tY) \in E_b}$.
	
	It holds that,
	\begin{align}\label{eq:CondSD}
	&\pr{(X,Y) \in \cA \mid (X,Y) \in E_b}-\pr{(\tX,\tY) \in \cA \mid (\tX,\tY) \in E_b} \\\nonumber
	& = \frac{\pr{(X,Y) \in \cA \cap E_b}}{\pr{(X,Y)\in E_b}}-\frac{\pr{(\tX,\tY) \in \cA \cap E_b}}{\pr{(\tX,\tY)\in E_b}}\\\nonumber
	&\leq \frac{\pr{(X,Y) \in \cA \cap E_b}}{\pr{(X,Y)\in E_b}}-\frac{\pr{(X,Y) \in \cA \cap E_b}-\eps}{\pr{(X,Y)\in E_b}+\eps}\\\nonumber
	&=\frac{\eps\cdot \pr{(X,Y) \in \cA \cap E_b}+ \eps \cdot \pr{(X,Y)\in E_b}}{\pr{(X,Y)\in E_b}(\pr{(X,Y)\in E_b}+\eps)}\\\nonumber
	&\leq \frac{2\eps \cdot\pr{(X,Y)\in E_b}}{(\pr{(X,Y)\in E_b})^2}\\\nonumber
	&\leq\frac{2\eps}{\mu}\\\nonumber
	\end{align}
	Where the last equality follows because $\frac{A}{B} -\frac{A-\eps}{B+\eps} = \frac{\eps(A+B)}{B(B+\eps)}$.
	Since \cref{eq:CondSD} holds for every set  $\cA  \subseteq \cX\times \cY$, we get that\\ $(X,Y)|_{\set{(X,Y) \in E_b}} \sindist_{2\eps/\mu}(X,Y)|_{\set{(\tX,\tY) \in E_b}}$, for every $b \in \zo$.
\end{proof}

\paragraph{Proving \cref{clm:spcl->WBSC}}
	
	\begin{proposition}[\cref{clm:spcl->WBSC}, recited]
		\PropWBSC
	\end{proposition}
	\begin{proof} 
		The correctness and the receiver security properties hold from the definition.
		
		For sender security,  first notice that for every  $b_\Bc$, we get from the symmetry of \SWBSC that:
		
		\begin{align*}
		\pr{\OB=b_\Bc\mid \OA=\OB} &= \frac{\pr{ \OA=\OB\mid\OB=b_\Bc} \pr{\OB=b_\Bc}}{\pr{\OA=\OB}}=1/2,
		\end{align*}
		and
		\begin{align*}
		\pr{\OB= b_\Bc\mid \OA\neq \OB} &= \frac{\pr{ \OA\neq \OB\mid\OB=b_\Bc} \pr{\OB=b_\Bc}}{\pr{\OA\neq \OB}}=1/2
		\end{align*}
		
		
		Now, assume for contradiction that, for some $b \in \zo$, a distinguisher $\Dc_b$  breaks the sender security in the \WBSC definition. That is,\\ $\pr{\Dc_b(\VB)=1|\OB=b,\OA=0}-\pr{\Dc_b(\VB)=1|\OB=b,\OA=1}> 2p$. Then, we can construct a distinguisher that breaks the specialized sender security:
		Let $\Dc_{1-b}$ be an algorithm such that $\pr{\Dc_{1-b}(\VB)=1|\OB=1-b,\OA=0}-\pr{\Dc_{1-b}(\VB)=1|\OB=1-b,\OA=1}\geq 0$, and consider the following algorithm:
		
		\begin{algorithm}[$\Dc'$]
			\item [Input:] $(v,y) \in \Supp(\VB,\OB)$. 
			\item [Operation:] Output  $\Dc_y(v)$.~\\
		\end{algorithm}
		It holds that:
		\begin{align*}
		&p \geq \pr{\Dc'(\VB,\OB)=1|\OA=\OB}-\pr{\Dc'(\VB,\OB)=1|\OA\neq\OB}\\
		&\begin{aligned}
		=&1/2\cdot [\pr{\Dc'(\VB,0)=1|\OA=\OB, \OB=0}-\pr{\Dc'(\VB,0)=1|\OA\neq\OB,\OB=0}]\\
		&+1/2\cdot [\pr{\Dc'(\VB,1)=1|\OA=\OB, \OB=1}-\pr{\Dc'(\VB,1)=1|\OA\neq\OB, \OB=1}]
		\end{aligned}\\
		&\begin{aligned}
		=&1/2\cdot [\pr{\Dc_0(\VB)=1|\OA=\OB, \OB=0}-\pr{\Dc_0(\VB)=1|\OA\neq\OB,\OB=0}]\\
		&+1/2\cdot [\pr{\Dc_1(\VB)=1|\OA=\OB, \OB=1}-\pr{\Dc_1(\VB)=1|\OA\neq\OB, \OB=1}]
		\end{aligned}\\
		&\begin{aligned}
		=&1/2\cdot [\pr{\Dc_0(\VB)=1|\OA=0, \OB=0}-\pr{\Dc_0(\VB)=1|\OA=1,\OB=0}]\\
		&+1/2\cdot [\pr{\Dc_1(\VB)=1|\OA=0, \OB=1}-\pr{\Dc_1(\VB)=1|\OA=1, \OB=1}]
		\end{aligned}\\
		&> p.
		\end{align*}
	\end{proof}

\paragraph{Proving \cref{Prop:simple}}
\begin{proposition}[\cref{Prop:simple}, recited]\label{Prop:simple:app}
	\PropSimple
\end{proposition}

\begin{proof}
	We start with the lower bound,
	\begin{align*}
\frac{(1/2+b)^\ell}{(1/2+b)^\ell+(1/2-b)^\ell}&=\frac{(1+2b)^\ell}{(1+2b)^\ell+(1-2b)^\ell}\\
&=
\frac{\sum_{i=0}^{\ell}\binom{\ell}{i}(2b)^i}{2\sum_{i=0}^{\floor{\ell/2}}\binom{\ell}{2i}(2b)^{2i}}.\\
&=\frac{\sum_{i=0}^{\floor{\ell/2}}\binom{\ell}{2i}(2b)^{2i}+\sum_{i=0}^{\floor{(\ell-1)/2}}\binom{\ell}{2i+1}(2b)^{2i+1}}{2\sum_{i=0}^{\floor{\ell/2}}\binom{\ell}{2i}(2b)^{2i}}\\
&=1/2+\frac{\sum_{i=0}^{\floor{(\ell-1)/2}}\binom{\ell}{2i+1}(2b)^{2i+1}}{2\sum_{i=0}^{\floor{\ell/2}}\binom{\ell}{2i}(2b)^{2i}}\\
&= 1/2+\frac{2b \ell}{2\sum_{i=0}^{\floor{\ell/2}}\binom{\ell}{2i}(2b)^{2i}} + \frac{\sum_{i=1}^{\floor{(\ell-1)/2}}\binom{\ell}{2i+1}(2b)^{2i+1}}{2\sum_{i=0}^{\floor{\ell/2}}\binom{\ell}{2i}(2b)^{2i}}\\
&= 1/2+\frac{b \ell}{1+\sum_{i=1}^{\floor{\ell/2}}\binom{\ell}{2i}(2b)^{2i}} + \frac{\sum_{i=1}^{\floor{(\ell-1)/2}}\binom{\ell}{2i+1}(2b)^{2i+1}}{2+2\sum_{i=1}^{\floor{\ell/2}}\binom{\ell}{2i}(2b)^{2i}}\\
&\geq 1/2+\frac{b \ell}{\sum_{i=0}^{\floor{\ell/2}}\binom{\ell}{2i}(2b)^{2i}}\\
&=1/2+\frac{b \ell}{1+\sum_{i=1}^{\floor{\ell/2}}\binom{\ell}{2i}(2b)^{2i}}\\
&\geq 1/2+\frac{b \ell}{1+\sum_{i=1}^{\floor{\ell/2}}(2b \ell)^{2i}}\\
&\geq 1/2+{b \ell}/{2}=\half(1+b\ell)
\end{align*}
Finally a similar calculation yields the following upper bound,
	\begin{align*}	
\frac{(1/2+b)^\ell}{(1/2+b)^\ell+(1/2-b)^\ell} &\leq 1/2+b \ell+ \frac{\sum_{i=1}^{\floor{(\ell-1)/2}}\binom{\ell}{2i+1}(2b)^{2i+1}}{2+2\sum_{i=1}^{\floor{\ell/2}}\binom{\ell}{2i}(2b)^{2i}}\\
&\leq 1/2+b \ell+ \frac{\sum_{i=1}^{\floor{(\ell-1)/2}}(2b \ell)^{2i+1}}{2}\\
&\leq\half(1+3b\ell)
\end{align*}
\end{proof}